\documentclass[12pt]{article}
\usepackage{amsmath, amssymb, amsfonts, amsthm}
\usepackage{color}
\title{Second order semiclassics with self-generated
magnetic fields}
\author{L\'aszl\'o Erd\H os
 \thanks{Partially supported by SFB-TR12 of
the German Science Foundation. {\text lerdos@math.lmu.de} }
\\Institute of Mathematics, University of Munich \\
Theresienstr. 39, D-80333 Munich, Germany \\
S\o ren Fournais \thanks{Work partially supported by the Lundbeck
  Foundation, the Danish Natural Science Research Council and the European 
Research Council under the
 European Community's Seventh Framework Program (FP7/2007--2013)/ERC grant
 agreement  202859.
{\text fournais@imf.au.dk}} \\ Department of Mathematical Sciences, Aarhus University\\
 Ny Munkegade 118, DK-8000 Aarhus, Denmark
\\ and \\
Jan Philip Solovej \thanks{Work partially supported
   by the Danish Natural Science Research Council and by a Mercator
   Guest Professorship from the German Science Foundation. {\text
solovej@math.ku.dk}}
\\ Department of Mathematics, University of Copenhagen\\
Universitetsparken 5, DK-2100 Copenhagen,
Denmark}

\date{May 29, 2011}

\newtheorem{theorem}{Theorem}[section]
\newtheorem{proposition}[theorem]{Proposition}
\newtheorem{corollary}[theorem]{Corollary}
\newtheorem{lemma}[theorem]{Lemma}

\newtheorem{remark}[theorem]{Remark}
\numberwithin{equation}{section}

\newcommand{\rd}{{\rm d}}
\newcommand{\be}{\begin{equation}}
\newcommand{\ee}{\end{equation}}
\newcommand{\bey}{\begin{eqnarray}}
\newcommand{\eey}{\end{eqnarray}}
\newcommand{\beys}{\begin{eqnarray*}}
\newcommand{\eeys}{\end{eqnarray*}}

\DeclareMathOperator{\supp}{supp}

\newcommand{\bZ}{{\mathbb Z}}

\newcommand{\bsigma}{\mbox{\boldmath $\sigma$}}

\renewcommand{\iint}{\int \!\! \int}

\newcommand{\bR}{{\mathbb R}}
\newcommand{\bC}{{\mathbb C}}

\newcommand{\bN}{{\bf N}}

\newcommand{\ep}{\varepsilon}

\newcommand{\ov}{\overline}

\newcommand{\e}{\varepsilon}

\newcommand{\Tr}{{\rm Tr\;}}
\newcommand{\tr}{{\rm Tr\;}}

\newcommand{\wh}{\widehat}
\newcommand{\wt}{\widetilde}

\newcommand{\cE}{{\cal E}}

\newcommand{\al}{\alpha}

\newcommand{\pt}{\partial}

\newcommand{\om}{\omega}

\newcommand{\non}{\nonumber}

\begin{document}
\maketitle

\begin{abstract}
We consider the semiclassical asymptotics of
the sum of negative eigenvalues of the three-dimensional Pauli operator
with an external potential and  a self-generated magnetic field $B$.
We also add
the field energy $\beta \int B^2$ and we
minimize over all magnetic fields. 
The parameter $\beta$ 
 effectively determines the strength of the field.
We consider the weak field regime with  $\beta  h^{2}\ge \mbox{const}>0$,
where $h$ is the semiclassical parameter.
For smooth potentials we prove that the semiclassical
asymptotics of the total energy is  given by the non-magnetic Weyl term to leading order
with an error bound that is smaller by a factor $h^{1+\e}$,
i.e. the subleading term vanishes.
However, for potentials with a Coulomb singularity the subleading term
does not vanish due to the non-semiclassical effect of the singularity.
Combined with a multiscale technique, this refined estimate is used 
in the companion paper \cite{EFS3} 
to prove the second order Scott correction 
to the ground state energy of large atoms and molecules.

\end{abstract}



\medskip\noindent
{\it Running title:} Semiclassics with self-generated field.

\section{Introduction}\label{sec:intro}

\bigskip

An important problem in semiclassical spectral analysis is to determine  the sum of
the negative eigenvalues of a Schr\"odinger operator
$-h^2\Delta-V(x)$ on $L^2(\bR^d)$, i.e.,
$$
\tr(-h^2\Delta-V(x))_-.
$$
We will use the convention that $x_-= (x)_- =\min\{ x, 0\}$
when $x$ is either a real number or a self-adjoint operator.
It is well known that  under appropriate integrability conditions on  $V$ 
the leading order term is given by the Weyl asymptotics
\be
   \tr(-h^2\Delta-V(x))_- = (2\pi h)^{-d}\iint_{\bR^d\times\bR^d} (p^2-V(x))_-\rd x\rd p 
  + o(h^{-d}), \qquad h\to 0,
\label{weyl}
\ee
and for smooth potentials the error term can be improved to  $O(h^{-d+2})$
(under some non-criticality assumption) by using pseudo-differential calculus.
In other words,  the
subleading term in the semiclassical expansion in powers of $h$ vanishes.
We remark that with more elementary methods and under less regularity assumptions on $V$, 
for a local version of this problem
$$
  \tr\big[\psi (-h^2\Delta-V(x))\psi\big]_-,   \qquad \psi\in C_0^\infty(\bR^d),
$$
the error term has been shown to be  $O(h^{-d+6/5})$ in Theorem 12 of \cite{SS}
(we will recall it in Theorem~\ref{thm:SS}).
This bound was extended to the relativistic case in \cite{SSS}.

A generalization of this problem is to consider not only a potential
$V$ but also an exterior magnetic field $B=\nabla\times A$ generated by a vector potential
$A$. The corresponding magnetic Schr\"odinger operator is 
$
(-i\nabla + A)^2-V(x).
$
A further generalization is to consider the particles as having 
spin$-\frac{1}{2}$
and introduce the magnetic Pauli operator 
$[\bsigma\cdot (-i\nabla + A)]^2-V$, where in  $d=3$
dimensions $\bsigma=(\sigma_1, \sigma_2, \sigma_3)$ denotes the vector
of Pauli matrices.  For simplicity we will consider 
the $d=3$ dimensional case only and
we denote both the Schr\"odinger
operator $(-ih\nabla + A)^2$ and the Pauli operator 
$[\bsigma\cdot (-ih\nabla + A)]^2$ by $T_h(A)$.
Our analysis will be carried out in the more
complicated case of the Pauli operator, analogous but easier
results can be proved for the Schr\"odinger case as well.

Much work has gone into understanding the semiclassical asymptotics of
the sum of negative eigenvalues, i.e., the asymptotics for small $h>0$
of
$$
\tr( T_h(A)-V(x))_-. \quad
$$
It is well known that under appropriate conditions on $A$ and $V$ the
leading behavior as $h$ tends to zero is given by the Weyl asymptotics 
\eqref{weyl} and note that  the limit behavior is non-magnetic, i.e.
fixed magnetic fields do not influence the leading order semiclassics.
If the magnetic field is rescaled, $B\to \mu B$, and the coupling constant $\mu$
increases along with the $h\to 0$ semiclassical limit at least as $\mu \gtrsim h^{-1}$,
then magnetic fields become relevant even in the leading term.
Most work in this direction has been carried out with 
a homogeneous magnetic field \cite{Sob, LSY1, LSY2} with some generalization
to an inhomogeneous one \cite{ES, I} but always subject to regularity conditions
on the field.

In this paper we will address a related and equally important issue,
namely the case when the magnetic field is not a fixed external field,
but the self-generated classical magnetic field generated by the
particles themselves. 
The vector potential $A$ will be optimized to
minimize the total energy consisting of the sum of negative eigenvalues
(corresponding to the ground state energy of non-interacting fermions)
and the field energy
$$
\int B^2 =\int |\nabla\times A|^2
$$
(we use the convention that unspecified integrals are always
on $\bR^3$ w.r.t. the Lebesgue measure). 
The problem we consider is thus to determine the
energy
\begin{equation}\label{eq:energy}
E(\beta,h,V)=\inf_A\left[\tr(T_h(A)-V)_-+\beta\int |\nabla\times A|^2\right]
\end{equation}
for $\beta,
h>0$, where the infimum runs over all vector fields $A\in
H^1(\bR^3;\bR^3)$; in fact minimizing only for all $A\in
C_0^\infty(\bR^3;\bR^3)$ gives the same infimum.
Alternatively, in addition to  $A\in
H^1(\bR^3;\bR^3)$, one could  impose a gauge
fixing condition, e.g., $\nabla\cdot A=0$,
  see more
details in Appendix A of \cite{EFS1}.

Here $\beta$ is an additional parameter setting the strength of
 the coupling of
the particles to the field. Formally $\beta=\infty$ corresponds to the non-magnetic case; smaller
$\beta$ means that a larger optimizing magnetic field is expected.
In a given physical system the values of 
$h$ and $\beta$ are given, but as is standard in semiclassical analysis we
leave them as free parameters.

The Euler-Lagrange  equation corresponding to the variational problem
\eqref{eq:energy} above is 
\be
\beta\,\nabla\times B=J_A,
\label{maxwell}
\ee
where $J_A$ is the current of the Fermi gas, which in the
Schr\"odinger case is
$$
J_A(x)=-\text{Re}\,\left[(-ih\nabla+A)1_{(-\infty,0]}(T_h(A)-V)\right](x,x)
$$ 
and in the Pauli case is
$$
J_A(x)=-\text{Re}\,\left[\tr_{\bC^2}\left(\bsigma (\bsigma\cdot(-ih\nabla+A))1_{(-\infty,0]}(T_h(A)-V)\right)\right](x,x).
$$
In other words the Euler-Lagrange equations are the non-linear coupled
Maxwell-Schr\"odinger or Maxwell-Pauli equations. 

Semiclassical results with magnetic fields mentioned above 
assume that the field is regular. However, if the magnetic field
arises as self-generated and thus determined internally via
a variational principle, the sufficient regularity is not a-priori given.

The first semiclassical result for the local problem with a self-generated magnetic field was
presented in \cite{ES3}, where the leading order asymptotics  was shown 
to be given by the non-magnetic Weyl term
in the weak field regime, $\beta h^2\ge \mbox{const}>0$.
The error term was by a factor of order $h^{1/2}$ smaller than the leading term
(see Theorem 1.3 of \cite{ES} for the lower bound; the matching upper
bound was not explicitly stated in \cite{ES} but it clearly  follows
by choosing $A\equiv 0$).

The main result of this paper
is a substantial improvement of the error term to a factor of order $h^{1+\ep}$.
This result can also be interpreted as showing that the
subleading term in the semiclassical expansion in powers of $h$ vanishes.

\begin{theorem}\label{thm:secondorder} Let $V\in C_0^\infty(\bR^3)$.
There exist a universal
 constant $\e>0$ such that for any fixed $\kappa_0>0$ we have 
\be\label{eq:secondorder}
  \lim_{h\to0\; , \; \beta h^2\ge \kappa_0}  h^{2-\e}\Big| E(\beta, h, V)- 
  2(2\pi h)^{-3}\iint \big[ p^2 - V(q)\big]_- \rd p \rd q \Big| =0
\ee
for the Pauli problem. The same result (without the prefactor 2) holds
for the Schr\"odinger problem as well.
\end{theorem}

In fact, we  can  replace the infimum over all $A$ in the definition of
$ E(\beta, h, V)$ with a good apriori bound on the field energy:

\begin{corollary}\label{cor:scott}  There exist universal constants $\e>0$ and $\kappa_0>0$ 
such that
if $A_h$ is a sequence of vector potentials satisfying
\be
  \lim_{h\to0} h^{2-\e}\beta \int |\nabla\times A_h|^2 =0,
\label{weaker}
\ee
then for any fixed smooth, compactly supported potential $V$, we have
\be\label{weakclaim}
  \lim_{h\to0\; , \; \beta h^2\ge \kappa_0}  h^{2-\e}\Big| 
\tr(T_h(A_h)-V)_-+\beta\int |\nabla\times A_h|^2- 
  2(2\pi h)^{-3}\iint \big[ p^2 - V(q)\big]_- \rd p \rd q \Big| =0
\ee
for the Pauli problem.  The same result (without the prefactor 2) holds
for the Schr\"odinger problem as well.
\end{corollary}

In \cite{EFS1} we have also analyzed the semiclassical behavior of $E(\beta,h,V)$
in other regimes of the parameter $\beta$. The main motivation for
a  precise second order asymptotics in the specific   regime $\beta h^2\ge \mbox{const}>0$ is that
Theorem~\ref{thm:secondorder} is the main technical input
for the proof of the Scott correction term 
of the ground state energy of large atoms and molecules
in the limit when the nuclear charge $Z$ tends to infinity.
We refer to
 Section 2.3 of \cite{EFS1} for the precise statement on
the Scott correction with a self-generated magnetic field and for an explanation 
on its connection with semiclassical asymptotics.
 The actual proof of the Scott correction is given in a separate paper
\cite{EFS3}.

\section{Localized models}\label{sec:loc}

The semiclassical asymptotics is essentially a local issue. 
If the potential $V$ has sufficient decay there will be no semiclassical
contribution from infinity. It is an interesting question whether the 
decay needed for stability will ensure that there is no
semiclassical contribution from inifinity, but at this point we do not have
a general proof of this. 
To separate the local semiclassics from the issue at infinity we 
introduce local versions of the energy. 

Let $\psi$ be a smooth function with $\mbox{supp}\, \psi\subset
B(1)$,  where in general $B(r)$ denotes the ball of radius $r$ centered at the origin.
We will always assume that $\|\psi\|_\infty\le 1$ and
$\psi$ is identically 1 in a neighborhood of the origin.
Denote by $\psi_r(x)=\psi(x/r)$. For any $r>0$ define
\be
  E_{r}(\beta, h, V): = 
\inf_A \Big[\tr \big[ \psi_r(T_h(A)-V)\psi_r\big]_- + 
\beta \int_{\bR^3} |\nabla \times A|^2\Big],
\label{Ebhpr}
\ee
where we minimize over all $A\in H^1(\bR^3 ;\bR^3)$, or, equivalently,
over all  $A\in C_0^\infty(\bR^3;\bR^3)$.
Alternatively, we may again restrict to vector potentials
$A\in  H^1(\bR^3 ;\bR^3)$ with $\nabla\cdot
A=0$  in \eqref{Ebhpr} without changing the value of the infimum.
Without loss of generality, we
can always assume that $V$ is supported in $B(r)$.

Here we have localized the particles but not the fields. We can also
do both. 
Let us first note that if $\nabla\cdot A=0$ then 
\begin{equation}\label{eq:times=otimes}
\int_{\bR^3}|\nabla\times A|^2=\int_{\bR^3}|\nabla\otimes A|^2,
\end{equation}
where the integrand of the right contains all first derivatives, i.e.
$|\nabla \otimes A|^2 = \sum_{ij=1}^3 |\partial_i A_j|^2$. This
identity is easily seen using the Fourier transform. As the local
version of the field energy we will use the localization of the
integral on the right. For any $R\ge r>0$ we thus  set
\be 
E_{r,R}(\beta, h, V): = \inf_A \Big[ \tr \big[
\psi_r(T_h(A)-V)\psi_r\big]_- + 
\beta \int_{B(R)} |\nabla\otimes A|^2\Big],
\label{Ebh}
\ee 
where we minimize over all $A\in H^1(B(R) ;\bR^3)$,
or, equivalently, over all  $A\in C_0^\infty(\bR^3;\bR^3)$.
The localized field energy is not fully gauge
invariant, we therefore cannot restrict attention to
divergence
free vector potentials. We are however free to add a constant vector
to $A$.  
Notice that the integral of $|\nabla\otimes A|^2$ is taken on a
larger ball $B(R)$ than $B(r)$ which contains 
the support of $\psi_r$. 
By virtue of gauge invariance and 
\eqref{eq:times=otimes}
we have 
$$
E_r=E_{r,\infty}:=\lim_{R\to\infty}E_{r,R}.
$$ 
By a simple rescaling we have
$$
E_{r,R}(\beta, h, V)=r^{-2}E_{1,R/r}(r\beta, h, V_r),
$$
where $V_r(x)=r^{2}V(xr)$.  It is thus enough to analyse the case
$r=1$.

 The definition
\eqref{Ebhpr} is physically somewhat better motivated than \eqref{Ebh}
since it contains the energy of the magnetic field in the whole space
and thus gives the correct magnetic interaction between the
particles.  The form \eqref{Ebh} is however more useful if we want to
localize all parts of the energy. 

Another version of the localized energy would be 
\be
  E'_{r,R}(\beta, h, V): = \inf_A 
\Big[ \tr \psi_r^2 \big[ T_h(A)-V\big]_- + \beta \int_{B(R)} |\nabla\otimes A|^2\Big] .
\label{Ebhprpr}
\ee
This has the disadvantage of being more complicated to calculate as it
requires knowledge of the operator $T_h(A)-V$ on the whole space,
e.g. it is not enough to know $V$ only on $B(r)$ or $A$ only on
$B(R)$. It has the advantage of not causing localization errors. 
Note that since $\psi$ is assumed to be equal to $1$ in a neighborhood of
$0$  we may identify 
$$
E(\beta, h, V)=
\lim_{r\to\infty}E'_{r,\infty}(\beta, h, V)=:E'_{\infty,\infty}(\beta, h, V).
$$
Finally we could also have defined the localized energy by introducing
Dirichlet boundary conditions on the boundary of $B(r)$, i.e., 
\be
  E^{\rm D}_{r,R}(\beta, h, V): = 
\inf_A 
\Big[ \tr\big[ (T_h(A)-V)_{B(r),\rm D}\big]_- 
+ \beta \int_{B(R)} |\nabla\otimes A|^2\Big].
\label{EbhprD}
\ee
We note that the leading order local semiclassical result
(Theorem 1.3 in \cite{ES}) was  stated for Dirichlet boundary conditions.

In the main part of this
work we will mainly use the global energy $E$ in \eqref{eq:energy} or 
the localized energy $E_{r,R}$ in 
\eqref{Ebh} but in
the next section we will compare the local versions
$E_r, E_{r,R}$, and $E'_{r,R}$. We will not discuss Dirichlet boundary
conditions.

\subsection{Comparison between the different localized energies}

We first compare the energies $E_{r,R}$ and  $E_{r}$.

\begin{lemma}\label{lm:version1}
There exists a universal constant $C_0$ such that
for all $0<r<R/2$ we have 
\be
   E_{r,R}(\beta, h, V)\le E_{r}(\beta, h, V) \le  E_{r,R}\big( ( 1+C_0(r/R)^3)\beta, h, V\big).
\label{EE}
\ee
This result holds for both the Pauli and Schr\"odinger operators.
\end{lemma}

{\it Proof.}  
The first inequality in \eqref{EE} is trivial since $E_r=E_{r,\infty}$
and $E_{r,R}$ is clearly an increasing function of $R$.

To prove the second inequality we start from an approximate minimizing
$A\in C_0^\infty$ for the energy $E_{r,R}$ on the right hand side of \eqref{EE}.
 By subtracting a constant vector from
$A$ we may assume that $A$ has average 0 on the sphere $|x|=r$.  We
may, moreover, assume that $A$ minimizes the Dirichlet integral
$\int_{r<|x|<R}|\nabla\otimes A(x)|^2\rd x$ given the boundary value of
$A$ on
$|x|=r$. If not, we could improve the trial energy of $A$ by replacing
it on the set $r<|x|<R$ with the vector field that agrees with $A$ for
$|x|\leq r$ and minimizes $\int_{r<|x|<R}|\nabla\otimes A(x)|^2\rd x$.  As a
trial vector field for the $E_r$ we choose the field $A'$ defined on
all of $\bR^3$ that
agrees with $A$ for $|x|\leq r$ and minimizes the integral
$\int_{r<|x|}|\nabla\otimes A'(x)|^2\rd x$.  The vector fields $A$ and $A'$ in
the region $|x|\geq r$ can be expressed in terms of the common boundary value
of $A$ on $|x|=r$. {F}rom the lemma below we see that there exists a universal constant $C_0$ such that 
$$
\int_{r<|x|}|\nabla\otimes A'(x)|^2\rd x\leq (1+C_0(r/R)^3)\int _{r<|x|<R}|\nabla\otimes A(x)|^2\rd x.
$$
Note that the constructed $A'$ is not necessarily divergence free.
Since $A$ and $A'$ agree on $|x|\leq r$ we obviously have 
$$
\int_{\bR^3}|\nabla\times A'(x)|^2\rd x
\leq \int_{\bR^3}|\nabla\otimes A'(x)|^2\rd x\leq (1+C_0(r/R)^3)\int _{|x|<R}|\nabla\otimes A(x)|^2\rd x
$$
and 
\be\label{ggg}
\tr \big[ \psi_r(T_h(A)-V)\psi_r\big]_- =  
\tr \big[ \psi_r(T_h(A')-V)\psi_r\big]_- .
\ee
The second inequality of \eqref{EE} follows from the last two observations. 
\qed

\medskip

We now give the  simple estimate needed in the previous proof. 
\begin{lemma} Assume $0<2r<R$ and let $S(r)=\{|x|=r\}$. Given $g\in
  L^2(S(r))$ with average 0.
Let $f_1\in H^1(B(R)\setminus B(r))$ and $f_2\in H^1(\bR^3\setminus B(r))$
satisfy the boundary conditions 
$f_{1|S(r)}=f_{2|S(r)}=g_{|S(r)}$ and minimize the 
respective Dirichlet integrals 
$$
\int_{B(R)\setminus B(r)}|\nabla f_1|^2,\quad \int_{\bR^3\setminus
  B(r)}|\nabla f_2|^2.
$$
Then 
$$
\int_{\bR^3\setminus B(r)}|\nabla f_2|^2\leq
(1+C_0(r/R)^3)\int_{B(R)\setminus B(r)}|\nabla f_1|^2,
$$
for some universal constant $C_0$.
\end{lemma}
\begin{proof}
By rescaling we may assume that $r=1$.
We know that $f_1$ is harmonic for $r<|x|<R$ and
satisfies Neumann boundary conditions on the
boundary $|x|=R$ and that $f_2$ is harmonic on $|x|>1$ 
and tends to 0 at infinity. We can write
$$
\int_{B(R)\setminus B(1)}|\nabla f_1|^2=
-\int_{S(1)}\overline{g}\partial_rf_1,\quad
\int_{\bR^3\setminus
  B(1)}|\nabla f_2|^2= -\int_{S(1)}\overline{g}\partial_rf_2.
$$
We expand $g$ in angular momentum eigenfunctions
$g=\sum_{\ell=0}^\infty\sum_{ m=-\ell}^\ell\alpha_{\ell m}Y_{\ell m}$.
Since $g$ has average 0, we have $\alpha_{00}=0$. 
We also expand the harmonic functions
$$
f_1(x)=\sum_{\ell m} (a_{1\ell m}|x|^\ell
+b_{1\ell m}|x|^{-\ell-1})Y_{\ell m}(x/|x|)
$$
and 
$$
f_2(x)=\sum_{\ell m} b_{2\ell m}|x|^{-\ell-1}Y_{\ell m}(x/|x|).
$$
The equality of the boundary conditions at $|x|=1$ implies
$$
b_{2\ell m}=a_{1\ell m}+b_{1\ell m}=\alpha_{\ell m}.
$$
The Neumann condition at $|x|=R$ implies 
$$
\ell a_{1\ell m}=(\ell+1)b_{1\ell m}R^{-2\ell-1}.
$$
(Note all coefficients vanish for $\ell=0$).
Hence 
$$
b_{2\ell m}=(1+\ell^{-1}(\ell+1)R^{-2\ell-1})b_{1\ell m}.
$$
We thus obtain
$$
-\int_{S(1)}\overline{g}\partial_rf_2=4\pi\sum_{\ell=1}^\infty\sum_{ m=-\ell}^\ell
(1+\ell^{-1}(\ell+1)R^{-2\ell-1})^2|b_{1\ell m}|^2(\ell+1)
$$
and 
\begin{eqnarray*}
-\int_{S(1)}\overline{g}\partial_rf_1&=&4\pi\sum_{\ell=1}^\infty\sum_{ m=-\ell}^\ell(1-R^{-2\ell-1})
(1+\ell^{-1}(\ell+1)R^{-2\ell-1})
(\ell+1)|b_{1\ell m}|^2\\
&\geq&-(1-3R^{-3})\int_{S(1)}\overline{g}\partial_rf_2.
\end{eqnarray*}
\end{proof}

\bigskip

Now we compare $E_{r,R}$ and $E'_{r,R}$.
Note that the only difference between $E_{r,R}$ and $E'_{r,R}$ is that in \eqref{Ebhprpr} $\psi_r$ is
outside the negative part. The following lemma shows that
$E_{r,R}(\beta, h, V)$ and  $E'_{r,R}(\beta, h, V)$ are essentially equivalent
up to a term of high order in $h$ (to appreciate
the next result, recall that $E_{r,R}(\beta, h, V)$ is
typically of order $h^{-3}$).

\begin{lemma}\label{lm:version2}
There exist universal constants $C_0,C_1,C_2>0$, such
that
\begin{eqnarray}
  E'_{r,R}(\beta, h, V) \le E_{r,R}(\beta, h, V)\label{EEprpr}
  &\le&   
  E'_{r,R}((1+C_0(r/R)^3+\delta)\beta, h, V) \\&&+
  C_1h^{-1}(\|V\|_{5/2}^{5/2}+r^{-2})
  +C_2h^2\delta^{-3}\beta^{-3}(\|V\|_4^4+ r^{-5})\non
  \end{eqnarray}
  for all $\delta>0$.
The same result holds for the Schr\"odinger case, with $\delta=C_0=C_2=0$, i.e.
\be
  E'_{r,R}(\beta, h, V) \le E_{r,R}(\beta, h, V) 
  \le  
  E'_{r,R}(\beta, h, V) +
  C_1h^{-1}(\|V\|_{5/2}^{5/2}+r^{-2}).
\ee
\end{lemma}

Note that norms of $V$ in this Lemma are in the whole space.

\medskip
{\it Proof.}  
Let $H=H_h(A)=T_h(A)-V$. Since
$$
   \tr \big[ \psi_r H \psi_r\big]_- \ge \tr \big[ \psi_r [H]_- \psi_r\big]_- = 
 \tr  \psi_r [H]_- \psi_r,
$$
the first bound in \eqref{EEprpr} is trivial.

In order to prove the second inequality, we write $\gamma: = {\bf 1}_{(-\infty,0]}(H)$ and calculate,
\begin{align}
\tr \psi_r^2 [H]_{-} &= \frac{1}{2} \tr\big( \psi_r^2 [H]_{-} + [H]_{-} \psi_r^2 \big) 
= \frac{1}{2} \tr\big( (\psi_r^2 H + H \psi_r^2) \gamma \big)\nonumber \\
&= \frac{1}{2} \tr\big( ([H, \psi_r],\psi_r] + 2\psi_r H \psi_r) \gamma \big).
\end{align}
Therefore, by the variational principle, and since $[[H, \psi_r],\psi_r]=-2h^2(\nabla\psi_r)^2$
\begin{align}\label{varp}
\tr\big[ \psi_r H \psi_r\big]_{-} \leq \tr \psi_r^2 [H]_{-} + h^2 \tr (\nabla \psi_r)^2 \gamma.
\end{align}
In order to estimate the last term we apply a Lieb-Thirring inequality from \cite{LLS}:
\begin{theorem}\cite{LLS}\label{thm:lls}
There exist a universal constant $C$ such that
for the semiclassical Pauli operator  $T_h(A)-V$ 
with a potential $V\in L^{5/2}(\bR^3)\cap L^4(\bR^3)$ and magnetic field $B=\nabla\times A
\in L^2(\bR^3)$ we have
\be
    \tr\big[ T_h(A)-V\big]_- \ge - Ch^{-3}\int \big[ V\big]_+^{5/2} 
 -  C\Big( h^{-2} \int |B|^2\Big)^{3/4}\Big( \int  \big[ V\big]_+^{4} \Big)^{1/4}.
\label{genlt}
\ee
\end{theorem}

Using $\tr H\gamma \le 0$ and applying \eqref{genlt}  to the operator $H - (\nabla \psi_r)^2$,
we obtain
\begin{align}
  - \tr (\nabla \psi_r)^2 \gamma& \geq
  \tr\big( H - (\nabla \psi_r)^2\big) \gamma \label{usemlt} \\
  & \geq - C h^{-3} \int \big[ V+(\nabla \psi_r)^2\big]_+^{5/2} -
  C\Big(h^{-2} \int_{\bR^3} |\nabla \times A|^2\Big)^{3/4}
  \Big( \int \big[V+(\nabla \psi_r)^2\big]_+^4 \Big)^{1/4}  \nonumber \\
  &\geq-Ch^{-3}\Big(\int [V]_+^{5/2}+r^{-2}\Big)-C\delta^{-3}\beta^{-3}\Big(\int [V]_+^4+ r^{-5}\Big)
  - \delta\beta h^{-2}
  \int_{\bR^3} |\nabla \times A|^2.\nonumber 
\end{align}
{F}rom \eqref{varp} we thus have 
\begin{eqnarray*}
  E_{r,R}(\beta, h, V) &\le& \inf_A\Big[ \tr \psi_r^2\big[ T_h(A) -
  V\big]_- 
  + (1+\delta)\beta \int_{\bR^3}
  |\nabla \times A|^2\Big]\\&&
 + Ch^{-1}\Big(\int [V]_+^{5/2}+r^{-2}\Big)
 +Ch^2\delta^{-3}\beta^{-3}\Big(\int [V]_+^4+ r^{-5}\Big).
\end{eqnarray*}
Exactly as in the upper bound in Lemma~\ref{lm:version1}, we can
replace $\int_{\bR^3} |\nabla \times A|^2$ with $\int_{B(R)} |\nabla
\otimes A|^2$ at the expense of increasing $\beta$ by a factor.
This proves Lemma~\ref{lm:version2} for the Pauli case.  For the
Schr\"odinger case we can use the non-magnetic Lieb-Thirring
inequality, i.e.  the second line of \eqref{usemlt} simplifies to
$$
   -   \tr (\nabla \psi_r)^2 \gamma \ge  
- C   h^{-3} \int \big[V+(\nabla \psi_r)^2\big]_+^{5/2}
$$
and the rest of the argument is the same.
\qed

\section{Refined local semiclassics with scales}

We consider a
local version of Theorem~\ref{eq:secondorder} for a model problem
living in the ball $B(\ell)$ of radius $\ell>0$. Note that not
only the Hamiltonian but also the magnetic field has been localized,
albeit on a twice bigger ball.   We introduce two
 scaling parameters, $\ell$ and $f$, that describe the lengthscale
and the strength of the potential and we follow the
dependence of the error term on these parameters.
This scale invariant formulation 
will be convenient for the multiscale analysis
in \cite{EFS3} and also in the reduction of the following
theorem to its special case, Theorem~\ref{thm:sc}.
Readers who are interested only in
the  semiclassical aspect of the result can think of
$f =\ell =1$.

\begin{theorem}[Semiclassical asymptotics]\label{thm:scMain} 
There exist universal constants $n_0 \in \bN$ and $\e >0$ such
that the following is satisfied.
Let $\kappa_0,f,\ell, h_0>0$ and let $\kappa \le \kappa_0 f^{-2}\ell^{-1}$.
Let $\psi\in C_0^\infty(\bR^3)$ with $\supp \psi \subset B(\ell)$
and let $V\in C^\infty(\ov{B}(\ell))$ be a real valued potential satisfying
\be
    |\pt^n \psi|\le C_n \ell^{-|n|}, \qquad |\pt^n V|\le C_n f^2\ell^{-|n|}
\label{derMain}
\ee 
for every multiindex $n$ with $|n|\le n_0$.
Then for the Pauli operator $H_h(A):= T_h(A)-V$, 
\begin{multline}
   \Bigg| \inf_A \Big( \tr [\psi H_h(A) \psi]_- + \frac{1}{\kappa  h^2}
\int_{B(2\ell)} |\nabla \otimes A|^2\Big)
  - 2(2\pi h)^{-3}\int_{\bR^3\times\bR^3} \psi(q)^2\big[ p^2 - V(q)\big]_- \rd q \rd p 
\Bigg| \\ \le  C  h^{-2+\e} f^{4-\e} \ell^{2-\e} 
\label{locscMain}
\end{multline}
for any $h\le h_0 f \ell$.
The constant $C$ depends only on $\kappa_0$, $h_0$ and on the constants $C_n$, 
 in \eqref{derMain}. The factor 2 in front of
the semiclassical term accounts for the spin and it is present only 
for the Pauli case.
\end{theorem}

{\it Remark 1.} For convenience, we introduced a new parameter $\kappa =\beta^{-1} h^{-2}$
instead of $\beta$. In the limit we consider, $h\to \infty$, $\beta h^2\ge \mbox{const}>0$,
the new parameter $\kappa$ will remain bounded uniformly.
We note that the constant factor $\beta$ in front of the magnetic energy
does not necessarily have to scale as $h^{-2}$ (see \cite{EFS1} for other
scalings) but our choice here
is dictated by the application to the Scott correction.
Similar results may be proven with a 
more general coefficient $\beta\sim h^{-\gamma}$ with a certain range of $\gamma$
 and the  exponent of the error term will be $\e=\e(\gamma)$.
We will not pursue the most general result in this paper.

\medskip

{\it Remark 2.}
By variation of $\kappa$, we obtain  from Theorem~\ref{thm:scMain}
the following estimate
\begin{align}
  \label{eq:21}
  \frac{1}{\kappa  h^2}
\int_{B(2\ell)} |\nabla \otimes A|^2\leq
 C  h^{-2+\e} f^{4-\e} \ell^{2-\e},
\end{align}
for (near) minimizing vector potentials $A$.

\bigskip

 The following result can be viewed as a partial  converse to \eqref{eq:21}
as it estimates the semiclassical error in terms of the magnetic field. Note that
the assumption in \eqref{eq:17} below is much weaker than \eqref{eq:21}.

\begin{theorem}\label{thm:UpperSemiclassNew}
Let the assumptions be as in  Theorem~\ref{thm:scMain}
 and assume that $A$ satisfies the bound
\begin{align}\label{eq:17}
\int_{B(2\ell)} |\nabla \otimes A|^2\leq
 C h^{-2}  f^{4} \ell^3.
\end{align}
Then, with $\e$ from Theorem~\ref{thm:scMain} we have
\begin{align}\label{eq:USNew}
C h^{-2}f^3\ell^{3/2}
 \Big\{ \int_{B(2\ell)}& |\nabla \otimes A|^2 \Big\}^{1/2}  + C h^{-1}f^3\ell \non\\ &\ge
\tr [\psi H_h(A) \psi]_{-} -  2(2\pi h)^{-3}\int_{\bR^3\times\bR^3} 
\psi(q)^2\big[ p^2 - V(q)\big]_- \rd q \rd p \nonumber \\
  & 
\ge  C  h^{-2+\e} f^{4-\e} \ell^{2-\e} - Ch^{-2}f^2\ell  \int_{B(2\ell)} |\nabla \otimes A|^2,
\end{align}
where the constants may depend on $h_0$ and $\kappa_0$ and on the constant in \eqref{eq:17}.
\end{theorem}

\subsection{Proof of the Main Theorem~\ref{thm:secondorder}}

Using the local semiclassical asymptotics Theorem~\ref{thm:scMain}
it is an easy localization argument to prove the global result
for compactly supported potentials. 

Recall the choice of a smooth cutoff function $\psi$
from Section~\ref{sec:loc}. Let $f=1$ and $\ell\ge 1$ sufficiently large
so that $\psi_\ell(x):= \psi(x/\ell)\equiv 1$ on the support of $V$.

The upper bound in Theorem~\ref{thm:secondorder} is obtained by setting $A=0$
 and using the first inequality in Lemma~\ref{lm:version2} (with $R=\infty$)
$$
   E(\beta, h, V)\le 2\tr [-h^2\Delta - V]_- \le2\tr \psi_\ell[-h^2\Delta - V]_- \psi_\ell
   \le 2\tr \big[\psi_\ell(-h^2\Delta - V) \psi_\ell\big]_-,
$$ 
where the traces are on $L^2(\bR^3)$ and
the factor 2 accounts for the spin. The upper bound now follows from
 the main result of \cite{SS} which is formulated for $\ell=1$, but it clearly 
extends to any value of $\ell$:

\begin{theorem}[{\cite[Theorem 12]{SS}}]\label{thm:SS}
Let $d\geq 3$ and $\psi \in C_0^{d+4}(\bR^d)$ be supported in a ball $B$ 
of radius $1$ and $V \in C^3(\overline B)$ be a real function. Let $H = -h^2 \Delta -V$, $h>0$,
acting on $L^2(\bR^d)$.
Then,
\begin{align}
\Big| \tr_{L^2(\bR^d)} [\psi H \psi]_{-} - (2\pi h)^{-d} \int_{\bR^d\times \bR^d}
\psi^2(q) [p^2-V(q)]_{-}\,\rd p\rd q
\Big| \leq C h^{-d+6/5}.
\end{align}
The constant $C>0$ here depends only on $d$, $\| \psi \|_{C^{d+4}}$ and $\| V \|_{C^3}$.
\end{theorem}

For the lower bound in Theorem~\ref{thm:secondorder} we use the IMS localization with a partition of unity
$\psi_\ell^2 + \wt\psi_\ell^2\equiv 1$ and we borrow a small fraction 
of the kinetic energy to control the localization error. We have
\be
   T_h(A)-V \ge (1-\delta)\psi_\ell\big( T_h(A) - V_\delta\big)\psi_\ell
 + \delta\big[ T_h(A) -  \delta^{-1}h^2 I_\ell\big]
\label{ims2}
\ee
with $\delta\in (0,\frac{1}{4})$,
  $I_\ell:=(\nabla\psi_\ell)^2 +(\nabla\wt\psi_\ell)^2$ and $V_\delta = (1-\delta)^{-1} V$.
Here we dropped the term $\wt\psi_\ell\big( T_h(A) - V_\delta\big)\wt\psi_\ell$ that is positive
since $V_\delta=0$ on the support of $\wt\psi_\ell$.
The second term in \eqref{ims2} is estimated by the magnetic Lieb-Thirring inequality
similarly to \eqref{usemlt}:
$$
\Tr \big[ T_h(A) -  \delta^{-1}h^2 I_\ell\big]_- \ge - Ch^2\delta^{-5/2} - C\beta^{-3}h^2\delta^{-4}
 - \beta  \int |\nabla \times A|^2,
$$
where the constants depend on $V$ and $\ell$. Using $\beta h^2\ge \kappa_0$ and assuming $h\le 1$,
 we can combine the two
error terms to obtain 
$$
   E(\beta, h, V) \ge (1-\delta) E'_{\ell,\infty} \big( \beta, h, V_\delta\big) -  Ch^2\delta^{-3}
  \ge E_{\ell, 2\ell} (\beta/2, h, V_\delta)  -  Ch^2\delta^{-3} - C h^{-1}.
$$
In the last step we used the second inequality from \eqref{EEprpr} with $R=\infty$ 
and then the monotonicity of $E_{\ell, R}<0$ in $R$.
The energy $E_{\ell, 2\ell} (\beta/2, h, V_\delta)$ is estimated in \eqref{locscMain}
(with $\kappa = 2\beta^{-1}h^{-2}$) and we get
$$
   E(\beta, h, V) \ge 2(2\pi h)^{-3} \int_{\bR^3\times \bR^3} \big[ p^2 - V_\delta(q)\big]_-\rd p\rd q
 -Ch^{-2+\e}-  Ch^2\delta^{-3},
$$
where we used again that $\psi_\ell\equiv 1$ on the support of $V_\delta$.
It is easy to see that
$$
  \int_{\bR^3\times \bR^3} \big[ p^2 - V_\delta(q)\big]_-\rd p\rd q 
 =\int_{\bR^3\times \bR^3} \big[ p^2 - V(q)\big]_-\rd p\rd q + O(\delta)
$$
and finally, choosing $\delta= h^{1+\e}$ we obtain 
an error term $Ch^{-2+\e}$ assuming $\e\le 1/4$.
This completes the proof of Theorem~\ref{thm:secondorder}. \qed

\section{Refined local semiclassics with lower bound}

The main technical result of this paper is the following
version of Theorem~\ref{thm:scMain} where additionally 
 a lower bound on $V$, namely
\eqref{lowerbV} is assumed. 
Furthermore, we impose the condition that $\nabla \cdot A =0$.

\begin{theorem}[Semiclassical asymptotics]\label{thm:sc} 
There exist universal constants $n_0 \in \bN$ and $\e >0$ such
that the following is satisfied.
Let $\kappa_0,f,\ell, h_0>0$ and let $\kappa \le \kappa_0 f^{-2}\ell^{-1}$.
Let $\psi\in C_0^\infty(\bR^3)$ with $\supp \psi \subset B(\ell)$
and let $V\in C^\infty(\ov{B}(\ell))$ be a real valued potential satisfying
\be
    |\pt^n \psi|\le C_n \ell^{-|n|}, \qquad |\pt^n V|\le C_n f^2\ell^{-|n|}
\label{der}
\ee 
for every multiindex $n$ with $|n|\le n_0$,
and for some $c_0>0$ let 

\be
 \inf \{ V(x)\; : \; x\in \ov{B}(\ell)\} \ge c_0 f^2.
\label{lowerbV}
\ee
Then
\begin{multline}
   \Bigg| {\inf}'_A \Big( \tr [\psi H_h(A) \psi]_- + \frac{1}{\kappa  h^2}
\int_{B(2\ell)} |\nabla \otimes A|^2\Big)
  - 2(2\pi h)^{-3}\int_{\bR^3\times\bR^3} \psi(q)^2\big[ p^2 - V(q)\big]_- \rd q \rd p 
\Bigg| \\ \le  C  h^{-2+\e} f^{4-\e} \ell^{2-\e} 
\label{locsc}
\end{multline}
for any $h\le h_0f\ell$.
Here ${\inf}'_A$ denotes that the infimum is taken over all $A \in
H^1({\mathbb R}^3)$ satisfying that 
\begin{align}
  \label{eq:22}
  \nabla \cdot A = 0,\qquad \text{ on } B(5 \ell/4).
\end{align}
The constant $C$ depends only on $\kappa_0$, $h_0$  and on the constants $C_n$, 
and $c_0$ in \eqref{der} and \eqref{lowerbV}. The factor 2 in front of
the semiclassical term accounts for the spin and it is present only 
for the Pauli case.
\end{theorem}

Note that this theorem is formulated in a scale invariant way, so it will be
sufficient to prove it for $f=\ell =1$.

We first show that the  condition $\nabla \cdot A = 0$
is irrelevant for the statement of Theorem~\ref{thm:sc}, i.e.
one gets the same
result by dropping the condition $\nabla \cdot A = 0$.

\begin{theorem}\label{thm:scDiv}
Suppose that Theorem~\ref{thm:sc} holds, where the infimum is taken over all 
$A$ with $\nabla \cdot A =0$ on $B(5\ell/4)$. Then  Theorem~\ref{thm:sc} 
also holds where the infimum is taken over all $A$.
\end{theorem}

\begin{proof}
Clearly the unrestricted infimum is smaller than the one with the condition
 $\nabla \cdot A =0$ imposed.
We will prove the opposite inequality, namely that there exists $C>0$ such that
\be\label{willpr}
{ \inf_A}' \Big( \tr [\psi H_h(A) \psi]_- + \frac{1}{C \kappa  h^2}
\int_{B(2\ell)} |\nabla \otimes A|^2\Big)
\leq  \inf_A \Big( \tr [\psi H_h(A) \psi]_- + \frac{1}{\kappa  h^2}
\int_{B(2\ell)} |\nabla \otimes A|^2\Big),
\ee
where $ \inf'_A$ means that we take the restricted infimum over $A$ with
 $\nabla \cdot A =0$ on $B(5\ell/4)$. Since the constant $\kappa$ is arbitrary this implies the result.

For an arbitrary $A \in H^1$, we can add a constant to $A$ in order to get, by a Poincar\'e
 inequality, that
\begin{align}\label{eq:Poin1}
\ell^{-2} \int_{B(2\ell)} A^2 \leq  C \int_{B(2\ell)} |\nabla \otimes A|^2.
\end{align}
Clearly, this additive constant does not change the energy, i.e. on
the right hand side of \eqref{willpr} we can assume that
$A$ satifies \eqref{eq:Poin1}.

Choose a localization function $\chi \in C_0^{\infty}(B(3\ell/2))$, $\chi = 1$ on $B(5\ell/4)$. 
Define $A_\chi$ to satisfy
\begin{align}
\nabla \times A_{\chi}= B_{\chi}:=\nabla \times (\chi A),\qquad \nabla \cdot A_{\chi} = 0.
\end{align}
This system defines $A_{\chi}$ up to an additive constant that we will choose in order 
to be able to apply the Poincar\'e inequality below.
Finally, define a vector field $A'$ by $A':=\chi A_{\chi}$, then clearly $\nabla\cdot A'=0$
 on $B(5\ell/4)$.

Then, 
\begin{align}\label{AA'}
 \tr [\psi H_h(A) \psi]_- =  \tr [\psi H_h(\chi A) \psi]_- = \tr [\psi H_h( A_\chi) \psi]_- =
 \tr [\psi H_h(A') \psi]_- ,
\end{align}
where the middle equality follows from the gauge equivalence
of $A_\chi$ and $\chi A$, and the other two follow from the fact that $\chi\equiv 1$
on $\mbox{supp} \,\psi$.
Also, using the Poincar\'e inequality twice,
once for $A$ \eqref{eq:Poin1} and once for $A_\chi$, we have
\begin{align}\label{pointwice}
\int_{B(2\ell)} |\nabla \otimes A'|^2& \leq \int_{B(2\ell)} |\nabla \otimes A_{\chi}|^2+ 
C \ell^{-2}\int_{B(2\ell)} | A_{\chi}|^2 \leq C\int_{B(2\ell)}| \nabla \otimes A_{\chi}|^2 \nonumber \\
&\leq C \int_{{\mathbb R}^3} | \nabla \otimes A_{\chi}|^2 =
C \int_{{\mathbb R}^3} B_{\chi}^2 \leq C \Big(\ell^{-2} \int_{B(2\ell)} A^2 +
 \int_{B(2\ell)}|\nabla \times A|^2\Big)\nonumber \\
&\leq C \int_{B(2\ell)} |\nabla \otimes A|^2.
\end{align}
\end{proof}

\subsection{Multiscaling}

To  prove that Theorem~\ref{thm:scMain} follows from Theorem~\ref{thm:scDiv}
requires to perform a multiscale decomposition around the
sets where $V$ is too small and it violates \eqref{lowerbV}.
We present the setup for multiscaling for general potentials that will also
be applicable for resolution of Coulomb like singularities.

We will assume that the potential has a multiscale structure. Intuitively, this means that
there exist two scaling functions, $f, \ell: \bR^3\to \bR_+$  such
that for any $u\in \bR^3$, within the ball $B_u(\ell(u))$ centered
at $u$ with radius $\ell(u)$, the size of $V$ is of order $f^2(u)$
and $V$ varies on scale $\ell(u)$.  Moreover, we also require
that  the continuous family of balls  $B_u(\ell(u))$
supports a regular partition of unity.
The following lemma states this condition precisely.
This statement was  proved in Theorem 22 of \cite{SS}
with an explicit construction\footnote{Multiscaling was introduced in
  semiclassical problems in \cite{IS} (see also \cite{Sob})}.

We will use the notation $B_x(r)$ for the ball of radius $r$ and 
center at $x$ and if $x=0$, we use $B(r) = B_0(r)$.

\begin{lemma}\label{lem:PartUnityMultScale} Fix a cutoff function $\psi\in C_0^\infty(\bR^3)$
supported in the unit ball $B(1)$ satisfying $\int\psi^2=1$. 
Let $\ell: \bR^3\to (0,1]$ be a $C^1$ function with $\|\nabla\ell\|_\infty<1$.
Let $J(x,u)$ be the Jacobian of the map $u\mapsto (x-u)/\ell(u)$ and we define
$$
   \psi_u(x) = \psi\Big( \frac{x-u}{\ell(u)}\Big) \sqrt{J(x,u)} \ell(u)^{3/2}.
$$
Then, for all $x\in\bR^3$,
\be
   \int_{\bR^3} \psi_u(x)^2 \ell(u)^{-3} \rd u =1,
\label{partun}
\ee
and for all multi-indices $n\in \bN^3$ we have
\be\label{psider}
  \| \pt^n \psi_u\|_\infty \le C_n \ell(u)^{-|n|} , \qquad |n|=n_1+n_2+n_3,
\ee
where $C_n$ depends on the derivatives of $\psi$ but is independent of $u$.
\end{lemma}

We will require that the potential satisfies
\be
     |\pt^n V(u)| \le C_n f(u)^{2}\ell(u)^{-|n|}
\label{Vbound}
\ee
for all $n \in \bN^3$
uniformly in $u$ in some domain $\Omega\subset \bR^3$.
In applications, $\Omega$ will exclude an $h$-neighborhood of the core 
of the Coulomb potentials.

For brevity, we will often use $\ell_u= \ell(u)$ and $f_u = f(u)$.

\subsection{Proof of   Theorem~\ref{thm:scMain}}

We now show how Theorem~\ref{thm:scMain} follows from Theorem~\ref{thm:sc}.

\begin{proof}[Proof of Theorem~\ref{thm:scMain}]
By Theorem~\ref{thm:scDiv} we may forget about the condition
\eqref{eq:22} in Theorem~\ref{thm:sc} and consider the infimum over all $A\in H^1$.
Moreover, since \eqref{locscMain} is local, we may assume without loss 
of generality that $V$ has compact 
support contained in $B(5\ell/4)$.
After these remarks, Theorem~\ref{thm:scMain} would follow from Theorem~\ref{thm:sc}
once we have shown that the additional condition \eqref{lowerbV} in Theorem~\ref{thm:sc} 
can be removed.

Since the statement of Theorem~\ref{thm:scMain} is scale invariant we may also assume
 that $f=\ell=1$ in the statement of that theorem. We will apply Theorem~\ref{thm:sc}
 and the partition of unity Lemma~\ref{lem:PartUnityMultScale} with a different choice 
of $f$ and $\ell$.

The statement of Theorem~\ref{thm:sc} consists of a lower and an upper
bound on the total quantum energy. Since the proof of the upper bound
(see \eqref{eq:upperSS}) does not use the condition \eqref{lowerbV}, 
the upper bound in Theorem~\ref{thm:scMain} follows immediately,
 so we only need to prove the lower bound.
Define the smooth functions
\be\label{newdeflf}
\ell(x) = f(x)^2: =\frac{1}{K}\big( V(x)^2 + h^{2\al}\big)^{1/2}
\ee
with some positive exponent with  $2/5< \al<1/2$
and recall the notation $\ell_x=\ell(x)$ and $f_x=f(x)$.
Here $K>0$ (independent of $h$, but depending on $V$ and  $h_0$ given in Theorem~\ref{thm:scMain})
 is chosen so large that 
\begin{align}
\label{eq:Boundsell}
\| \nabla \ell \|_{\infty} <1/4, \qquad \ell_u \leq 1/4, \quad u\in B( 2).
\end{align}
By Lemma~\ref{lem:PartUnityMultScale} there is a partition of unity $\psi_u$ associated to $\ell$.
Inserting this partition of unity, 
we have for sufficiently small values of $h$ that
\begin{align}\label{eq:AnotherIMS}
\tr \big[ \psi (T_h(A) &- V) \psi \big]_{-} + \frac{1}{\kappa h^2} \int_{B(2)} |\nabla \otimes A|^2 
\rd x  \\ 
& =\tr \Big[\int_{B(3/2)} \frac{\rd u}{\ell_u^3} \psi\Big(
   \psi_u [ T_h(A) - V]\psi_u - h^2|\nabla\psi_u|^2 \Big)\psi\Big]_{-} +
 \frac{1}{\kappa h^2} \int_{B(2)} | \nabla \otimes A|^2 \cr
 & \ge \tr \Big[\int_{B(3/2)} \frac{\rd u}{\ell_u^3} 
\wt\psi_u [ T_h(A) - V - Ch^2\ell_u^{-2}]\wt\psi_u  \Big]_{-} +
 \frac{1}{\kappa h^2} \int_{B(2)} | \nabla \otimes A|^2\cr
&\geq \int_{u \in B(3/2)} \Big\{
\tr \big[ \wt \psi_u (T_h(A) - V - C h^2 \ell_u^{-2}) \wt \psi_u \big]_{-} + \frac{1}{\kappa' h^2}
 \int_{B_u(2\ell_u)} |\nabla \otimes A|^2 \rd x\Big\} \frac{\rd u}{\ell_u^3} \non
\end{align}
with $\wt \psi_u := \psi_u \psi$ and $\kappa'$ a fixed constant times $\kappa$.
We also used
 $\tr [\int O_u \rd u]_-\ge \int \tr [O_u]_-\rd u$
for any continuous family of operators $O_u$.
In arriving at \eqref{eq:AnotherIMS} we 
reallocated the localization error by using
\be\label{reall}
\frac{\ell_u}{2} \leq \ell_x \leq \frac{3\ell_u}{2},
 \qquad \text{ for all } x \in B_u(2\ell_u)
\ee
which follows from $\| \nabla \ell\|_{\infty}\le 1/4$.
We also reallocated the
magnetic energy by using the estimate
\begin{align}
\int_{u \in B(3/2)} \int_{B_u(2\ell_u)}
 |\nabla \otimes A(x)|^2 \rd x \frac{\rd u}{\ell_u^3}
&\leq
\int_{x \in B(2)} \int_{u \in B_x(4\ell_x)}
 |\nabla \otimes A(x)|^2 (2\ell_x/3)^{-3} \rd u \rd x \nonumber \\
&=
C \int_{x \in B(2)}  |\nabla \otimes A(x)|^2 \rd x,
\end{align}
where we again used \eqref{reall}.

Consider now $u\in B(3/2)$ and
suppose first that $|V(u)| \geq h^{\al}$. Then, if $K$ in the definition of $f, \ell$ 
\eqref{newdeflf}
is chosen sufficiently large, we have by the uniform bounds on the derivatives of
 $V$ that $|V(x)| \geq |V(u)|/2\geq f_u^2/4$ for all $x \in B_u(\ell_u)$. 
So \eqref{lowerbV} is satisfied (with $c_0=1/4$). Using \eqref{psider},  the estimates in \eqref{der} 
are also easily seen to be satisfied for the cutoff function
 $\wt \psi_u$ and for the potential $V + Ch^2 \ell_u^{-2}$
with $\ell=\ell_u$ and $f=f_u$ (here we use that $\al<2/3$ and that $K$ may depend on $h_0$).
 We can 
therefore conclude by Theorem~\ref{thm:scDiv} that for any $A$, 
\begin{align}\label{eq:Vlarge}
&\int_{\{u \in B(3/2):|V(u)|\geq h^{\al}\}} \Big\{
\tr \big[ \wt \psi_u (T_h(A) - V - C h^2 \ell_u^{-2}) \wt \psi_u \big]_{-} + 
\frac{1}{\kappa' h^2} \int_{B_u(2\ell_u)} |\nabla \otimes A|^2 \rd x\Big\} \frac{\rd u}{\ell_u^3}
\non \\
&\geq \int_{\{u \in B(3/2):|V(u)|\geq h^{\al}\}} \Big\{
 \frac{2}{(2\pi h)^{3}}\iint \wt \psi_u(q)^2\big[ p^2 - V(q)- C h^2 \ell_u^{-2}\big]_- 
\rd q \rd p - C  h^{-2+\e} \ell_u^{4-3\e/2}
 \Big\} \frac{\rd u}{\ell_u^3} \nonumber \\
 &\geq 2(2\pi h)^{-3}\int_{\bR^3\times\bR^3} 
\psi(q)^2\big[ p^2 - V(q) \big]_- \rd q \rd p - Ch^{2-2\alpha} h^{-3}
 - C h^{-2+\e} \int_{B(3/2)} \ell_u^{1-3\e/2} \rd u,
\end{align}
where we used \eqref{partun} and  that $h^2 \ell_u^{-2} \leq C h^{2-2\al}$ to get the last inequality.
Since $\ell$ is bounded and $\al<1/2$,  both error terms in \eqref{eq:Vlarge} are acceptable.

For the set of $u$'s where $V(u)$ is small,  i.e., $|V(u)|\le h^{\al}$,
we  use the magnetic Lieb-Thirring inequality.
 We introduce $A'= \chi (A-c)$, where $\chi \in C_0^{\infty}(B_u(2\ell_u))$, $\chi = 1$
 on $B_u(\ell_u)$ and $|\nabla \chi|\leq C/\ell_u$. 
Here $c=\int_{B_u(2\ell_u)} A$.
Then we have, using the Poincar\'e inequality,
\begin{align*}
\int |\nabla \otimes A'|^2  \leq C \int_{B_u(2\ell_u)} \Big( \ell_u^{-2} (A-c)^2 
+|\nabla \otimes A|^2\Big) \leq C'  \int_{B_u(2\ell_u)}|\nabla \otimes A|^2 .
\end{align*}
So we get
\begin{align}
&\tr \big[ \wt \psi_u (T_h(A) - V - C h^2 \ell_u^{-2}) \wt \psi_u \big]_{-} 
+ \frac{1}{\kappa' h^2} \int_{B_u(2\ell_u)} |\nabla \otimes A|^2 \nonumber \\
&\geq
\tr \big[ T_h(A') - (V + C h^2 \ell_u^{-2} ){\bf 1}_{B_u(\ell_u)}) \big]_{-} 
+ \frac{1}{\kappa' h^2} \int_{B_u(2\ell_u)} |\nabla \otimes A|^2  \nonumber \\
&\geq - C\Big\{ h^{-3} \int_{B_u(\ell_u)} \big[ V + C h^2 \ell_u^{-2}\big]_+^{5/2} + 
\Big(\int_{B_u(\ell_u)} \big[V + C h^2 \ell_u^{-2} \big]_+^{4}\Big)^{1/4} 
\Big( h^{-2} \int_{B_u(2\ell_u)} |\nabla \otimes A|^2 \Big)^{3/4} \Big\} \nonumber \\
&\qquad\qquad\qquad+ \frac{1}{\kappa' h^2} \int_{B_u(2\ell_u)} |\nabla \otimes A|^2 \nonumber \\
&\geq 
- C (\kappa')\Big\{ h^{-3} \int_{B_u(\ell_u)} 
\big[V + C h^2 \ell_u^{-2}\big]_+^{5/2}
 +\int_{B_u(\ell_u)} \big[V + C h^2 \ell_u^{-2} \big]_+^{4} \Big\} \non\\
&\geq - C(\kappa') h^{-3+5\al/2} \ell_u^3,
\end{align}
since $|V(u)|\le h^{\al}$ and  $\ell_u\sim h^{\al}$,
so we get $|V(x)|\le |V(u)| + C \ell_u \le Ch^{\al}$ for any $x\in B_u(\ell_u)$.
In summary, we have
\begin{align}
\int_{\{u \in B(3/2):|V(u)|\leq h^{\al}\}} &\Big\{
\tr \big[ \wt \psi_u (T_h(A) - V - C h^2 \ell_u^{-2}) \wt \psi_u \big]_{-} 
+ \frac{1}{\kappa' h^2} \int_{B_u(2\ell_u)} |\nabla \otimes A|^2 
\Big\} \frac{\rd u}{\ell_u^3} \non \\
&\geq -C  
h^{-3+5\al/2}.\label{Vsmalll}
\end{align}
Since $\al>2/5$, this error term is acceptable, in fact $\al=4/9$
optimizes the errors from \eqref{eq:Vlarge} and \eqref{Vsmalll}, 
which shows that $\e < 1/9$. Thus the final $\e$ in Theorem~\ref{thm:scMain}
is the minimum of 1/9 and the $\e$ obtained in Theorem~\ref{thm:sc}.
This finishes the proof of the lower bound and therefore the proof of Theorem~\ref{thm:scMain}.
\end{proof}

\subsection{Proof of Theorem~\ref{thm:UpperSemiclassNew}}\label{sec:frommaintoupper}
We now show how Theorem~\ref{thm:UpperSemiclassNew} follows from Theorem~\ref{thm:scMain}.

\begin{proof}[Proof of Theorem~\ref{thm:UpperSemiclassNew}]
By the argument given in the proof of Theorem~\ref{thm:scDiv}, we may assume that $\nabla \cdot A = 0$
on a neighborhood of $\supp \psi$. Also, by a rescaling we may assume
that $f=\ell =1$.

The lower bound in  \eqref{eq:USNew} follows directly
 from Theorem~\ref{thm:scMain}.
The upper bound in \eqref{eq:USNew} will be constructed using an explicit trial
state. Define
\begin{align}
  \label{eq:18}
  \gamma := \int \chi_u \exp(-ic_u x) \gamma_u \exp(-ic_u x) \chi_u
  \frac{\rd u}{L_0^3}
\end{align}
with a length scale  $L_0\le 1$ that
will be optimized at the end of the proof.
Here $\chi_u(x) = \chi((x-u)/L_0)$, with $\chi \in
C_0^{\infty}(B(1))$ a real, positive function with $\int
\chi^2 = 1$. Clearly $\chi_u$ is supported in the ball $B_u(L_0)$.
We choose the parameter $c_u := \int_{B_u(L_0)} A$. 
Also
\begin{align*}
  \gamma_u = {\bf 1}_{(-\infty,0]} \big( \chi_u \psi H_h(A=0) \psi \chi_u \big).
\end{align*}
Notice that $\gamma_u$ is real and acts like a scalar in spinor
space. Therefore the contribution of the cross terms
$[\bsigma\cdot (-ih\nabla)][\bsigma \cdot A] +[\bsigma \cdot A] [\bsigma\cdot (-ih\nabla)]$
on $\psi\chi_u\gamma_u\chi_u\psi$  vanishes and we get
\begin{align*}
  \tr \psi H_h(A) \psi \gamma &= \int \tr[ \chi_u \psi H_h(A-c_u) \psi\chi_u
  \gamma_u] \frac{\rd u}{L_0^3} \\
  &= \int_{B(2)} \Big( \tr [ \chi_u \psi H_h(A=0) \psi \chi_u ]_{-} + \tr[ \chi_u^2
  \psi^2 (A-c_u)^2 \gamma_u]\Big) \frac{\rd u}{L_0^3}. 
\end{align*}
The $\rd u$ integration can be restricted to $B(2)$ by the support
properties of $\psi$ and $\chi_u$ and by $L_0\le 1$.
By a rescaling to a ball of unit size and an application of an optimal
semiclassical result (recalled below as Theorem~\ref{thm:ivrii2}) we
get that
\begin{align}
  \label{eq:23}
    \tr [ \chi_u \psi H_h(A=0) \psi \chi_u ]_{-} \leq 2(2\pi
  h)^{-3}\int_{\bR^3\times\bR^3} \chi^2_u(q)\psi(q)^2\big[ p^2 -
  V(q)\big]_- \rd q \rd p + C (h/L_0)^{-1},
\end{align}
with a constant $C$ independent of $u$.
Performing the integration over $u$ we will therefore get the desired
upper bound if we can estimate the $A^2$ term as well.
For this, we apply H\"{o}lder, Lieb-Thirring and Sobolev inequalities
to get (with $\wt \gamma_u := \chi_u \psi \gamma_u \psi \chi_u$ and
$\wt \rho_u$ the associated density that is supported on $B_u(L_0)$),
\begin{align*}
  \int (A-c_u)^2 \wt \rho_u & \leq \Big( \int_{B_u(L_0)} (A-c_u)^5 \Big)^{2/5} \Big(\int \wt
  \rho_u^{5/3} \Big)^{3/5}\\
& \leq \Big( \int_{B_u(L_0)} (A-c_u)^6 \Big)^{3/10}
\Big( \int_{B_u(L_0)} (A-c_u)^2 \Big)^{1/10}
 \Big(\int \wt
  \rho_u^{5/3} \Big)^{3/5}\\
& \leq \Big( \int_{B_u(L_0)} |\nabla\otimes A|^2 \Big)^{9/10}
\Big( L_0^2\int_{B_u(L_0)}  |\nabla\otimes A|^2 \Big)^{1/10}
 \Big(\int \wt
  \rho_u^{5/3} \Big)^{3/5}\\
&\leq C L_0^{1/5} \| \nabla \otimes A\|_{L^2(B_u(L_0))}^2
\Big( h^{-2} \tr[ D^2 \wt \gamma_u] \Big)^{3/5},
\end{align*}
where we defined $D:= -ih\nabla$ for brevity.
Here we used the choice of $c_u$ and the Poincar\'e inequality in
$B_u(L_0)$ to control $A- c_u$ in $L^2$ by $L_0 \| \nabla \otimes
A\|_{L^2(B_u(L_0))}$.

Since we have the following bound, which we will prove below,
\begin{align}
  \label{eq:19}
  \tr[ D^2 \wt \gamma_u] \leq C h^{-3} L_0^3,
\end{align}
we can estimate
\begin{align}
  \label{eq:20}
  \int_{B(2)}\tr[ \chi_u^2
  \psi^2 (A-c_u)^2 \gamma_u] \frac{\rd u}{L_0^3} &
\leq C L_0^2 h^{-3}
  \int \| \nabla \otimes A\|_{L^2(B_u(L_0))}^2 \frac{\rd u}{L_0^3}
  \nonumber \\
&\leq C L_0^2 h^{-3} \int_{B(2)} | \nabla \otimes A|^2.
\end{align}
Now we choose $L_0$ by optimizing the error terms $CL_0^{-3}(h/L_0)^{-1}$
from \eqref{eq:23} and from \eqref{eq:20}. Taking into account $L_0\le 1$, the result is the choice
\be
  L_0= \min\Big\{ 1, \; h^{1/2} \Big( \int_{B(2)} |\nabla\otimes A|^2\Big)^{-1/4}\Big\}.
\label{L0choice}
\ee
Therefore, we obtain the desired upper bound in \eqref{eq:USNew}.

It only remains to establish \eqref{eq:19}.
Applying \eqref{eq:23} for an upper bound (and using that the
first term on the right  hand side is non-positive) and the non-magnetic
Lieb-Thirring estimate for a lower bound we get
\begin{align}
  \label{eq:24}
C L_0h^{-1} &\geq
\tr [ \chi_u \psi H_h(A=0) \psi \chi_u ]_{-} \nonumber \\
&\geq  
\frac{1}{2} \tr[ D^2 \wt \gamma_u] +
\tr[ \frac{1}{2} D^2 - V 1_{\supp \chi_u}]_{-} \nonumber \\
&\geq \frac{1}{2} \tr[ D^2 \wt \gamma_u]
- C h^{-3} \int_{\supp \chi_u} V^{5/2}.
\end{align}
Since  $L_0^{-1} \le Ch^{-1}$ from \eqref{L0choice} and \eqref{eq:17}, 
this estimate clearly implies \eqref{eq:19}.
This completes the proof of Theorem~\ref{thm:UpperSemiclassNew}.
\end{proof}

\section{Proof of the Semiclassical Theorem~\ref{thm:sc}}\label{sec:sc}

The rest of the paper is devoted to the proof of the main technical result,
Theorem~\ref{thm:sc}.
We can assume that $f=\ell =1$ since \eqref{locsc} is scale invariant.

For the proof of the upper bound, we just set $A=0$, i.e.
\begin{multline}\label{eq:upperSS}
 {\inf_A}' \Big( \tr [\psi H_h(A) \psi]_- + \frac{1}{\kappa
 h^{2}}\int_{B(2)} |\nabla \otimes A|^2\Big) \le  \tr [\psi (D^2 - V) \psi]_-
\\
\le 2(2\pi h)^{-3}\iint \psi(q)^2\big[ p^2 - V(q)\big]_- \rd q \rd p 
+ C  h^{-2+1/5} ,
\end{multline}
where the last estimate is a consequence of Theorem~\ref{thm:SS}.
Actually, the exponent in the error term in the upper bound in \eqref{eq:upperSS} 
can be improved from $-2+1/5$ to $-1$ by appealing to Theorem~\ref{thm:ivrii} below.
 However, we will not need this.

The 
proof of the lower bound in Theorem~\ref{thm:sc} has several steps.
First, in Section \ref{sec:L0loc},
 we localize onto balls of radius $L_0\gg h^{1/2}$
at the expense of an acceptable error of order $O(h^{1+\e})$, in fact
we will choose $L_0:= h^{1/2-\e_0}$ with some small $\e_0$, and
$\e$ will depend on $\e_0$. (Note that here $L_0$ is chosen
differently than in the proof of Theorem~\ref{thm:UpperSemiclassNew}.)
Then we  wish to replace $A$ by a smoothed out version $A_r$ on scale $r:= h^{1/2+\rho}
\ll h^{1/2}$ for some small $\rho>0$
with an error of order $O(h^{1+\e})$, where $\e$ will
depend on $\varrho$. This will eventually be achieved
in Section \ref{sec:smooth} (Theorem~\ref{thm:replace}). In order to do that,
we will need apriori bounds on the magnetic field
and the momentum distribution of the low energy trial states. 
To obtain these apriori bounds, in  Section \ref{sec:remA}   we will introduce a second
localization on an intermediate scale $L_1= h^{1/2+\e_0}$ with $r\ll L_1\ll L_0$,
and we show that on this scale $A$ can be
neglected.  These apriori bounds will 
have a weaker precision of order $O(h^{1-\e})$ since localization
onto a short scale $L_1\ll h^{1/2}$ is expensive, but this will
be sufficient as an input for Section  \ref{sec:smooth}.

Finally we go back to the larger scale $L_0$ and use a semiclassical result
for the operator with the smoothed
vector potential $A_r$. Note that the scale of regularity, $r$, of the
 vector potential is much smaller than $L_0$, so
it is not a straightforward application of standard semiclassical
results with $C^\infty$ data. However, one can keep track of the dependence
of the error terms in the standard semiclassical statements
on the derivatives of the symbol, which in our case will
be powers of $L_0/r$. If $L_0/r$ is not too large, 
this error can be compensated by the smallness of $A$, still
rendering a few derivatives of $A_r$  bounded which is
sufficient for the semiclassical asymptotics.

\subsection{Localization onto balls of size $L_0$}\label{sec:L0loc}

We introduce a partition of unity  on the lengthscale $L_0 = h^{1/2-\e_0}\le 1/4$
with some sufficiently small $\e_0>0$,
i.e. we choose $\phi_*\in C_0^\infty(\bR^3)$, $\int \phi_*^2=1,$ 
$\supp \phi_{*} \subset B(1)$, and define
$$
   \phi_u (x) = \phi_*\Big( \frac{x-u}{L_0} \Big)
$$
then
$$
    \int_{\bR^3} \phi_u^2(x) \frac{\rd u}{L_0^3} \equiv 1.
$$
Inserting this identity into  $\tr [\psi H_h(A) \psi]_-$ and using
IMS localization, we have
\begin{multline}\non
 \tr [\psi H_h(A) \psi]_- = \tr \Bigg[ \int_{B(3/2)}
  \frac{\rd u}{L_0^3} \psi\Big[ \phi_u H_h(A) \phi_u 
- h^2|\nabla \phi_u|^2 \Big]\psi \Bigg]_-  \\
 \ge\tr \Bigg[ \int_{B(3/2)} \frac{\rd u}{L_0^3} \psi \phi_u \big(H_h(A)  
- Ch^2 L_0^{-2} \big)\phi_u\psi \Bigg]_-
\ge
   \int_{B(3/2)} \frac{\rd u}{L_0^3}\tr \Bigg[ \psi \phi_u \big(H_h(A)  
- Ch^2L_0^{-2}\big)\phi_u \psi \Bigg]_- 
\end{multline}
after reallocating the localization error. Notice that the $\rd u$ integration
can be restricted to $B(3/2)$, otherwise $\psi\phi_u =0$.
We can again redefine the potential
$$
   V \to  V^+:= V + Ch^2L_0^{-2}.
$$
The change of the semiclassical term due to this modification,
\begin{multline}
\Bigg|  2(2\pi h)^{-3}\iint \psi(q)^2\big[ p^2 - V(q)\big]_- \rd q \rd p 
 -2(2\pi h)^{-3}\iint \psi(q)^2\big[ p^2 - V^+(q)\big]_- \rd q \rd p 
\Bigg|  \\
 \le C h^{-3}  h^2L_0^{-2} \le Ch^{-2+2\e_0},
\end{multline}
can be incorporated into the error term in \eqref{locsc}.
Therefore to prove the lower bound in \eqref{locsc}, it is sufficient to prove that,
for some $\e>0$, 
\begin{multline}\label{locsc1}
   {\inf_A}' \Big( \int_{B(3/2)} \frac{\rd u}{L_0^3} \tr [\psi\phi_u H_h^+(A)\phi_u \psi]_-
 + \frac{1}{\kappa h^2}\int_{B(2)} |\nabla \otimes A|^2\Big)\\ 
  \ge 
  2(2\pi h)^{-3}\iint \psi(q)^2\big[ p^2 - V^+(q)\big]_- \rd q \rd p 
 -  C  h^{-2+\e},
\end{multline}
where $ H_h^+ := T_h - V^+$. Since $V^+$ satisfies the same bound as $V$ 
(second bound  in \eqref{der} with $f=\ell=1$), we can drop the upper index + and
we will focus on the proof of \eqref{locsc1}.

We reallocate the magnetic energy and consider the infimum over $A$ for each $u$
separately. Reallocation changes the coefficient of the field energy
by a universal constant factor $c>0$  using the inequality
$$
  \int_{B(2)} |\nabla \otimes A|^2
  \ge c\int_{B(3/2)} \frac{\rd u}{L_0^3} \int_{B_u(2L_0)} |\nabla \otimes A|^2.
$$
Therefore, it is sufficient to prove that for each fixed $u\in B(3/2)$ 
and $c>0$ we have
\begin{multline}
 {\inf_A}'  \Bigg(\tr [\psi\phi_u  H_h(A)\phi_u \psi]_- + \frac{c}{\kappa h^2}
\int_{B_u(2L_0)} |\nabla \otimes A|^2\Bigg) \\ \ge 
  2(2\pi h)^{-3}\iint \psi(q)^2\phi_u^2\big[ p^2 -  V(q)\big]_- \rd q
  \rd p 
 -  C  h^{-2+\e} L_0^3,
\label{locsc2}
\end{multline}
and then integrating this inequality w.r.t. $\rd u/L_0^3$ over $u\in B(3/2)$,
 we will obtain \eqref{locsc1}. Recall that $\inf_A'$ denotes infimum over vector potentials
with $\nabla\cdot A= 0$ on $B(5/4)$, in particular they are divergence free in
a neighborhood of the support of $\psi$.

Defining $\phi = \psi \phi_u$,  Theorem \ref{thm:sc} will immediately follow from 
 the following
theorem

\begin{theorem}\label{thm:locsc} Let $L_0= h^{1/2-\e_0}$ with
a sufficiently small $\e_0>0$.
Let  $\phi$
be supported in $B(L_0)$, with $|\pt^n \phi|\le C_n L_0^{-|n|}$. 
Let the potential $V$ satisfy 
\be 
|\pt^n V|\le C_n \qquad \mbox{on}\quad \ov B(L_0), \qquad \mbox{and}\qquad
 V(0) \ge c_0
\ee
with some positive constants $C_n$ and $c_0$. 
Then for any
$\kappa>0$ we have
\begin{multline}
 {\inf_A}'  \Bigg(\tr [\phi  H_h(A)\phi]_- + \frac{1}{\kappa 
 h^2}\int_{B(2L_0)} |\nabla \otimes A|^2\Bigg) \\ \ge 
  2(2\pi h)^{-3}\iint \phi(q)^2\big[ p^2 -  V(q)\big]_- \rd q \rd p 
 -  Ch^{1+\e_0}\Lambda_{L_0}  
\label{locsc3}
\end{multline}
with a constant $C$ depending on $\kappa$, $C_n$ and $c_0$,
where we recall $\Lambda_{L_0}= h^{-3}L_0^3$.
Here $\inf_A'$ is taken for all vector potentials satisfying $\nabla\cdot A=0$ on
$B(5L_0/4)$.
\end{theorem}

In the next subsections of Section \ref{sec:sc}
we will prove Theorem \ref{thm:locsc}.
We first smooth out
the magnetic vector potential, $A$, on a scale $r\ll h^{1/2}$.
Theorem  \ref{thm:replace} in Section~\ref{sec:smooth}  states that
the error of this replacement is negligible.
This is the main technical result and it will be proven
in Section \ref{sec:repl}.
Once $A$ is replaced with a smoothed version $A_r$,
and we estimate the size of its derivatives,
 we can apply a  semiclassical result (Theorem~\ref{thm:ivrii}) 
 to evaluate the r.h.s. of \eqref{AtoAr}.
Combining Theorem  \ref{thm:replace} and Theorem \ref{thm:ivrii}
will then yield Theorem \ref{thm:locsc}.

\subsection{Replacing $A$ with a smooth version on $L_0$-box}\label{sec:smooth}

To prove  Theorem \ref{thm:locsc}, we first smooth out
the magnetic potential on a scale $r=h^{1/2+\varrho}\ll h^{1/2}$.
Using the lower bound from Theorem~\ref{thm:SS},
it is sufficient to consider magnetic vector potentials $A$ such that
$$
 \cE(A):=\tr [\phi  H_h(A)\phi]_- + \frac{1}{ \kappa 
 h^2}\int_{B(2L_0)} |\nabla \otimes A|^2 \le \cE(0).
$$
We can assume that $A$ has zero average, i.e.
\begin{align}\label{eq:gaugezero}
    \int_{B(2L_0)} A = 0,
\end{align}
since $\cE(A)=\cE(A-c)$ for any constant shift $c\in \bR^3$
in the vector potential. 
To see this, note that for any  cutoff function $\psi$, 
by the variational principle,
\begin{align}\label{gaugeinv}
  \tr [\psi (T_h(A) - V)\psi]_- & = \inf \big\{ \tr \psi (T_h(A) - V)\psi\gamma \; :
 \; 0\le \gamma\le 1  \big\} \non \\
& = \inf \big\{ \tr \psi (T_h(A-c) - V)\psi e^{-ic\cdot x/h}\gamma e^{ic\cdot x/h} \; :
 \; 0\le \gamma\le 1  \big\} \non \\
& = \inf \big\{ \tr \psi (T_h(A-c) - V)\psi \gamma  \; :
 \; 0\le \gamma\le 1  \big\}\non\\  & =  \tr [\psi (T_h(A-c) - V)\psi]_-.
\end{align}
In fact, $\tr [\psi (T_h(A) - V)\psi]_-$ is invariant under
any gauge transformation, $A\to A-\nabla \varphi$, $\varphi:\bR^3\to \bR$,
but $\int_{B(2L_0)} |\nabla \otimes A|^2$ is not.

Now we state the main result of this section:

\begin{theorem}\label{thm:replace}
Let $\kappa >0$ be given.
Let $L_0 = h^{1/2 - \e_0}$ for some sufficiently small $\e_0 >0$. Let $\phi$
with $\supp \phi \subset B(L_0)$ with $|\pt^n \phi|\le C_n L_0^{-|n|}$.
Assume that the potential $V$ satisfies
\begin{align}
\label{eq:lowerboundV}
|\pt^n V|\le C_n \quad \mbox{ on $\ov B(L_0)$} \qquad \mbox{and} \quad V(0)\ge c_0
\end{align}
with some positive constants $C_n$ and $c_0$.
Let $A$ be a vector potential such that $\cE(A)\le \cE(0)$ and
 satisfying \eqref{eq:gaugezero} and $\nabla\cdot A=0$ on $B(5L_0/4)$.
Then  for $\al:= 1-3\e_0$ 
we have
\begin{align}
  \label{eq:1}
  h^{-2} \int_{B(2L_0)} |\nabla \otimes A|^2  \leq Ch^\al  \Lambda_{L_0}\,.
\end{align}
Moreover, for $\rho>0$ sufficiently small there exists $\e_0>0$
 (in the definition of $L_0$) such that with $r := h^{1/2 + \rho}$ we have
\be\label{AtoAr}
\tr\big( \phi [T_h(A) - V ] \phi )_{-} + \frac{1}{\kappa h^2} \int_{B(2L_0)} 
| \nabla \otimes A|^2 
\geq
\tr\big(\phi   [T_h(A_{r}) - V ] \phi \big)_{-} - Ch^{1+\e_0} \Lambda_{L_0}.
\ee
Here $A_{r} = A * \chi_{r}$ is a smoothed out version of $A$ on the length scale $r$
and the constants in the estimates depend on 
the fixed unscaled cutoff function $\chi$, on $\kappa$, $C_n$ and $c_0$
in \eqref{eq:lowerboundV}.
\end{theorem}

We remark that this theorem involves only the
potential in $B(L_0)$. However, under the conditions \eqref{eq:lowerboundV}
one can extend $V$ to $\ov B(2L_0)$ with similar bounds on the derivative.
In the sequel we will thus assume that $V$ is defined in $B(2L_0)$ with
\be\label{eq:lowerboundVuj}
|\pt^n V|\le C_n \quad \mbox{ on $\ov B(2L_0)$} \qquad \mbox{and} \quad V(0)\ge c_0,
\ee
for some constants $C_n$ and $c_0$ that is determined by, but may differ from, 
the constants in \eqref{eq:lowerboundV} and with a slight abuse
of notations we will continue to use the
same letters.

Theorem~\ref{thm:replace} will be proved in Section~\ref{sec:repl} below.

\subsection{A precise semiclassical result}\label{sec:ivrii}

We state the following simplified version of \cite[Theorem~4.5.2]{I}
(see also \cite[Theorem~5.2.2]{I2}) for reference. Recall that $D=-ih\nabla$.

\begin{theorem}\label{thm:ivrii}
Suppose that $\phi$ is a localization function with
\begin{align}
  \label{eq:6}
  \supp \phi \subset B(1),
\end{align}
and $H = (D+A)^2 - V$ is a self-adjoint magnetic
Schr\"{o}dinger operator acting on $L^2(\bR^3)$ with
\begin{align}
  \label{eq:7}
  |\partial^{n} V| \leq c,\qquad |\partial^{n} A |\leq
  c,\qquad 
\text{ and } \quad |\partial^{n} \phi| \leq c
\end{align} 
for all multi-indices $n \in \bN^3$ with $|n|\leq K$.
Then, for all $h \in (0,1]$,
\begin{align}
  \label{eq:8}
  \Big|\tr_{L^2(\bR^3)}
 \phi^2 \big[(D+A)^2 - V \big]_{-} - (2\pi h)^{-3}\iint \phi^2(q) \big[ (p+A(q))^2 -
  V(q)\big]_{-} \,\rd q \rd p\Big| \leq C h^{-1}. 
\end{align}
Here $K$ is a universal constant and $C$ only depends on the constant $c$
in \eqref{eq:7}.

The same result holds for the Pauli operator $T_h(A)= [\bsigma\cdot (D+A)]^2$ as well
\begin{align}
  \label{eq:8Pauli}
  \Big|\tr
 \phi^2 \big[ T_h(A) - V \big]_{-} - 2 (2\pi h)^{-3}\iint \phi^2(q) \big[ (p+A(q))^2 -
  V(q)\big]_{-} \,\rd q \rd p\Big| \leq C h^{-1}. 
\end{align}
\end{theorem}

\begin{remark}
By a simple shift of variables $p\to p+A(q)$
for each fixed $q$, we obtain that
$$
   (2\pi h)^{-3  }\iint \phi^2(q) \big[ (p+A(q))^2 -
  V(q)\big]_{-} \,\rd q \rd p
 = (2\pi h)^{-3  }\iint \phi^2(q) \big[ p^2 -
  V(q)\big]_{-} \,\rd q \rd p,
$$
in other words, the presence of the magnetic vector potential plays no
role in the semiclassical formula.
\end{remark}

\begin{remark}
Theorem \ref{thm:ivrii} in Ivrii's book is formulated for Schr\"odinger operators,
as stated in \eqref{eq:8} above. For the purpose of semiclassical analysis,
the Pauli operator, written as $T_h(A)=(D+A)^2 + h\bsigma\cdot B - V$, can be considered
as a Schr\"odinger operator with $2\times 2$ matrix valued potential
that is subprincipal. Therefore the analysis of Ivrii goes through for
the Pauli case as well and this gives \eqref{eq:8Pauli}. Alternatively, one can apply Ivrii's result
for the Dirac operator \cite[Theorem 7.3.14]{I}
 or \cite[Theorem 5.2.23]{I2}
and use that in the large mass $m$ limit, the square of the Dirac
operator converges to the Pauli operator. Although not stated explicitly,
the error estimates in the cited theorems are uniform as $m\to \infty$.
In this way, taking the semiclassical limit followed by the large mass
limit, the semiclassical estimate for the Pauli operator can be
deduced from Ivrii's result on the Dirac operator.
\end{remark}

We also need a slight modification of Theorem~\ref{thm:ivrii} when
the localization function $\phi$ is inside the negative part:

\begin{theorem}\label{thm:ivrii2}
Suppose that $\phi$ is a localization function with
\begin{align}
  \label{eq:6-2}
  \supp \phi \subset B(1),
\end{align}
and $H = T_h(A) - V$ is a self-adjoint  Pauli operator acting 
or $L^2(\bR^3, \bC^2)$ with
\begin{align}
  \label{eq:7-2}
  |\partial^{n} V| \leq c,\qquad |\partial^{n} A |\leq
  c,\qquad 
\text{ and } \quad |\partial^{n} \phi| \leq c
\end{align} 
for all multi-indices $n \in \bN^3$ with $|n|\leq K$.
Then, for all $h>0$,
\begin{align}
  \label{eq:8-2}
  \Big|\tr
 \big[\phi \big(T_h(A) - V\big) \phi \big]_{-} - 2(2\pi h)^{-3}\iint \phi^2(q) \big[ p^2 -
  V(q)\big]_{-} \,\rd q \rd p\Big| \leq C h^{-1}. 
\end{align}
Here $K$ is a universal constant and $C$ only depends on the constant $c$
in \eqref{eq:7-2}. Similar statement holds for the magnetic Schr\"odinger operator.
\end{theorem}

\begin{proof}
Since this estimate is local, we may assume that $\supp V, \supp A \subseteq B(2)$. 
For a lower bound we estimate
\begin{align}\label{inoutt}
\tr \big[\phi \big(T_h(A) - V\big) \phi \big]_{-} \geq \tr \big[\phi \big(T_h(A) - V\big)_{-} \phi \big]_{-} 
= \tr \big[\phi \big(T_h(A) - V\big)_{-} \phi \big].
\end{align}
In order to prove the upper bound, we write $\gamma: = {\bf 1}_{(-\infty,0]}(H)$ and calculate,
\begin{align}
\tr \phi^2 [H]_{-} &= \frac{1}{2} \tr\big( \phi^2 [H]_{-} + [H]_{-} \phi^2 \big) 
= \frac{1}{2} \tr\big( (\phi^2 H + H \phi^2) \gamma \big)\nonumber \\
&= \frac{1}{2} \tr\big( ([H, \phi],\phi] + 2\phi H \phi) \gamma \big).
\end{align}
Therefore, by the variational principle,
\begin{align}
\tr\big[ \phi H \phi\big]_{-} \leq \tr \phi^2 [H]_{-} + h^2 \tr (\nabla \phi)^2 \gamma.
\end{align}
In order to estimate the last term we apply a Lieb-Thirring inequality to the operator 
$H - (\nabla \phi)^2$. Using $\tr H\gamma \le 0$, this yields 
\begin{align}
-   \tr (\nabla \phi)^2 \gamma& \geq 
\tr\big( H - (\nabla \phi)^2\big) \gamma \nonumber \\
& \geq - C   h^{-3} \int \big([V]_++(\nabla \phi)^2\big)^{5/2} 
- C\Big(h^{-2} \int |\nabla \times A|^2\Big)^{3/4}
 \Big( \int \big([V]_++(\nabla \phi)^2\big)^4 \Big)^{1/4}  \nonumber \\
&\geq -C h^{-3} - Ch^{-3/2}.
\end{align}
For Schr\"{o}dinger a similar Lieb-Thirring inequality holds but without the magnetic term.

Combining the upper and lower bounds the estimate \eqref{eq:8-2} follows from \eqref{eq:8}
in case of $h\le 1$. Finally, for $h\ge 1$ \eqref{eq:8-2} trivially holds since
the semiclassical integral is of order $h^{-3}$ which is smaller than $Ch^{-1}$.
The quantum energy can be bounded by the Lieb-Thirring inequality after using
\eqref{inoutt}:
$$
  0\ge \tr \big[\phi \big(T_h(A) - V\big) \phi \big]_{-} \ge \tr  \big(T_h(A) - V\big)_{-}
  \ge - Ch^{-3}- Ch^{-3/2}  
$$
for the Pauli case and the $h^{-3/2}$ term is absent for the Schr\"odinger case.
Both terms are smaller than the error bar $Ch^{-1}$ in \eqref{eq:8-2}. 
\end{proof}

\bigskip

Now we explain how  Theorem~\ref{thm:ivrii} can be applied to
prove Theorem~\ref{thm:locsc} from Theorem~\ref{thm:replace}.

\medskip

To estimate the right hand side of \eqref{AtoAr} from below, we
first move the localization outside the negative part by the simple inequality
\begin{align}
  \label{eq:5}
  \tr\big(\phi   [(T_h(A_{r}) - V ] \phi \big)_{-} \geq \tr \phi 
  [T_h(A_{r}) - V ]_{-}\phi =  \tr \phi^2 
  [T_h(A_{r}) - V ]_{-}.
\end{align}

By unitary scaling $x\to L_0x$, we scale the unit ball 
to the ball $B(L_0)$ and
 we have 
\begin{align}
\tr \phi^2   [T_h(A_{r}) - V ]_{-} =
\tr {\wt \phi}^2 
  [T_{h/L_0}(\wt A_{r}) - \wt V ]_{-},
\end{align}
where 
\begin{align}
\wt \phi(x) = \phi(L_0 x), \qquad \supp \wt \phi \subset B(1),
\end{align}
and
\begin{align}
\wt V(x) = V(L_0 x), \qquad \wt A_r(x) = A_r(L_0 x).
\end{align}
Notice that the semiclassical parameter has changed from $h$ to $h_{new}:=h/L_0
=h^{\frac{1}{2}+\e_0}$
which is much smaller than $h^{1/2}$.  Theorem \ref{thm:ivrii}
will be used for the data with tilde with $h_{new}$ and it provides
a precise semiclassical asymptotics with a relative  error term of order $Ch_{new}^2$
compared with the main term. In terms of the original $h$, this error is of
order  $C h^{1+2\e_0}$.

We now check that the derivative 
estimates  \eqref{eq:7} hold for $\wt V$, $\wt A_r$ and $\wt\phi$. 
Notice that we may
replace $V$ and $A_r$ by localized versions in the left hand side of
\eqref{eq:5}. Therefore, we only need to control the (finitely many)
derivatives in \eqref{eq:7} in the ball $B(3/2)$.

Clearly, the derivatives of $\wt \phi$ and $\wt V$ are bounded
since $L_0\le 1$. For
$\wt A_r$ we have the following estimate which will be proved 
as a consequence of Lemma \ref{lem:derivatives} in Section~\ref{sec:EstimatingAr}.

\begin{lemma}\label{lem:ScaledDerivatives}
Assume that $\int_{B(2L_0)} A =0$, then
the following estimate holds, for $|n|\geq 0$,
\begin{align}
\int_{B(3/2)} | \partial^{n} \wt A_{r}|^2 \, \leq C_n
\Big(\frac{L_0}{r}\Big)^{2\max(|n|-1,0)} L_0^{-1} \int_{B(2L_0)} | \nabla
\otimes A|^2\,.
\end{align}
\end{lemma}

Lemma~\ref{lem:ScaledDerivatives} combined with \eqref{eq:1}, yield
\begin{align}\label{eq:NumberOfDerivs}
\int_{B(3/2)} | \partial^{n} \wt A_{r}|^2 \,\leq C
 \Big(\frac{L_0}{r}\Big)^{2\max(|n|-1,0)} L_0^{2} h^{\alpha - 1}.
\end{align}
In order to use Theorem~\ref{thm:ivrii} we need uniform bounds on $K$ derivatives 
of $\wt A_r$. By Sobolev inequalities this corresponds to $K' > K+3/2$ $L^2$-derivatives. 
Inserting $|n| = K'$ and the definitions of $L_0, r$ and $\alpha$ in 
\eqref{eq:NumberOfDerivs}, we get 
\begin{align}\label{eq:NumberOfDerivs2}
\int_{B(3/2)} | \partial^{n} \wt A_{r}|^2 \,
\leq C h^{1 - 2(K'-1) \rho - (2K' +2) \e_0},
\end{align}
for all $n \in \bN^3$ with $|n| \leq K'$. Since $K'$ is universal, the
 right side of \eqref{eq:NumberOfDerivs2} is clearly bounded for 
sufficiently small $\rho$ and $\e_0$, i.e. if
\begin{align}
1 \geq  (2K'-2) \rho - (2K' +2) \e_0.
\end{align}
By Theorem~\ref{thm:ivrii}, applied to the Pauli operator,
 we get for such values of $\rho, \e_0$ that
\begin{align}
\tr \wt \phi^2 
  [T_{h/L_0}(\wt A_{r}) - \wt V ]_{-} &= 2(2\pi (h/L_0))^{-3}\iint
 \wt \phi^2(q) \big[ (p+\wt A_r(q))^2 -
  \wt V(q)\big]_{-} \,\rd q \rd p \! + \!{\mathcal O}( (h/L_0)^{-1} ) \nonumber \\
  &= 2(2\pi h)^{-3}\iint \phi^2(q) \big[ p^2 -
  V(q)\big]_{-} \,\rd q \rd p  + {\mathcal O}(h^{1+2\e_0} \Lambda_{L_0}).
\end{align}
Recalling that the error in \eqref{AtoAr} was ${\mathcal O}(h^{1+\e_0} \Lambda_{L_0})$
this completes the proof of  Theorem~\ref{thm:locsc}.

\qed

\section{Proof of  Theorem  \ref{thm:replace}}\label{sec:repl}

\subsection{Reduction to a constant potential}
\label{sec:red}

We use the notations and assumptions  of Theorem  \ref{thm:replace},
the extension of $V$ given in \eqref{eq:lowerboundVuj}, and assume that
$\cE(A)\le \cE(0)$. From \eqref{eq:lowerboundVuj} we can assume  $V\ge0$ 
 for $h$ sufficiently small.
Let $\chi$ be a nonnegative, smooth, symmetric, cutoff function on $\bR^3$
with $\int \chi^2 =1$ and $\supp \chi\subset B(1)$.
 We choose a new lengthscale $L_1\le \frac{1}{4}L_0$ 
such that $h\le L_1\leq h^{1/2} |\log h|^{-2}$.
For any $u\in \bR^3$ we denote
\be
   \chi_u(x)=\chi_{u, L_1}(x) =\chi \big( \frac{x-u}{L_1} \big),
   \qquad \int \frac{\rd u}{L^3_1}
  \chi_u^2(x) \equiv 1,
\label{chiint}
\ee
and  we will mostly omit the lengthscale $L_1$. In the applications, $L_1=h^{1/2+\e_0}$.

We first prove a rough lower bound on the left hand side of \eqref{AtoAr}.
This will eventually be used to get apriori bounds when we prove \eqref{AtoAr}.
Along the way, we also prove \eqref{eq:1}.

\begin{theorem}\label{thm:1}  Fix $\kappa >0$
and assume that the potential satisfies
\eqref{eq:lowerboundV}.  Let $\phi\in C_0^\infty(B(L_0))$
with $|\pt^n \phi|\le C_n L_0^{-|n|}$ and let $A$ be a vector potential
satisfying \eqref{eq:gaugezero} and $\nabla\cdot A=0$ on $B(5L_0/4)$.
For any $\al<1$ and $\delta>0$ there is a constant
$C_{\delta,\al}$ such that if $h\le h_\delta$, then
\begin{align}
  \tr \Big[ \phi [T_h(A)-V]\phi \Big]_-  &+ \frac{1}{ \kappa
  h^2}\int_{B(2L_0)} |\nabla \otimes A|^2 \nonumber \\
 \ge &
  \int\tr\Big[ \chi_u \phi \big[D^2-V(u) \big]_-\phi\chi_u\Big]\frac{\rd u}{L^3_1} 
 +(\kappa^{-1}- C\delta) h^{-2}\int_{B(2L_0)} |\nabla \otimes A|^2 \non\\
  & - \Big( C_{\delta,\al}h^{\al}+ Ch^2L^{-2}_1 + CL^2_1\Big)\Lambda_{L_0}.
\label{lower}
\end{align}
Moreover, we also have 
\begin{multline}
\inf_{A}\Bigg\{ \tr  \Big[ \phi [T_h(A)-V]\phi \Big]_- + \frac{1}{\kappa h^2}\int_{B(2L_0)}
 |\nabla \otimes A|^2\Bigg\} \\
\le  \int \tr \Big[ \chi_u \phi \big[D^2-V(u)\big]_-
 \phi\chi_u\Big] \frac{\rd u}{L_1^3} 
 + C\big(h^2L_1^{-2}+L^2_1\big) \Lambda_{L_0}.
\label{upper}
 \end{multline}
\end{theorem}

\begin{remark}\label{remark:6.2}
We have explicitly, for any $u\in\bR^3$,
$$
    \tr \Big[ \chi_u \phi \big[D^2-V(u)\big]_- \phi\chi_u\Big]
 = -\frac{2}{15\pi^2} h^{-3} [V(u)]^{5/2}\int \phi^2\chi_u^2
$$
 and thus
$$
    \int \tr \Big[ \chi_u \phi \big[D^2-V(u)\big]_-
 \phi\chi_u\Big] \rd u =  -\frac{2}{15\pi^2}h^{-3} \iint [V(u)]^{5/2}
 \chi_u^2(x)\phi(x) \rd x \rd u.
$$
After expansion
$$
    [V(u)]^{5/2} = [V(x)]^{5/2} + \nabla [V(x)]^{5/2} (x-u) + O( [x-u]^2)
$$
and since, for any $x\in\bR^3$,
$$
   \int \rd u \chi_u^2(x) (x-u) =0
$$
by symmetry, so we have
\be
    \int \tr \Big[ \chi_u \phi \big[D^2-V(u)\big]_-
 \phi\chi_u\Big] \frac{\rd u}{L^3_1} = -\frac{2}{15\pi^2}
 h^{-3}\int V^{5/2}\phi^2 + O(L_1^2h^{-3})\int \phi^2.
\label{eval}
\ee
\end{remark}

Moreover, we have the following upper bound on the magnetic energy:
\begin{corollary}\label{cor:BoundMagn}
Under the conditions of Theorem~\ref{thm:replace} we have
\begin{align*}
    \frac{1}{2\kappa h^2}\int_{B(2L_0)} |\nabla \otimes A|^2 \le 
\Big( C_{\delta,\al}h^{\al}+ Ch^2L^{-2}_1 + CL^2_1\Big)\Lambda_{L_0}
\end{align*}
for all sufficiently small $h$.
\end{corollary}

Corollary~\ref{cor:BoundMagn} is an immediate consequence
 of \eqref{lower} from
Theorem~\ref{thm:1} choosing $\delta$ sufficiently small,
 using $\cE(A)\le \cE(0)$ and Theorem~\ref{thm:ivrii2} (with rescaling
it to the ball of size $L_0$).
Choosing $L_1=h^{1/2+\e_0}$ and  $\al= 1-3\e_0$, we immediately obtain
\eqref{eq:1}, the first claim of Theorem \ref{thm:replace}.

\bigskip

{\it Proof of Theorem~\ref{thm:1}.} First we give the lower bound \eqref{lower}.
Recall that we can extend the potential $V$ from $\ov B(L_0)$ to $B(2L_0)$ by keeping 
similar derivative bounds as in \eqref{eq:lowerboundV}, see \eqref{eq:lowerboundVuj}.
By the IMS formula, and Taylor expansion
\be
  \Big| V(x) - V(u) - \nabla V(x) \cdot (x-u) \Big|\le C(x-u)^2
\label{Vexpansion}
\ee
for all $x\in \supp \phi$, $u\in \supp \phi + L_1\supp \chi$, we have
\begin{align}
  \phi [ T_h(A)&-V]\phi  =  \int \chi_u^2(x) \phi(x)
 [T_h(A)-V(x)]\phi(x)\frac{\rd u}{L_1^3}\non\\
        \ge & \int \chi_u(x) \phi(x) \Big[T_h(A)-V(x)-Ch^2L_1^{-2}\Big]\phi(x)
\chi_u(x) \frac{\rd u}{L_1^3} \non\\
    \ge & \int \chi_u(x) \phi(x) \Big[ T_h(A)-V(u) -\nabla V(x)\cdot(x-u) 
-Ch^2L_1^{-2}- CL_1^2
 \Big]\phi(x)\chi_u(x) \frac{\rd u}{L_1^3} \non \\ 
 = &  \int \chi_u(x) \phi(x) \Big[T_h(A)-V(u) -Ch^2L_1^{-2}- CL_1^2
 \Big]\phi(x)\chi_u(x) \frac{\rd u}{L_1^3} .
\end{align}
In the last step we used that
\be
   \int (x-u)\chi_u^2(x) \rd u =0.
\label{int0}
\ee
Using that $\tr [\sum_j A_j]_- \ge\sum_j \tr [A_j]_-$, we get 
\begin{align}\label{2}
 \tr \Big[ \phi [T_h(A) &-V]\phi \Big]_- \ge
  \int\tr\Big[ \chi_u \phi [T_h(A)-V(u) -Ch^2L_1^{-2}- CL_1^2
 ]\phi\chi_u\Big]_- \frac{\rd u}{L_1^3} .
\end{align}

In Theorem~\ref{thm:2} of the 
 next section, we will prove that $A$ can be neglected. 
In order to facilitate the estimate, it is convenient to
ensure that $A$ has  zero average on the ball $B_u(2L_1)$.
We define
\be
    c_u : = \int_{B_u(2L_1)} A.
\label{cu}
\ee
Similarly to the argument in \eqref{gaugeinv}, we can subtract $c_u$
from $A$ in \eqref{2} and we have
\begin{align}\label{2new}
 \tr \Big[ \phi [T_h(A) &-V]\phi \Big]_- \ge
  \int \tr\Big[ \chi_u \phi [T_h(A-c_u)-V(u) -Ch^2L_1^{-2}- CL_1^2
 ]\phi\chi_u\Big]_- \frac{\rd u}{L_1^3} .
\end{align}

The $u$-integration in \eqref{2new} can be restricted to $u\in B(\frac{5}{4}L_0)$ 
by the support properties of $\phi$ and $\chi_u$.
 For each fixed $u\in B(\frac{5}{4}L_0)$ we will use
Theorem \ref{thm:2} proven in the next section to estimate the
trace in the integrand. Define
\begin{align}
  \label{eq:10}
  V_u:=  V(u) + Ch^2L_1^{-2}+ CL_1^2.
\end{align}
By \eqref{eq:lowerboundVuj} we know that $c_0/2\le V_u\le C$.
With the
choice of $L=L_1$, $W= V_u$
and $\eta = \chi_u\phi$, we see from Theorem \ref{thm:2} that
 for any $\al<1$, $\e>0$ and
for any $\delta >0$ there is a constant $C_{\delta, \e}$ such that
\begin{align}\label{26}
 \tr\Big[ \chi_u \phi & [T_h(A-c_u)-V(u) -Ch^2L_1^{-2}- CL_1^2
 ]\phi\chi_u\Big]_- \non\\
 & \ge  \tr \chi_u\phi \big[D^2-V_u \big]_-\phi\chi_u - \delta V_u^{1/2} h^{-2}
\int_{B_u(2L_1)}|\nabla \otimes A|^2
 - C_{\delta,\e}h^\al V_u^{5/2}\Lambda_{L_1}
\end{align}
as long as
\be
h^{1-\frac{\al}{2}-\e}V_u^{-1/2}\le L_1 \le Ch^{1/2}|\log h|^{-2} V_u^{-1/2}
\label{hL}
\ee
and
\be
\label{Bapr}
    h^{-2}\int_{B_u(2L_1)} |\nabla \otimes A|^2 \le C\Lambda_{L_1}V_u^2
\ee
hold for some constant $C$. 
The constant $C_{\delta,\e}$ will also depend on the constants
in \eqref{hL} and \eqref{Bapr}.

Recalling $L_1=h^{1/2+\e_0}$,
with the choice $\al=1-3\e_0$ and  $\e<\e_0/2$, we see that \eqref{hL} is always satisfied
since $c_0/2\le V_u\le C$
 uniformly in $u\in B(2L_0)$.
To guarantee the second condition \eqref{Bapr} for the availability of the
 estimate \eqref{26},
we split the integral on the r.h.s. of \eqref{2new} as follows
\begin{align}
  \label{eq:9}
  \int_{B(5L_0/4)} \tr \big[ \ldots \big]_-\frac{\rd u}{L_1^3} 
  = &\int {\bf 1}\big\{u \in \Xi_<\big\} \tr\big[ \ldots \big]_-\frac{\rd u}{L_1^3} + \int {\bf 1}
\big\{u \in \Xi_>\big\} \tr\big[ \ldots \big]_-\frac{\rd u}{L_1^3},
\end{align}
where we defined
\begin{align}
   \Xi_< : & = \Big\{ u\in B\big(\frac{5}{4}L_0\big) \; : \;  
\delta h^{-2} \int_{B_u(2L_1)} |\nabla \otimes A|^2
    \leq  \Lambda_{L_1}V_u^2\Big\} \non\\
    \Xi_> : &= \Big\{ u\in B\big(\frac{5}{4}L_0\big) \; : 
 \; \delta h^{-2} \int_{B_u(2L_1)} |\nabla \otimes A|^2
    \ge \Lambda_{L_1}V_u^2\Big\}. \non
\end{align}

The estimate \eqref{26} will thus be available for $u$'s in the first integral
and it yields
\begin{align}\label{fir1}
\int {\bf 1}&\big\{u \in \Xi_<\big\} 
\tr \Big[ \chi_u \phi [T_h(A-c_u)-V(u) -Ch^2L_1^{-2}- CL_1^2
 ]\phi\chi_u\Big]_- \frac{\rd u}{L_1^3}  \non \\
&\ge   \int {\bf 1}\big\{u \in \Xi_<\big\}\Big( \tr \chi_u\phi \big[D^2-V_u \big]_-\phi\chi_u \Big)
\frac{\rd u}{L_1^3}
  - C\delta  h^{-2}
\int_{B(2L_0)}|\nabla \otimes A|^2
 - C_{\delta,\al}h^\al\Lambda_{L_1}.
\end{align}
With  an explicit calculation, for any $\beta>0$ constant and cutoff function $\eta$, we have
$$
   \tr \; \eta [D^2 - \beta]_- \eta = 2(2\pi)^{-3}\int \eta^2(x)\rd x  \int [ (hp)^2 -\beta]_-\rd p
  = \frac{2}{15\pi^2} h^{-3}\beta^{5/2} \int \eta^2
$$
(the additional factor 2 comes from the spin degeneracy).
Since $|V(u)|\le C$, $h\le L_1\le 1$ and $V_u>0$, we obtain
\be
   V_u^{5/2} \le [V(u)]_+^{5/2} +  C(h^2L_1^{-2} + L_1^2),
\label{5/2}
\ee
and thus
$$
 \tr \chi_u\phi \big[D^2-V_u \big]_-\phi\chi_u \ge \tr \chi_u\phi\big[D^2-V(u) \big]_-\phi\chi_u 
- C(h^2L_1^{-2} + L_1^2)\Lambda_{L_1}.
$$
Inserting this into \eqref{fir1}, we have the following estimate
for the first integral on  the r.h.s. of \eqref{eq:9}:
\begin{align}\label{firs}
\int {\bf 1}\big\{u \in \Xi_<\big\} \big[ \ldots \big]\frac{\rd u}{L_1^3}
\ge &   \int {\bf 1}\big\{u \in \Xi_<\big\}\Big( \tr \chi_u\phi \big[D^2-V(u) \big]_-\phi\chi_u \Big)
\frac{\rd u}{L_1^3} \non\\
&  - C\delta  h^{-2}
\int_{B(2L_0)}|\nabla \otimes A|^2
 - \big[C_{\delta,\al}h^\al+C h^2L_1^{-2}+ CL_1^2\big]\Lambda_{L_1}.
\end{align}

The traces in the second integral of \eqref{eq:9} are directly
estimated using a Lieb-Thirring bound. For this estimate we will have to localize
the magnetic field. Fix $u\in \Xi_>$ and choose a cutoff function $ \phi'$ with
$\supp \phi' \subset B_u(3L_1/2)$, $|\pt^n \phi'|\le C_n L_1^{-n}$ and
such that 
$\phi'\equiv 1$ on $B(L_1)$ and define, for the
purpose of this proof, 
\be
   A'= (A-c_u)\phi', \qquad V'(x) = V_u \cdot {\bf 1}\{x\in B_u(L_1)\}.
\label{A1prime}
\ee
Clearly 
\be\label{nablaA'}
  \int_{\bR^3}
 |\nabla\otimes A'|^2 \le \int_{B_u(2L_1)} |\nabla\otimes A|^2 + CL_1^{-2} \int_{B_u(2L_1)} 
|A-c_u|^2  \le C\int_{B_u(2L_1)} |\nabla\otimes A|^2
\ee
where in the second step we used the Poincar\'e inequality.

By the magnetic Lieb-Thirring inequality,
Theorem~\ref{thm:lls}, and \eqref{nablaA'}:
\begin{align}\label{longg}
   \tr \Big[ \chi_u \phi [T_h(A&-c_u)-V(u) -Ch^2L_1^{-2}- CL_1^2
 ]\phi\chi_u\Big]_- =  \tr \Big[ \chi_u \phi [T_h(A')-V']\phi\chi_u\Big]_- \non\\
  & \ge - Ch^{-3}\int_{B_u(L_1)} V_u^{5/2} - C\Big(h^{-2}\int_{\bR^3}
   |\nabla \times A'|^2\Big)^{3/4} \Big(\int_{B_u(L_1)} V_u^{4}\Big)^{1/4} \non\\
  &\ge  - \delta h^{-2}\int_{B_u(2L_1)}
   |\nabla \otimes A|^2
- C(V_u^{5/2}+ h^3\delta^{-3}V_u^4) \Lambda_{L_1} \non\\
 &\ge  - C\delta h^{-2}\int_{B_u(2L_1)}
   |\nabla \otimes A|^2,
\end{align}
using that $V_u\le C$ and $h\le h_\delta$. In the last step we used
$u\in \Xi_>$ to estimate $V_u^2\Lambda_{L_1}$ by the local magnetic field energy.
Integrating out this inequality for $u\in B(5L_0/4)$, we obtain the following
bound for the second integral  on the r.h.s in \eqref{eq:9}
\be\label{secs}
\int {\bf 1}\big\{u \in \Xi_>\big\} \big[ \ldots \big]\frac{\rd u}{L_1^3}
\ge  - C\delta h^{-2}\int_{B(2L_0)}
   |\nabla \otimes A|^2.
\ee

Finally, the missing piece of the non-magnetic term on the r.h.s of \eqref{firs}
for $u\in \Xi_>$ can be estimated by the standard Lieb-Thirring similarly to
\eqref{longg} and \eqref{secs}
\be\label{miss}
  \int {\bf 1}(u\in \Xi_>) \tr \Big[ \chi_u \phi [D^2-V_u ]\phi\chi_u\Big]_-
  \frac{\rd u}{L_1^3}
  \ge  - C\delta h^{-2}\int_{B_u(2L_0)}
   |\nabla \otimes A|^2.
\ee
The estimates \eqref{firs}, \eqref{secs} and \eqref{miss} inserted into 
\eqref{eq:9} and \eqref{2new}
 complete the proof of \eqref{lower}.

\medskip

For the proof of \eqref{upper}, we define the spectral projection
$$ 
   \gamma_u : = {\bf 1}_{(-\infty,0]} \big(D^2 - V(u) + C(h^2L_1^{-2}+L_1^2)\big)
$$
and
$$
    \gamma : = \int \chi_u \gamma_u \chi_u \frac{\rd u}{L_1^3}.
$$
Note that $0\le \gamma \le 1$ by \eqref{chiint}. 
We will also use that the density of $\gamma_u$ is given by
$$
  \varrho_{\gamma_u}(x): =\gamma_u(x,x)= \frac{4\pi}{3}h^{-3} 
\big[ V(u) - C(h^{2}L_1^{-2} +L_1^2)\big]_+^{3/2} .
$$

Then, by Taylor expanding $V$ up to second order, similarly as in \eqref{Vexpansion}
but using $\nabla V(u)$ instead of $\nabla V(x)$,
 we have
\begin{align}
\inf_{A}\Bigg\{ \tr & \Big[ \phi [T_h(A)-V]\phi \Big]_-
 + \frac{1}{\kappa h^2}\int |\nabla \otimes A|^2\Bigg\}
 \le  \tr \Big[ \phi [D^2-V]\phi \Big]_-  \non\\
 \le & \tr  \phi [D^2-V]\phi \gamma \non\\
= &  \int \tr \Big( \chi_u \phi [D^2-V]\phi\chi_u \gamma_u \Big) \frac{\rd u}{L_1^3} \non\\
\le & \int \tr \Big( \chi_u \phi [D^2-V(u) -\nabla V(u)\cdot(x-u) + CL_1^2]
 \phi\chi_u \gamma_u \Big) \frac{\rd u}{L_1^3} \non\\
 = & \int \tr \Big( \chi_u \phi [D^2-V(u)  + CL_1^2]
 \phi\chi_u \gamma_u \Big) \frac{\rd u}{L_1^3} - \int\int \varrho_{\gamma_u}(x) \nabla V(u)\cdot (x-u) 
 \phi^2(x) \chi_u^2(x) \frac{\rd u}{L_1^3}\rd x \non\\
\le & \int \tr \Big( \chi_u^2 \phi^2 [D^2-V(u)  + Ch^2L_1^{-2}+CL_1^2]
 \gamma_u \Big)  \frac{\rd u}{L_1^3} \non \\
 &- c_1h^{-3}  \int\int \big[ V(u) - C(h^{2}L_1^{-2} +L_1^2)\big]_+^{3/2}
    \nabla V(u) \cdot(x-u)  \phi^2(x) \chi_u^2(x)  \frac{\rd u}{L_1^3}\rd x \non\\
\le & \int \tr \Big( \chi_u \phi [D^2-V(u)  + Ch^2L_1^{-2}+CL_1^2]_-
 \phi\chi_u\Big)  \frac{\rd u}{L_1^3}  + CL_1^2 h^{-3}\int \phi^2\non \\
\le & \int \tr \Big( \chi_u \phi [D^2-V(u)]_-
 \phi\chi_u\Big)  \frac{\rd u}{L_1^3} + C(h^2L_1^{-2}+L_1^2) h^{-3}\int \phi^2.
 \end{align}
In the last but one step we used \eqref{int0} and  that
$$
\big[ V(u) - C(h^{2}L_1^{-2} +L_1^2)\big]_+^{3/2}\nabla V(u) = 
\big[ V(x) - C(h^{2}L_1^{-2} +L_1^2)\big]_+^{3/2}\nabla V(x) + O(L_1)
$$
if $|x-u|\le CL_1$. In the last step we used an estimate
similar to \eqref{5/2}. This proves \eqref{upper} and completes the proof
of Theorem \ref{thm:1}, assuming Theorem \ref{thm:2} to be proven in
the next section. \qed

\subsection{Removing $A$}\label{sec:remA}

In this section we estimate the effect of
removing the vector potential from the operator
$H(A)= T_h(A)-W$ with a constant potential $W>0$
on a scale $L$. In the applications,  $L$ will be $L_1=h^{1/2+\e_0}$.
Let $H_0= H(A=0)=D^2-W$ and define the free
 projection 
\begin{align}\label{eq:freeP}
  P = {\bf 1}_{(-\infty,0]}(H_0).
\end{align}
We also set $\Lambda_L:= h^{-3} L^3$.

\begin{theorem}
\label{thm:2} Given $0\le \al<1$ and $0<\e< \frac{1}{2}(1-\al)$, let
 $L$ be a lengthscale that satisfies 
\be
 h^{1-\frac{\al}{2}-\e} W^{-1/2}\le L \le   Ch^{1/2}(|\log h|)^{-2} W^{-1/2}.\label {Lh}
\ee
Let $\eta$ be a smooth cutoff function supported on $B(L)$
with $|\pt^n \eta |\le C_n L^{-n}$.
Let $A \in H^1_{\rm loc}(\bR^3)$ be a vector potential. 
We will assume the gauge to be such that $\nabla\cdot A=0$ on $B(5L/4)$
and
\begin{align}
  \label{eq:3}
  \int_{B(2L)} A = 0.
\end{align}
We also assume that
\be
    h^{-2}W^{1/2}\int_{B(2L)} |\nabla \otimes A|^2 \le C W^{5/2}\Lambda_L.
 \label{eq:1new}
\ee
Then for any $\delta>0$  there exists $C_{\delta,  \e}$, depending
also on the constants in \eqref{Lh} and \eqref{eq:1new}, such that
for any density matrix $0\le \gamma\le 1$ 
we have
\be \label{eq:2alt}
 \tr H(A) \eta \gamma \eta  = \tr [T_h(A)-W] \eta \gamma \eta
\ge \tr \eta [H_0]_-\eta-
  C_{\delta,\e} h^\al W^{5/2}\Lambda_L
 - C\delta h^{-2}W^{1/2} \int_{B(2L)} |\nabla \otimes A|^2.
\ee

\end{theorem}

\begin{remark}\label{rem:MagLocEstimate}
The gauge choice \eqref{eq:3} implies that one can use the
Poincar\'{e} inequality on $B(2L)$ to conclude that
\begin{align}
  \label{eq:4}
  \int_{B(2L)} A^2 \leq C L^2 \int_{B(2L)} |\nabla \otimes A|^2 \, .
\end{align}
\end{remark}

{\it Proof of Theorem~\ref{thm:2}.}
We start with localizing the vector potential which will be used later.
Let $\phi'$ be a standard localization function on $B(5L/4)$, i.e.
 $\phi' \equiv 1$ on $B(5L/4)$,
 $\supp \phi' \subset B(3L/2)$, and define 
\begin{align}\label{def:Aprim}
A' := \phi' A,
\end{align} 
in particular $\nabla\cdot A'=0$ on $B(5L/4)$.
Note that this definition of $A'$ is different from \eqref{A1prime}
used in the proof of Theorem \ref{thm:1} in Section \ref{sec:red};
the prime notation will always indicate a trivial cutoff outside of
the appropriate domain we actually work on. In Section \ref{sec:remA}
we will use \eqref{def:Aprim}.

We have  $\eta A = \eta A'$,
so the sole purpose of this modification is 
 to guarantee that only the local part of $A$ will
be taken into account. The Poincar\'e  inequality \eqref{eq:4} remains valid
in the form
\be
   \int [A']^2 \le \int_{B(2L)} A^2 \leq CL^2\int_{B(2L)} |\nabla \otimes A|^2.
\label{Aprime}
\ee

\subsubsection{Decomposition in energy space}\label{sec:decomp}

We introduce a dyadic decomposition around the non-magnetic Fermi
surface using cutoff functions in energy space. The final result of
this section is given in  Lemma \ref{lm:decomp}.

Let $\chi_i$, $i>0$, be  smooth cutoff functions such that $\supp \chi_i 
\subset [\frac{3}{4}\cdot 2^{i-1}, \frac{5}{4}\cdot 2^{i}]$,
$\chi_i\le 1$, $\chi_i(t)=1$ for $t\in  [\frac{5}{4}\cdot  2^{i-1}, \frac{3}{4}
 \cdot 2^{i}]$, $|\nabla \chi_i| \le C\cdot 2^{-i}$
and
$$
   \sum_{i\ge 1} \chi_i^2(t) \equiv 1,\qquad \forall t>2.
$$
Define cutoff functions on $\bR^3$ by
$$
   f_i(u) = \chi_i \Big( \frac{u^2-W}{wW} \Big) \quad \mbox{for} \quad i>0, \quad u\in\bR^3
$$
and
$$
   f_i(u) = \chi_{|i|} \Big( -\frac{u^2-W}{ wW} \Big) \quad\mbox{for} \quad i<0, \quad u\in\bR^3,
$$
with some $w$ such that $h \le w \le 1$.
Setting $i_0: = [|\log _2 w|]+1$, where $[\cdot ]$ denotes the integer part, we 
clearly have $f_i\equiv 0$ if $i < -i_0$.  Define $w_i := 2^{|i|}w$, then
$f_i$ is supported in a spherical shell of volume $Cw_iW^{3/2}$ for $|i|\le i_0$.

Clearly
$$
   \sum_{i>0} f_i^2 + \sum_{i<0} f_i^2 \le 1
$$
so we can define 
$$
   f_0(u) = \chi_0 \Big( \frac{u^2-W}{ wW} \Big)
$$
with an appropriate cutoff function $\chi_0$, with $0\le \chi_0\le 1$, $|\nabla\chi_0|\le C$,
 so that
$$
  \sum_{i\in \bZ} f_i^2 \equiv 1,
$$
i.e.,  $f_0(u)$ is supported in the regime where $|u^2-W|\le \frac{5}{4}wW$
 and $f_0(u)\equiv 1$ where $|u^2-W|\le \frac{3}{4}wW$.
Note that 
$$
    f_if_j \equiv 0, \quad \mbox{if} \quad |i-j|\ge 2.
$$
We also define $f_>$ by
$$
   f_>^2 : = \sum_{i> i_0} f_i^2.
$$
We note that  the support of $f_>(t)$ lies entirely in the regime $|t|\ge 2$  and
\be
\label{f>}
    |\nabla f_>| \le C, \qquad |\mbox{supp}(\nabla f_>)| \le C.
\ee

For each $i$ with $0<|i|\le i_0$ we also define 
enlarged cutoff functions $\wt f_i$ by
$$
  \wt  f_i(u) = \wt \chi_i \Big( \frac{u^2-W}{wW} \Big) \quad \mbox{for} \quad i>0
$$
and
$$
   \wt f_i(u)  =\wt \chi_{|i|} \Big( -\frac{u^2-W}{ wW} \Big) \quad\mbox{for} \quad i<0,
$$
where  $\wt\chi_i$, $i>0$, are  cutoff functions such that $\supp \wt \chi_i 
\subset [\frac{1}{2}\cdot 2^{i-1}, \frac{3}{2}\cdot 2^{i}]$,
$\chi_i\le 1$, $\wt\chi_i(t)=1$ for $t\in  [\frac{3}{4}\cdot  2^{i-1}, \frac{5}{4}
 \cdot 2^{i}]$ and  $|\nabla \wt\chi_i| \le C\cdot 2^{-i}$.
We also set
$$
    \wt f_0(u) = \wt \chi_0 \Big(  \frac{u^2-W}{wW} \Big),
$$
where $\wt\chi_0 \le 1$, $\wt\chi_0(t) \equiv 1$ for $|t|\le 2$, $\supp \wt\chi_0 \subset [-3,3]$
and $|\nabla\wt\chi_0|\le C$. 
We can similarly define 
  $\wt f_>$ by
$$
   \wt f_>^2 : = \sum_{i\ge i_0} f_i^2.
$$
Note that $\wt \chi_i \equiv 1$ on the support of $\chi_i$, therefore we have
\be
     f_i \wt f_i = f_i, \qquad 0\le |i|\le i_0, \qquad \mbox{and}
  \quad f_> \wt f_> = f_>.
\label{Fi}
\ee
These extended functions clearly satisfy
\be\label{overlap1} 
\wt f_i \le f_{i-1} + f_i + f_{i+1}\quad \mbox{for}  \quad |i|<i_0,
\ee
\be
\wt f_{i_0} \le  f_{i_0-1} +   f_{i_0} + f_>
\label{overlap}
\ee
and $0\le \wt f_i \le 1$, $0\le \wt f_> \le 1$. 

Finally we define the  momentum cutoff operators by
\begin{align}
  \label{eq:2}
     F_i: = f_i(D),\quad F_>: = f_>(D), \quad 
 \wt F_i: = \wt f_i(D),\quad \wt F_>: = \wt f_>(D),
\end{align}
then  all inequalities \eqref{Fi}, \eqref{overlap1} and \eqref{overlap}
are clearly satisfied as operator inequalities if the functions $f$ are replaced with
the operators $F$.

\bigskip

In the first step, we neglect the positive $A^2$ term:
$$
  \tr H(A) \eta \gamma \eta  \ge \tr\wt H(A) \eta \gamma \eta 
$$
with
\be 
 \qquad \wt H(A) : = H_0+ \sigma(D)\sigma(A)+\sigma(A)\sigma(D) = H_0 + 2DA + h\sigma(B)
\label{tildeH}
\ee
using $\nabla\cdot A=0$ on the support of $\eta$ and  $\nabla \times A=B$ in the last identity
and introducing the shorter notation $\sigma(v)=\bsigma\cdot v$ for any vector
$v$. Note that 
the formula \eqref{tildeH} holds
if the kinetic energy $T_h(A)$ is the Pauli operator, for the Schr\"odinger
case we would have the simpler formula
$$
  \qquad \wt H(A) : = H_0+ DA+AD = H_0 + 2DA .
$$
Here we adopted the convention that $DA = \sum_{j=1}^3 (-ih\nabla_j)A_j$, $AD =
\sum_{j=1}^3 A_j(-ih\nabla_j)$ and $D\cdot A =\sum_{j=1}^3 -ih(\nabla_j A_j)
 = -ih\, \mbox{div} A$, in particular, $[D,A]= D\cdot A$. 

\begin{remark}
We remark that neglecting $A^2$  is affordable since $\gamma$ is close to
the projection $P$, defined in \eqref{eq:freeP}, and
the density of the projection $\varrho_P(x)$ is bounded by $Ch^{-3}W^{3/2}$, thus
\begin{multline*}
 \tr A^2 \eta\gamma \eta \approx 
  \int A^2\eta^2 \varrho_P = ch^{-3}W^{3/2} \int A^2\eta^2 \le ch^{-3} L^2W^{3/2}
  \int_{ B(2L)} |\nabla \otimes A|^2  \\
= (h^{-1}L^2W)
h^{-2}W^{1/2}\int_{ B(2L)} |\nabla \otimes A|^2
\ll h^{-2}W^{1/2}\int_{ B(2L)} |\nabla \otimes A|^2 
\end{multline*}
by the Poincar\'e inequality \eqref{eq:4} and  $L\ll h^{1/2}W^{-1/2}$
(see \eqref{Lh}).  
\end{remark}
By the reality of the projection $P$, and since it acts like a scalar in spinor space, we have
\be
   \tr \big[ \sigma(D)\sigma(A)+\sigma(A)\sigma(D) + \sigma(B)\big] \eta P\eta =
  \tr (DA+AD)\eta P\eta =0.
\label{reality}
\ee
Thus we get
\begin{align}\label{bigsplit1}
\tr H(A) \eta \gamma \eta  
\geq  &
\tr \wt H(A) \eta \gamma \eta -  
 \tr \big[ \sigma(D)\sigma(A)+\sigma(A)\sigma(D)\big] \eta P\eta\nonumber\\
= & \tr F_0^2H_0 \eta \gamma \eta + 
\tr F_0^2\big[ \sigma(D)\sigma(A)+\sigma(A)\sigma(D)\big]\eta(\gamma-P) \eta 
+\tr (1-F_0^2) H_0 \eta P\eta 
\non\\
&+\sum_{i<0} \tr  F_i^2 \wt H(A) \eta (\gamma-1) \eta 
 - \sum_{i<0} \tr F_i^2 \wt H(A)\eta (P-1) \eta  \non\\
  & +\sum_{i>0} \tr F_i^2 \wt H(A) \eta \gamma \eta  
 -\sum_{i>0} \tr  F_i^2 \wt H(A)\eta P \eta.
\end{align}
The combination of the first and third term on the r.h.s  of \eqref{bigsplit1}
gives the main
term in \eqref{eq:2alt}:
\begin{align}
    \tr F_0 H_0 F_0 \eta \gamma
  \eta  + \tr H_0 (1- F_0^2) \eta P \eta 
\ge & \tr \big[ \eta F_0 H_0 F_0 \eta\big]_- +  \tr \big[\eta (1-F_0^2)^{1/2}
   H_0 (1-F_0^2)^{1/2} \eta\big]_-  \non\\
 \ge & \tr \eta \Big( \big[ F_0 H_0 F_0 \big]_- +   \big[ (1-F_0^2)^{1/2}
   H_0 (1-F_0^2)^{1/2}\big]_- \Big) \eta  \non \\
 = & \tr \eta [H_0]_-\eta.  
\end{align}
Notice that the sum on the RHS of \eqref{bigsplit1} is real though each individual term
might not be.
 To symmetrize some of the terms we take the real part and we
have proved
\begin{lemma}\label{lm:decomp} 
With the notations introduced above, we have
\begin{align}\label{bigsplit}
\tr H(A) \eta \gamma \eta  
\geq & \tr \eta [H_0]_-\eta + 
\Re \tr F_0^2\big[ \sigma(D)\sigma(A)+\sigma(A)\sigma(D)\big]\eta(\gamma-P) \eta 
\non\\
&+\sum_{i<0} \Re \tr  F_i^2 \wt H(A) \eta (\gamma-1) \eta 
 - \sum_{i<0} \Re \tr F_i^2 \wt H(A)\eta (P-1) \eta  \non\\
  & +\sum_{i>0} \Re \tr F_i^2 \wt H(A) \eta \gamma \eta  
 -\sum_{i>0} \Re \tr  F_i^2 \wt H(A)\eta P \eta.  \qquad \qed
\end{align}
\end{lemma}
The main term is the first on the r.h.s of \eqref{bigsplit}.
In the next subsections we  estimate the other terms
and we show that they can be included in  the negative
 error terms in \eqref{eq:2alt}.

\subsubsection{Estimate of the $\sigma(D)\sigma(A)+\sigma(A)\sigma(D)$ term in 
\eqref{bigsplit}}

\begin{lemma}\label{lm:ad} For any $\delta>0$ we have
\begin{align}
   \Big|\tr F_0^2\big[& \sigma(D)\sigma(A)+\sigma(A)\sigma(D)\big]
\eta(\gamma-P) \eta\Big| \non\\
&\le  C\delta^{-1} \big[ h + wLW^{1/2} +h^{-1}w^2L^2W\big]W^{5/2}  \Lambda_L
 + \delta h^{-2}W^{1/2} \int_{B(2L)} |\nabla \otimes A|^2.
\end{align}
\end{lemma} 

{\it Proof.}
Using $\nabla\cdot A = 0$ on $B(5L/4)$ and $\eta A= \eta A'$ (recall the definition of $A'$
from \eqref{def:Aprim}) and the locality of the operator $D$, we have
\be\label{gg}
F_0^2\big[ \sigma(D)\sigma(A)+\sigma(A)\sigma(D)\big]\eta(\gamma-P) \eta 
= 2 F_0^2DA'\eta(\gamma-P) \eta + hF_0^2\sigma(B')\eta(\gamma-P) \eta
\ee
with $B' = \nabla\times A'$.  Note that
\be
   \int [B']^2 \le CL^{-2}\int_{B(2L)} A^2 + \int_{\bR^3} [\phi']^2 |\nabla\times A|^2
 \le C\int_{B(2L)} |\nabla \otimes A|^2 
\label{nablaA1}
\ee
using \eqref{eq:3} and the Poincar\'e inequality.

We will estimate $|\tr F_0^2DA'\eta\gamma\eta|$ 
and $|h\tr F_0^2\sigma(B')\eta\gamma\eta|$ for any density
matrix $0\le\gamma\le 1$, and then the same 
estimate can be applied to $\gamma$
replaced with $P$ as well. 
We start with the second (magnetic) term in \eqref{gg},
 and in fact we prove the following stronger
estimate that will be used later
\begin{align}\label{Bterm}
\sum_{|i|\le i_0} h |\tr F_i^2\sigma(B') \eta \gamma \eta|
 &\le \sum_{|i|\le i_0} h \Big(\tr F_i^2 [B']^2 \Big)^{1/2}
 \Big(\tr F_i^2 \eta \gamma 
\eta\Big)^{1/2} \non\\
&\le\sum_{|i|\le i_0} Ch^{-2} w_i W^{3/2} \Big( \int
 [B']^2\Big)^{1/2} \Big( \int \eta^2 \Big)^{1/2} \non
\\
&\le \delta h^{-2}W^{1/2}
 \int_{B(2L)} |\nabla \otimes A|^2 + C\delta^{-1} h W^{5/2}\Lambda_L.
\end{align}
Here we used \eqref{nablaA1} and that the diagonal element
 of $F_i$, $|i|\le i_0$, is bounded by
\be\label{diag}
  \sup_x F_i(x,x) \le \int_{\bR^3} f_i\big( hp\big)\rd p \le Ch^{-3} w_i W^{3/2}
\ee
since the support of $f_i$ is a spherical shell of width $C w_i W^{1/2}$ and radius of
order $W^{1/2}$, where we recall $w_i = 2^{|i|}w$.
In particular, \eqref{Bterm} shows that the magnetic term in \eqref{gg}
can be estimated as required in Lemma \ref{lm:ad}.

\medskip
For the first term in \eqref{gg}, with a general density matrix $\gamma$,
we insert $\wt F_0$ and $\wt G_0 : = 1- \wt F_0$ to get
\be\label{tildesplit}
   |\tr F_0^2DA'\eta\gamma\eta| \le
 |\tr F_0^2DA'\wt F_0 \eta\gamma\eta \wt F_0|+
|\tr F_0^2A' \wt G_0 \eta\gamma\eta \wt F_0D|,
\ee
where we used \eqref{Fi} and the cyclicity of the
trace. We will estimate these two terms separately.

The  first term in \eqref{tildesplit} is estimated as
\begin{align}\label{1st}
 |\tr F_0^2DA'\wt F_0 \eta\gamma\eta \wt F_0| &\le 
 \Big( \tr [A']^2\wt F_0 \eta\gamma\eta \wt F_0\Big)^{1/2}
 \Big( \tr F_0^4D^2\wt F_0 \eta\gamma\eta \wt F_0\Big)^{1/2}\non \\
  &\le C \Big( \tr [A']^2\wt F_0^2\Big)^{1/2}
 \Big(W \tr F_0^4\eta^2\Big)^{1/2} \non \\
 &\le Ch^{-3} w W^2 \Big(  \int [A']^2 \Big)^{1/2} \Big( \int \eta^2\Big)^{1/2} \non \\
  &\le \delta^{-1}h^{-1} w^2 L^2W^{7/2}\Lambda_L+ \delta  h^{-2}
 W^{1/2}\int_{B(2L)} |\nabla \otimes A|^2 .
\end{align}
Here we first used that $D^2$ is bounded by $CW$ when multiplied by $F_0$.
Then, similarly to \eqref{diag}, we estimated
 the diagonal element of $F_0$ and $\wt F_0$ as
\be
   \sup_x  F_0(x,x)\le \sup_x \wt F_0(x,x) \le Ch^{-3} w W^{3/2},
\label{F0x}
\ee
and finally we used \eqref{Aprime}.

The  second term in \eqref{tildesplit} is estimated as
\begin{align}\label{2st}
 |\tr F_0^2A'\wt G_0 \eta\gamma\eta \wt F_0D| &\le 
 \Big( \tr  F_0^2A'\wt G_0^2A' F_0^2\Big)^{1/2}
 \Big( \tr D^2\wt F_0^2 (\eta\gamma\eta)^2 \Big)^{1/2} .
\end{align}
For the second factor we can use the previous bound
\be
  \tr D^2\wt F_0^2 (\eta\gamma\eta)^2  \le CW \tr \wt F_0^2 \eta^2 \le C h^{-3} w
 W^{5/2}L^3  = CwW^{5/2}\Lambda_L.
\label{Ddens}
\ee
To estimate the first factor, we will need the
 following lemma,  whose proof is postponed:
\begin{lemma}\label{lem:comm2} Let $f$ and $g$ be positive real
 functions on $\bR^3$ such that $fg=0$.
Then for any real function $a\in H^1(\bR^3)$ we have
\begin{align}
\tr_{L^2(\bR^3)}[ f(D) a(x) g(D) a(x) f(D)] \leq h^{-2} \|f\|_{\infty} \|g\|_\infty
  \|\nabla f\|_1 \| a\|_2 \| \nabla \otimes a\|_2\,.
\end{align}
\end{lemma}
Using that  $F_0 = f_0(D)$ with a function $f_0$ satisfying
$|\nabla f_0|\le Cw^{-1}W^{-1/2}$, $|\supp f_0|\le CwW^{3/2}$, thus $\|\nabla f_0\|_1\le CW$, 
we obtain
\be
 \tr  F_0^2A'\wt G_0^2A' F_0^2 \le Ch^{-2}W\| A'\|_2 \|\nabla \otimes A'\|_2 
 \le Ch^{-2}LW \int_{B(2L)} |\nabla \otimes A|^2.
\label{appl}
\ee
Here we used \eqref{Aprime} and the second inequality in  \eqref{nablaA1}.

Combining these estimates and separating the two factors
in \eqref{2st} by a Schwarz inequality we obtain the following bound
\begin{align}
 |\tr F_0^2A'\wt G_0 \eta\gamma\eta \wt F_0D| &\le 
 C \delta^{-1}wLW^3\Lambda_L +\delta h^{-2}W^{1/2}\int_{B(2L)} |\nabla\otimes A|^2  .
\end{align}
This completes the proof of Lemma \ref{lm:ad}. \qed

\bigskip

{\it Proof of Lemma \ref{lem:comm2}.}
By passing to Fourier space we get
\begin{align}
\tr_{L^2(\bR^3)}[ f(D) a(x) g(D) a(x) f(D)] = \iint_{\bR^3\times\bR^3}
 f(hp) \hat a(p-q) g(hq) \hat a(q-p) f(hp) \,\rd p\rd q.
\end{align}
We now use that $fg = 0$ to rewrite the above integral as
\begin{align}
&= \iint [f(hp) - f(hq)] |\hat a(p-q)|^2 g(hq) f(hp) \,\rd p \rd q \nonumber \\
&\leq \|f\|_{\infty} \|g\|_\infty
\iint \int_0^1 |h(q-p)\cdot\nabla f (hp + th(q-p))|\,\rd t |\hat a(p-q)|^2\,\rd p \rd q
 \nonumber\\
&= h^{-2}  \|f\|_{\infty} \|g\|_\infty \int_{\bR^3} |\nabla f(v)|\,\rd v \int_{\bR^3} |z| |\hat a(z)|^2\,\rd z.
\end{align}
The result now follows upon passing back in $x$-space in the integral over $\hat a$. 
\qed

\subsubsection{Error terms in  \eqref{bigsplit} for $i> i_0$}

 We now deal with the two sums in the last line of \eqref{bigsplit}
in the regime where $i> i_0$.

\begin{lemma}\label{lm:largeenergy}
Under the conditions  $LW^{1/2} \ll h^{1/2}$  and assuming
\eqref{eq:1new},
 we have for any fixed $\delta>0$
\be
  \Re \sum_{i> i_0} \tr  F_i^2 \wt H(A) \eta \gamma \eta 
 \ge \frac{1}{2} \tr[ D^2 F_>^2 \eta \gamma \eta]- C\delta h^{-2}W^{1/2}
\int_{B(2L)} |\nabla \otimes A|^2 
 - \delta \tr[ D^2 \wt F_>^2 \eta \gamma \eta]
\label{largeenergygamma}
\ee
if $h\le h_\delta$. Moreover,
for any $N\ge 1$ and $h\le h_\delta,$ we also have
\be
  \Big|\sum_{i> i_0} \tr  F_i^2 \wt H(A) \eta P \eta \Big|
 \le C\delta h^{-2}W^{1/2}\int_{B(2L)} |\nabla \otimes A|^2 + 
C_{N} W^{5/2}\Big(\frac{h}{LW^{1/2}}\Big)^N.
\label{largeenergyP}
\ee
\end{lemma}

{\it Proof.}
Using the first formula in \eqref{tildeH}, we
 write
\begin{align}\label{newad}
 \Re \sum_{i> i_0} \tr F_i^2 \wt H(A) \eta \gamma \eta 
  &=  \Re \tr F_>^2 \big[ H_0+\sigma(D)\sigma(A)+\sigma(A)\sigma(D)\big]
 \eta \gamma \eta   \\
&\ge
   \frac{1}{2} \tr \big[ D^2 F_>^2 \eta\gamma \eta\big]  - |\tr  F_>^2 \sigma(D)
\sigma(A')\eta \gamma \eta|- |\tr  F_>^2 \sigma(A')\sigma(D)\eta \gamma \eta|. \non
\end{align}
Inserting $\wt G_>:=1-\wt F_>$ into the second term, and using  \eqref{Fi}, we have
\be
   |\tr  F_>^2 \sigma(D)\sigma(A')\eta \gamma \eta| 
\le |\tr  F_>^2 \sigma(D)\sigma(A') \wt F_> \eta \gamma \eta \wt F_> |+
|\tr  F_>^2 \sigma(A') \wt G_> \eta \gamma \eta \wt F_> \sigma(D) |.
\label{DA'}
\ee
For the first term in \eqref{DA'}, with the notation 
$\wt\om_>:=\wt F_> \eta \gamma \eta \wt F_> $,
 we use
\be
 |\tr  F_>^2 \sigma(D)\sigma(A') \wt F_> \eta \gamma \eta \wt F_> |\le
 \Big( \tr [A']^2 \wt\om_>\Big)^{1/2}
 \Big( \tr D^2\wt \om_>\Big)^{1/2} .
\label{fir}
\ee
Applying H\"older, Sobolev and Lieb-Thirring inequalities and the
bounds \eqref{Aprime}, \eqref{nablaA1}, we have
\begin{align}\label{hol}
    \tr [A']^2 \wt\om_> 
  & \le  \Big( \int [A']^5\Big)^{2/5} \Big(\int \wt\varrho_>^{5/3}\Big)^{3/5} \non\\
 &\le   C\Big( \int [A']^2\Big)^{1/10} \Big( \int [A']^6\Big)^{3/10}
\Big(h^{-2}\tr D^2\wt\om_>\Big)^{3/5}
 \non\\
& \le CL^{1/5} \Big( \int_{B(2L)} |\nabla\otimes A|^2\Big)
 \Big(h^{-2}\tr D^2\wt\om_>\Big)^{3/5}, 
\end{align}
where $\wt\varrho_>(x):= \wt \om_>(x,x)$ denotes the density of $\wt \om_>$.
Thus \eqref{fir} can be estimated as
\begin{align}
    |\tr  F_>^2 \sigma(D)\sigma(A') \wt F_> \eta \gamma \eta \wt F_> |
  & \le CL^{1/10} h^{-3/5} \Big( \int_{B(2L)} |\nabla\otimes A|^2\Big)^{1/2}
   \Big(\tr D^2\wt\om_>\Big)^{4/5} \non\\
&\le C_\delta L^{1/2}h^{-3}  \Big( \int_{B(2L)} |\nabla\otimes A|^2\Big)^{5/2}
 + \delta \tr D^2\wt\om_>\non\\
&\le C_\delta (h^{-1/2}L)^5  h^{-2}W^3 \int_{B(2L)} |\nabla\otimes A|^2
 + \delta \tr D^2\wt\om_>\, ,
\end{align}
where we used \eqref{eq:1new} in the last step. Considering that $LW^{1/2}\ll h^{1/2}$,
we obtain that the first term in  \eqref{DA'} is bounded by the 
two negative error terms in \eqref{largeenergygamma}.

For the second term in  \eqref{DA'}, after a Schwarz inequality, we use Lemma \ref{lem:comm2}
similarly to \eqref{appl}:
\begin{align}\label{secse}
|\tr  F_>^2 \sigma(A') \wt G_> \eta \gamma \eta \wt F_> \sigma(D) |&\le
\Big(\tr  F_>^2 A' \wt G_>^2 A' F_>^2 \Big)^{1/2}
\Big(\tr  D^2 \wt F_>^2(\eta \gamma \eta)^2  \Big)^{1/2} \non\\
&\le C\Big(L h^{-2} W  \int_{B(2L)} |\nabla\otimes A|^2\Big)^{1/2}
 \Big(\tr  D^2 \wt F_>^2\eta \gamma \eta  \Big)^{1/2}\non\\
&\le C\delta h^{-3/2}W^{1/2} \int_{B(2L)} |\nabla\otimes A|^2
 + \delta \tr  D^2 \wt F_>^2\eta \gamma \eta
\end{align}
for $h\le h_\delta$,
using $LW^{1/2}\ll h^{1/2}$ and estimating
$\| \nabla f_>\|_1\le CW$.
 Therefore the second term in \eqref{DA'} is also bounded by the 
two negative error terms in \eqref{largeenergygamma}.

We now estimate the third term in \eqref{newad}:
\begin{align}\label{third}
 \big|\tr  F_>^2 \sigma(A')\sigma(D)\eta \gamma \eta\big| &\le
  \big|\tr  F_>^2 \sigma(A')\wt F_> \sigma(D)\eta \gamma \eta\big|+
 \big|\tr  F_>^2 \sigma(A') \wt G_>\sigma(D)\eta \gamma \eta\big| \\
 &\le \Big(\tr F_>^2 \sigma(A')^2 F_>^2 \wt\om_>\Big)^{1/2}
  \Big(\tr D^2 \wt\om_>\Big)^{1/2} \non\\
  & \quad +  \Big(\tr F_>^2 \sigma(A')\wt G_> \sigma(A') F_>^2\Big)^{1/2}
  \Big(\tr \wt F_> \eta\gamma\eta \sigma(D)\wt G_> \sigma(D)\eta\gamma\eta \wt F_>\Big)^{1/2}
  \non \\
 & \le  C \Big( \tr [A']^2 F_>^2 \wt \om_> F_>^2\Big)^{1/2}
    \Big( \tr D^2\wt\om_>\Big)^{1/2} \non \\
&  \quad
+C \Big( \tr  F_>^2 A'\wt G_> A' F_>^2\Big)^{1/2} \Big( \tr D^2\wt\om_>\Big)^{1/2}. \non
\end{align}
Here we used that $F_>^2= F_>^2\wt F^2_>$ and
that $\sigma(D) \wt G_>\sigma (D) = D^2 \wt G_>\le D^2$.
The second term in the r.h.s of \eqref{third} is estimated exactly as 
\eqref{secse}. For the first term we essentially repeat the estimate
\eqref{hol}:
\begin{align*}
  \tr [A']^2 F_>^2 \wt \om_> F_>^2 \le 
 CL^{1/5} \Big( \int_{B(2L)} |\nabla\otimes A|^2\Big) \Big(h^{-2}\tr D^2 \wt \om_> 
\Big)^{3/5}
\end{align*}
after estimating $\tr D^2 
F_>^2 \wt \om_>  F_>^2 \le\tr D^2 \wt \om_> $.
This completes the
proof of \eqref{largeenergygamma}.

\bigskip
For the proof of  \eqref{largeenergyP}, we write
\begin{align*}
  | \tr F_>^2 \wt H(A) \eta P \eta | \le  | \tr F_>^2 H_0 \eta P \eta |
 + 2\big|\tr  F_>^2 \big[ \sigma(D)\sigma(A')+\sigma(A')\sigma(D)\big]\eta P \eta\big|.
\end{align*}
Notice that in the estimate of the $\sigma(D)\sigma(A')$ and $\sigma(A')\sigma(D)$
  terms on the r.h.s of \eqref{newad} explained above is valid for any 
density matrix $\gamma$, in particular also for $\gamma =P$, thus
\begin{align}\label{fr}
  | \tr F_>^2 \wt H(A) \eta P \eta | \le  (1+\delta)| \tr \wt F_>^2 H_0 \eta P \eta |
  + C\delta h^{-2} W^{1/2}\int_{B(2L)} |\nabla\otimes A|^2.
\end{align}
On the support of $\wt F_>$ we have $|H_0|\le 2 D^2$, thus
 $\tr \wt F_>^2 H_0 \eta P \eta\le 
2 \tr D\wt F_>\eta  P\eta \wt F_>D$.
To estimate  
$ \tr D\wt F_>\eta  P\eta \wt F_>D$, we use 
\eqref{compact} from Lemma \ref{lm:disjoint} below.
We set $\ell: = h W^{-1/2}$, $f(p):= \wt f_> (W^{1/2} p)p$
and $g(p):= {\bf 1}( |p|\le 1)$,
so that $\wt F_>(D) D = W^{1/2}f(-i\ell \nabla)$ and
$P= g(-i\ell \nabla)$
\be
   \tr D \wt F_>\eta  P\eta \wt F_>D= \| D\wt F_>\eta P\|_{HS}^2
  \le  C_{N} W^{5/2}\Big(\frac{h}{LW^{1/2}}\Big)^N
\label{eq:separate}
\ee
using that the supports of $f$ and $g$ are separated by a distance $d$ of order one.
The condition $\ell \le Ld$ is guaranteed by \eqref{Lh}.
Upon combining this with \eqref{fr} we have completed the proof of
Lemma~\ref{lm:largeenergy}. \qed

\begin{lemma}\label{lm:disjoint} Let $f$ and $g$ be real functions so that
their supports are separated by $d$, i.e.
\be
   \mbox{dist}(\mbox{supp}(f), \mbox{supp}(g) )\ge d.
\label{dist}
\ee
Let $\eta(x):=\eta_0(x/L)$ with $L>0$,  where $\eta_0$
is a smooth function with compact support. Then for any $N>0$
and $\ell \le Ld$ we have the following bounds
\be\label{L2}
  \| f(-i\ell\nabla)\eta g(-i\ell\nabla) \|_{HS} \le
   C_{N} \Big(\frac{\ell }{dL}\Big)^N \Big(\frac{L}{\ell}\Big)^3\|f\|_2\|g\|_2
\ee
and
\be\label{Linfty}
\| f(-i\ell\nabla)\eta g(-i\ell\nabla) \|_{HS} \le
   C_{N} \Big(\frac{\ell}{dL}\Big)^N\Big(\frac{L}{\ell}\Big)^{3/2}\|f\|_2\|g\|_\infty,
\ee
\be\label{compact}
\| f(-i\ell\nabla)\eta g(-i\ell\nabla) \|_{HS} \le
   C_{N} \Big(\frac{\ell}{dL}\Big)^N \Big(\frac{L}{\ell}\Big)^3
\Big\| \frac{f(p)}{1+ p^2}\Big\|_\infty
\|g(p)(1+p^2)^3\|_\infty
\ee
where $C_N$ depends only on $\eta_0$ and $N$.
\end{lemma}

{\it Proof.} Since $\eta_0\in C_0^\infty$, we have
$$
   |\wh \eta_0(p)|\le C_N (1+p^2)^{-N/2}
$$
for any $N\ge 0$. We compute in Fourier space, using 
that $\wh\eta(p)= L^3\wh \eta_0(pL)$
(with an appropriate convention about the $2\pi$), and \eqref{dist}, we get
\begin{align}
    \| f(-i\ell\nabla)\eta g(-i\ell\nabla) \|_{HS}^2 = &
\iint \rd p \rd q  \; |f(\ell p)|^2 |\wh \eta(p-q)|^2  |g(\ell q)|^2\non\\
= & (L/\ell)^{6} \iint \rd p \rd q  \; |f(p)|^2 |\wh \eta_0\big((p-q)L/\ell\big)|^2
  |g(q)|^2 \non\\
\le & C_N  (L/\ell)^{6} \iint \rd p \rd q  \;
 \frac{|f(p)|^2|g(q)|^2}{[1+ (p-q)^2(L/\ell)^2]^N}
\non\\
\le & C_N(\ell/Ld)^{2N}\Big(\frac{L}{\ell}\Big)^6\|f\|_2^2\|g\|_2^2 
\end{align}
which proves \eqref{L2}. For the proof of \eqref{Linfty}, we
extract a decay of order $(\ell/Ld)^{2N-4}$, estimate $|g(q)|\le \|g\|_\infty$
and then  integrate out $q$.
The proof of \eqref{compact} is similar:
\begin{multline}
    \| f(D)\eta g(D) \|_{HS}^2 \\
\le C_N\Big\| \frac{f(p)}{1+p^2}\Big\|^2_\infty
\|g(q)(1+q^2)^3\|_\infty^2
  \Big( \frac{\ell}{Ld} \Big)^{2N-6} \Big(\frac{L}{\ell}\Big)^6 
\iint \rd p \rd q  \; \frac{(1+p^2)}{(1+q^2)^3[1+ (p-q)^2]^3}
\end{multline}
and the last integral is finite.
This completes the proof of Lemma \ref{lm:disjoint}. \qed

\subsubsection{Error terms in  \eqref{bigsplit} for $|i|\le i_0$}

 In Lemma \ref{lm:largeenergy} we have estimated the terms $i>i_0$
in the last two summations in \eqref{bigsplit}.
Note that the terms with $i<-i_0$ in the second line of \eqref{bigsplit} identically vanish.
It now remains to  estimate the last four sums in \eqref{bigsplit}
for $0<|i|\le i_0$. 

\begin{lemma}\label{lm:sumi} Define $w_i:= 2^{|i|}w$ and assume \eqref{Lh} and
 $LW^{1/2}\ll w$. Then
for any $\delta>0$ and $h\le h_\delta$ we have
\begin{align}
\sum_{0<i\le i_0} \Re \tr & F_i^2 \wt H(A) \eta \gamma \eta  + 
\sum_{-i_0\le i<0} \Re \tr F_i^2 \wt H(A) \eta (\gamma-1) \eta \non \\
\ge 
& \sum_{0<i\le i_0} \frac{w_i}{40} \tr[ D^2 F_i^2 \eta \gamma \eta]
+\sum_{-i_0\le i<0}\frac{w_i}{40} \tr[ D^2 F_i^2 \eta (1-\gamma) \eta] \non\\
&-C\delta h^{-2}W^{1/2}\int_{B(2L)} |\nabla \otimes A|^2
\label{ilegamma} 
\end{align}
and
\begin{multline}
\sum_{0<i\le i_0} \Big| \tr F_i^2 \wt H(A)\eta P \eta 
\Big|+ \sum_{-i_0\le i<0}\Big| \tr F_i^2 \wt H(A)\eta (P-1) \eta \Big| \\
\le  C\delta h^{-2}W^{1/2}\int_{B(2L)} |\nabla \otimes A|^2  +
 C_NW^{5/2}\Big(\frac{h}{wLW^{1/2}}\Big)^N.
\label{ileP} 
\end{multline}
\end{lemma}

{\it Proof of Lemma \ref{lm:sumi}.}
We will present the proof  for the summations over $0<i\le i_0$,
the estimate of the negative $i$'s are identical.
 We start with the first term 
in \eqref{ilegamma}.
Since $w_i= 2^{|i|}w$, then $w_i\le 2$ for $|i|\le i_0$.
 Since on the support of $F_i$, $i>0$, it holds that
$\frac{3}{8} w_iW  \le D^2-W\le \frac{5}{4} w_iW $,  we obtain
$$
   H_0F_i^2 \ge \frac{3}{8} w_iW  F_i^2  \ge \frac{3}{8}
\frac{w_iD^2 F_i^2}{\frac{5}{4}w_i+1} 
 \ge \frac{w_i}{20} (D^2+W) F_i^2,
  \qquad 0 < i \le i_0.
$$
The analogous estimate for negative $i$ will be
$$   
  H_0F_i^2 \le -  \frac{3}{8} w_i W F_i^2 \le  -\frac{w_i}{20} (D^2+W) F_i^2,
  \qquad 0 > i \ge -i_0.
$$
The additional $W$ is necessary only for $i<0$, when
$\tr  F_i\eta\gamma\eta F_i$ may not be comparable with $\tr D^2 F_i\eta\gamma\eta F_i$, but
it is always comparable with
$\tr (D^2+W) F_i\eta\gamma\eta F_i$. 

Using the second identity in \eqref{tildeH},
we have for $0<i\le i_0$,
\begin{align}\label{ii1}
\Re \tr  F_i^2 \wt H(A)\eta\gamma\eta  = & 
  \tr  F_i^2  H_0\eta\gamma\eta +2\Re\tr F_i^2 DA' \eta\gamma\eta
 + h\Re \tr F_i^2 \sigma(B) \eta\gamma\eta
   \non\\
  \ge &  
   \frac{w_i}{20}
 \tr (D^2+W) F_i\eta\gamma\eta F_i  - 2|\tr F_i^2 DA' \eta\gamma\eta|
- h| \tr F_i^2 \sigma(B') \eta\gamma\eta|.
\end{align}
The last term was already treated in  \eqref{Bterm} even after
summation over $i$.

To control the  $DA'$ error term, we proceed as before
in \eqref{tildesplit} and \eqref{DA'} by inserting
$\wt G_i:=1-\wt F_i$:
\be\label{secc}
   |\tr F_i^2 DA' \eta\gamma\eta| \le
 |\tr F_i^2 DA'\wt F_i \eta\gamma\eta\wt F_i|+
|\tr F_i^2 A' \wt G_i\eta\gamma\eta \wt F_i D|.
\ee
The first term is estimated similarly to \eqref{1st}:
\begin{align}\label{4st}
  |\tr F_i^2 DA'\wt F_i \eta\gamma\eta\wt F_i| &\le 
 \Big( \tr [A']^2\wt F_i \eta\gamma\eta \wt F_i\Big)^{1/2}
 \Big( \tr F_i^4D^2 \wt F_i \eta\gamma\eta \wt F_i\Big)^{1/2}\non \\
  &\le C \Big( \tr [A']^2\wt F_i^2\Big)^{1/2}
\Big( \tr D^2 F_i \eta\gamma\eta  F_i\Big)^{1/2}
  \non \\
 &\le C \Big( h^{-3}w_iW^{3/2} \int [A']^2 \Big)^{1/2} 
\Big( \tr D^2 F_i \eta\gamma\eta  F_i\Big)^{1/2}
 \non \\
  &\le \delta w_i \tr D^2 F_i \eta\gamma\eta  F_i
 + C\delta^{-1}(h^{-1}L^2W) h^{-2}W^{1/2}\int_{B(2L)} |\nabla \otimes A|^2 .
\end{align}
Here we used $F_i\wt F_i= F_i$ to remove the tildes from the terms
with $D^2$. We also used  \eqref{Aprime} and the estimate
\begin{align}\label{eq:supFi}
   \sup_x  F_i(x,x)\le \sup_x \wt F_i(x,x) \le Ch^{-3} w_iW^{3/2}.
\end{align}

The second term in \eqref{secc} is estimated similarly to
 \eqref{secse} by using Lemma \ref{lem:comm2} and the fact that
$\|\nabla f_i\|_1\le CW$.
\begin{align}\label{5st}
|\tr F_i^2 A' \wt G_i\eta\gamma\eta \wt F_i D|
&\le 
 \Big( \tr F_i^2 A' \wt G_i^2 A' F_i^2\Big)^{1/2}
 \Big( \tr D^2 F_i^2 \eta\gamma\eta \Big)^{1/2}\non \\
  &\le C \Big( Lh^{-2}W \int_{B(2L)} |\nabla \otimes A|^2\Big)^{1/2}
\Big( \tr D^2 F_i \eta\gamma\eta  F_i\Big)^{1/2}
  \non \\
  &\le \delta w_i \tr D^2 F_i \eta\gamma\eta  F_i
 + C\delta^{-1}\frac{LW^{1/2}}{w_i} h^{-2}W^{1/2}\int_{B(2L)} |\nabla \otimes A|^2 .
\end{align}
Summing up \eqref{4st} and \eqref{5st} for $i\le i_0 \le C|\log h|$ 
and using that $LW^{1/2}\ll h^{1/2}|\log h|^{-1}$ and $LW^{1/2}\ll w$, we obtain
\be
  \sum_{0<i\le i_0}  |\tr F_i^2 DA' \eta\gamma\eta| \le
  \delta \sum_{0<i\le i_0} w_i \tr D^2 F_i \eta\gamma\eta  F_i
 + \delta  h^{-2}W^{1/2}\int_{B(2L)} |\nabla \otimes A|^2
\label{ii2}
\ee
if $h\le h_\delta$.
Combining \eqref{ii2} with the positive terms from
\eqref{ii1} and choosing $\delta$ sufficiently small, we obtain\be\label{iii}
\Re \sum_{0<i\leq i_0} \tr  F_i^2 \wt H(A)\eta\gamma\eta  
  \ge 
  \sum_{0<i\leq i_0}   \frac{w_i}{40}
 \tr (D^2+W) F_i\eta\gamma\eta F_i 
 -\delta  h^{-2}W^{1/2}\int_{B(2L)} |\nabla \otimes A|^2.
\ee
This completes the proof of \eqref{ilegamma}.
The proof of \eqref{ileP} is analogous by using
 that the momentum supports of $F_i$ and $P$ are
separated, so one can  apply \eqref{compact} similarly to \eqref{eq:separate} to obtain
\be
   \Big| \tr F_i^2 H_0\eta P\eta \Big| \le  \tr D F_i  \eta P\eta F_iD =
  \| DF_i \eta P\|_{HS}^2 \le C_NW^{5/2}\Big(\frac{h}{wLW^{1/2}}\Big)^N.
\label{kinsmall}
\ee
The details  will be left to the reader.
This completes the proof of Lemma \ref{lm:sumi}. \qed.

\bigskip

Inserting the estimates from
 Lemma~\ref{lm:ad}, Lemma~\ref{lm:largeenergy} and
Lemma~\ref{lm:sumi}  into \eqref{bigsplit}, we have proved
the following

\begin{proposition}\label{prop:KineticBoundsScaled}
Under the conditions of Theorem \ref{thm:2} and assuming $LW^{1/2}\ll w$,
 for any $\delta>0$ and $h\le h_\delta$ we have
\begin{align}\label{eq:long}
 \tr H(A) \eta \gamma \eta  &\ge \tr \eta [H_0]_-\eta-
  C\delta^{-1} \Big[ h   +w^2\Big]W^{5/2}\Lambda_L
- C_NW^{5/2}\Big(\frac{h}{wLW^{1/2}}\Big)^N\non\\
 &\quad
 +  \frac{1}{100} \tr[ D^2 F_>^2 \eta \gamma \eta]
 +\sum_{0<i\le i_0} \frac{2^iw}{100} \tr[ D^2 F_i^2 \eta \gamma \eta]
+\sum_{-i_0\le i<0}\frac{2^{|i|}w}{100} \tr[ D^2 F_i^2 \eta (1-\gamma) \eta] \non\\
 &\quad - C\delta h^{-2}W^{1/2} \int_{B(2L)} |\nabla \otimes A|^2.
\end{align}
\end{proposition}

Finally, given $0\le\al<1$ and $\e>0$ as in Theorem \ref{thm:2}, we  choose 
$w = h^{\al/2}$ and an integer $N\ge 3\e^{-1}$,
then  Proposition \ref{prop:KineticBoundsScaled}
implies  Theorem~\ref{thm:2}. \qed

\bigskip

As we already mentioned,  Theorem~\ref{thm:2}  implies Theorem~\ref{thm:1}.

\bigskip

Upon combining Proposition~\ref{prop:KineticBoundsScaled}
with the upper bound we get bounds on the kinetic energy 
terms  in \eqref{eq:long}. These bounds
will be used to control the 
substitution of $A_0$ by $A_{r}$ in Section~\ref{sec:EstimatingAr}.

For each fixed $u\in B(\frac{5}{4}L_0)$,
we will use Theorem  \ref{thm:2} with the constant potential $W$ replaced with
$V_u$ defined in  \eqref{eq:10}.  We know that $c_0/2\le V_u\le C$, so
factors $W=V_u$ can be replaced by constants in the estimates.
To indicate the $u$-dependence,
we define $H^u:= D^2- V_u$. Recall that the operators $F_i$ and $F_>$
defined in Section  \ref{sec:decomp} depend on $W$. We will denote
them by $F_i^u$ and $F_>^u$ in case of $W=V_u$.

\begin{theorem}\label{thm:5.14}
Recall that $L_0 = h^{1/2 -\e_0}$ and $\alpha = 1- 3\e_0$ with $\e_0$ sufficiently small. Choose $L_1 = h^{1/2 +\e_0}$.
If
\begin{align}
\label{eq:EstFunctional}
{\mathcal E}(A) \leq {\mathcal E}(0) + C' h^{1-3\e_0} \Lambda_{L_0},
\end{align}
then,
\begin{align}
  \label{eq:15}
  \int {\mathcal T}_u \frac{\rd u}{L_1^3} \leq C h^{1-3\e_0} \Lambda_{L_0},
\end{align}
where ${\mathcal T}_u$ is given by
\begin{align}
  \label{eq:12}
  {\mathcal T}_u:=  \frac{1}{100} \Big\{ &\tr[ D^2 [F_>^u]^2
  \eta_u  \gamma_u \eta_u]
+\sum_{0<i\le i_0} 2^i w \tr[ D^2 [F_i^u]^2 \eta_u 
\gamma_u \eta_u]\nonumber \\
&+\sum_{-i_0\le i<0} 2^{|i|}w \tr[ D^2 [F_i^u]^2
\eta_u
(1-\gamma_u) \eta_u] \Big\},
\end{align}
where
\be
  \label{eq:13}
  \gamma_u := {\bf 1}_{(-\infty,0]}\big( \chi_u \phi [T_h(A-c_u)-V_u
 ]\phi\chi_u \big). 
\ee
with $c_u$ defined in \eqref{cu} and we set $\eta_u:= \chi_u\phi$.
\end{theorem}

\begin{proof}
By \eqref{2new}
\begin{align}\label{2newNew}
 \tr \Big[ \phi [T_h(A) &-V]\phi \Big]_- \ge
  \int \tr\Big[ \chi_u \phi [T_h(A-c_u)-V_u
 ]\phi\chi_u\Big]_- \frac{\rd u}{L_1^3} ,
\end{align}
with $V_u$ from \eqref{eq:10}. Since $V(0)=c_0$ by
\eqref{eq:lowerboundV} and using $L_0 \ll 1$ and the bound on $\nabla
V$, we can assume that $V \geq c_0/2$ on $B(\frac{5}{4}L_0)$. Also, using \eqref{eq:EstFunctional} we get as in the proof of Corollary~\ref{cor:BoundMagn} that
\begin{align*}
    \frac{1}{2\kappa h^2}\int_{B(2L_0)} |\nabla \otimes A|^2 \le 
C h^{1-3\e_0}\Lambda_{L_0}.
\end{align*}
Therefore,
(using the smallness of $\e_0$) the condition \eqref{Bapr} is satisfied
for all $u$ in the domain of integration. 
Note that this is exactly the condition \eqref{eq:1new}
in Theorem \ref{thm:2} with $W=V_u$, with $L=L_1$ and shifting the center of the
ball $B(2L)$ to $u$. The condition 
\eqref{eq:3} is satisfied since we consider the vector potential $A-c_u$.
Finally, the condition \eqref{Lh} is satisfied  by \eqref{hL}.
We can thus apply
Proposition~\ref{prop:KineticBoundsScaled} and we therefore have
\begin{align}
  \label{eq:11}
  \tr \Big[ \phi [T_h(A) -V]\phi \Big]_- &\ge
\int\Big\{
\tr \chi_u \phi [H_0 - V_u]_{-} \phi \chi_u 
-
C_{\delta} [ h + w^2 ]  V_u^{5/2} \Lambda_{L_1} \\
&\quad - C_N V_u^{5/2}\left( \frac{h}{ w L_1V_u^{1/2}} \right)^N 
- C\delta h^{-2} V_u^{1/2} \int_{B_u(2L_1)} |\nabla \otimes A|^2
+ {\mathcal T}_u \Big\}\frac{\rd u}{L_1^3}. \non
\end{align}
Estimating as around \eqref{5/2} (and choosing $N$ sufficiently
large),
 we get
\begin{align}
  \label{eq:14}
  \tr \Big[ \phi [T_h(A) -V]\phi \Big]_- &\ge
\int \tr \chi_u \phi [ H_0 - V(u)]_{-} \phi \chi_u \frac{\rd
  u}{L_1^3}  - C\delta h^{-2}\int_{B(2L_0)} |\nabla \otimes A|^2
 \nonumber \\
&\quad-  \Big(
C_{\al}h^{\al}+ Ch^2L^{-2}_1 + CL^2_1\Big)\Lambda_{L_0}
+ \int {\mathcal T}_u \frac{\rd
  u}{L_1^3} .
\end{align}
We evaluate the main term using precise semiclassics.
Define $\wt \phi(x) = \phi(L_0 x)$, $\wt V(x) = V(L_0 x)$. 
Then we get by unitary scaling and applying Theorem~\ref{thm:ivrii2} 
with the effective semiclassical parameter $\wt h:= h/L_0= h^{1/2+\e_0}$ that
\begin{align}
{\mathcal E}(0) &= \tr\big( \phi [-h^2 \Delta - V ] \phi \big)_{-}
= \tr\big( \wt \phi [-\wt h^2 \Delta - \wt V ] \wt \phi \big)_{-}\nonumber \\
&=2(2\pi \wt h)^{-3} \int \wt \phi(x)^2 [p^2 - \wt V(x)]_{-} \rd x \rd p + 
{\mathcal O}(\wt h^{-1}) \nonumber \\
&=2 \frac{1}{(2\pi h)^3} \int \phi(x)^2 [p^2 - V(x)]_{-} \rd x \rd p + {\mathcal O}(\Lambda_{L_0} \wt h^{2}).
\end{align}
By Remark~\ref{remark:6.2}
\begin{align}
\int \tr \chi_u \phi [ H_0 - V(u)]_{-} \phi \chi_u \frac{\rd u}{L_1^3}
&=
\frac{2}{15\pi^2}
 h^{-3}\int V_{-}^{5/2}\phi^2 + O(L_1^2h^{-3})\int \phi^2 \nonumber \\
&=
2 \frac{1}{(2\pi h)^3} \int \phi(x)^2 [p^2 - V(x)]_{-} \rd x \rd p + {\mathcal O}(\Lambda_{L_0} L_1^2).
\end{align}
Inserting the choices of $L_0=h^{1/2-\e_0}$ and $L_1=h^{1/2+\e_0}$ we get
\be
    \Big|  {\mathcal E}(0) - \int \tr \chi_u \phi [ H_0 - V(u)]_{-} 
\phi \chi_u \frac{\rd  u}{L_1^3}\Big|\le C h^{1+2\e_0} \Lambda_{L_0}.
\label{Esc}\ee
Finally, choosing sufficiently small $\delta$ (depending on $\kappa$) and
using \eqref{eq:EstFunctional} and \eqref{Esc}, we get \eqref{eq:15}.
\end{proof}

\subsection{Estimates on the smoothed vector field}\label{sec:EstimatingAr}

This section is a technical preparation for the next Section \ref{sec:smoothing}.
Let $A \in H^1(\bR^3)$ be a vector field. 
We fix a radial function $\chi \in C^{\infty}_0(B(1))$
with $\int\chi=1$ 
 and define $\chi_{r}(x) = r^{-3} \chi(x/r)$ and $A_{r} = A * \chi_r$.
We note that $\nabla\cdot A_r =0$ on a ball $B(L)$ if $\nabla\cdot A=0$
on $B(L+r)$.

Suppose also given $\phi'$ with $0\leq \phi' \leq 1,$ $\supp \phi'
\subset B(3L/2)$, $\phi'\equiv 1$ on $B(L)$ and $|\nabla \phi'| \leq C/L$. With this function
$\phi'$ we define $A' = \phi' A$ and $A_r' = \phi' A_r$, then
$A=A'$ and $A_r=A'_r$ on $B(L)$. The constants in the following sections
may depend on $\chi$ and on the constant $C$ in the estimate $|\nabla \phi'| \leq C/L$,
but we will neglect these dependences.

\begin{lemma}\label{lem:AminusAl2}
If $r \leq L/2$, we have
\begin{align}
\label{eq:AminusAl}
\| A' - A_{r}' \|_2^2 \leq Cr^2 \int_{B(2L)}  |\nabla \otimes A|^2
\end{align}
and
\begin{align}\label{eq:AminusAl6}
\| A' - A_{r}' \|_6^2 \leq C \int_{B(2L)}  |\nabla \otimes A|^2.
\end{align}
\end{lemma}

\begin{proof}
We have
\begin{align}\label{previ}
\| A' - A_{r}' \|_2^2 &\leq \int_{B(3L/2)}|A - A_{r}|^2 =
\int_{B(3L/2)}\Bigg| \int_{B(1)} \big[A(x) - A(x + r z)\big] \chi(z)\,\rd z\Bigg|^2\,\rd
x\nonumber \\
&= \int_{B(3L/2)}\sum_{j=1}^3\Bigg| \int_{B(1)} \int_0^1 r z \cdot \nabla A_j(x+tr z)
\chi(z) \,\rd t \rd z\Bigg|^2 \rd x\nonumber \\
&\leq  \frac{4\pi}{3} r^2\int_{B(3L/2)} \sum_{j=1}^3 \int_{B(1)} 
\int_0^1 | z \cdot \nabla A_j(x+tr z)|^2
\,\rd t \rd z \rd x,
\end{align}
by Cauchy-Schwarz. Upon changing variable in the $x$-integral, we obtain \eqref{eq:AminusAl}.

By the Sobolev inequality, we have
\begin{align}
\| A' - A_{r}' \|_6^2 \leq C \int |\nabla \otimes (A' - A_{r}' )|^2 \leq
C \int_{B(3L/2)} \Big[ |\nabla \otimes (A - A_{r} )|^2 + L^{-2} |A - A_{r} |^2\Big].
\end{align}
In the first term we use Young's inequality,
 $\int_{B(3L/2)} |\nabla \otimes A_r|^2 = \int_{B(3L/2)} |\chi_r *(\nabla \otimes A)|^2
\le \int_{B(2L)} |\nabla \otimes A|^2$, after
extending the integration to $\bR^3$ and using $r\le L/2$
to insert a cutoff function ${\bf 1}_{B(2L)}$.
 In the second term we use the previous
calculation \eqref{previ}.
\end{proof}

\begin{lemma}\label{lem:derivatives} Suppose the vector field $A$ satisfies 
\begin{align}\label{eq:gaugezero2}
\int_{B(2L)} A = 0.
\end{align}
Then, for all $r<L/2$ we have
\begin{align}\label{eq:noderiv}
 \| A_{r} \|_{L^2(B(3L/2))}^2 \leq C L^2 \int_{B(2L)} |\nabla
 \otimes A|^2,
\end{align}
and for all multi-indices $n \in \bN^3\setminus \{ 0 \}$.
\begin{align}
\| \partial^{n} A_r \|_{L^2(B(3L/2))}^2\leq C_{n} r^{2-2|n|}
\int_{B(2L)} |\nabla \otimes A|^2 \, ,
\end{align}
where the constants $C, C_{n}$ are independent of $A, L$ and $r$.
\end{lemma}

\begin{proof}
Analogously to the proof of Lemma~\ref{lem:AminusAl2} we get
\begin{align}
 \| A_{r} - A \|_{L^2(B(3L/2))}^2 \leq C L^2 \int_{B(2L)}
 |\nabla \otimes A|^2 \, .
\end{align}
Also, using \eqref{eq:gaugezero2} and the Poincar\'{e} inequality,
\begin{align}
\| A_r \|_{L^2(B(3L/2))}^2 \leq \| A \|_{L^2(B(2L))}^2 \leq C L^2 \int_{B(2L)}
 |\nabla \otimes A|^2 \, ,
\end{align}
where, in the first step, we used $r \leq L/2$ and Young's inequality
as above. 
This
 finishes the proof of \eqref{eq:noderiv}.

To prove the estimate for the derivatives, for the given multi-index 
$n\in \bN^3\setminus \{ 0 \}$,
 we choose $e \in \bN^3$ with $|e| =1$, 
such that $n':= n - e \in \bN^3$, i.e., $|n|= |n'|+1$.
 Then calculate
\begin{align}
\label{eq:convolution}
\| \partial^{n} A_{r} \|_{L^2(B(3L/2))}^2 & = \int_{B(3L/2)}\Big| r^{-|n'|}
\int_{B(1)}  (\partial^{n'}\chi)(y) 
(\partial^{e} A)(x-r y)\,\rd y\Big|^2 \rd x\nonumber \\
&\leq r^{2-2|n|} \| \partial^{n'}\chi \|_2^2 \int_{B(1)}
 \int_{B(3L/2)} |(\partial^{e} A)(x-r y)|^2\,
\rd x \rd y\nonumber \\
&\leq C r^{2-2|n|}  \int_{B(2L)} |\nabla \otimes A|^2 \,\rd x,
\end{align}
where we used the assumption that $r \leq L/2$, so $B(3L/2)+B(r) \subset B(2L)$.
\end{proof}

We will use  the  results of this section 
for $ L=L_0 = h^{1/2 - \e_0}$
 where $\e_0$ is a small positive number.
First, with the choice of $L=L_0$, by change of variable and using $r\leq L_0/2$, we get
Lemma~\ref{lem:ScaledDerivatives} as a consequence of
Lemma~\ref{lem:derivatives}.
Second, we will use as an input the apriori bound \eqref{eq:1},
which was already proven as a result of Corollary~\ref{cor:BoundMagn}, namely
\begin{align}\label{eq:MagneticApriori}
h^{-2} \int_{B(2L_0)} |\nabla \otimes A|^2 \leq Ch^{\al} \Lambda_{L_0}, \qquad \al = 1-3\e_0.
\end{align}
This will give bounds on various norms of $A'-A_r'$ and $\pt^nA_r$
that will be used in the next section.

\subsection{Smoothing $A$}\label{sec:smoothing}

In this section we refine the result of Theorem \ref{thm:1}
and complete the proof of Theorem \ref{thm:replace} 
by proving \eqref{AtoAr}. 
The vector potential $A$ will not be removed as in 
Section~\ref{sec:remA}, but rather replaced with $A_r$ and this
results in a smaller error. 
We fix three length scales $L_0 \gg L_1 \gg r$, with
\begin{align}
\label{eq:parameters}
L_0 = h^{1/2-\e_0},\qquad L_1 = h^{1/2+\e_0},\qquad r = h^{1/2+\rho},
\end{align}
where $\rho$ will be chosen as  $\rho=C\e_0$ with a sufficiently large constant, 
and we assume that $A$ satisfies \eqref{eq:MagneticApriori}.
We will perform a similar analysis as in Section \ref{sec:remA},
in particular, we will again
perform a dyadic decomposition in energy and 
 the parameter $w$ in the dyadic decomposition will be chosen as
\begin{align}\label{deffw}
w = h^{\al/2} = h^{\frac{1}{2} -  \frac{3}{2}\e_0}.
\end{align}
\begin{lemma}\label{lm:smoothA}
For any sufficiently small $\rho>0$ there exists a
positive constant $\e_0$ (in fact $\e_0= c\rho$ can be chosen
where $c$ is a universal positive constant) such that the following is satisfied. 
Let $L_0$ and $r$ be given by
 \eqref{eq:parameters}, and $\phi\in C_0^\infty(B(L_0))$ with $|\pt^n\phi|\le C_nL_0^{-|n|}$. 
 Let $V$ satisfy the assumptions given in Theorem~\ref{thm:replace}.
 Let furthermore, $A$ satisfy \eqref{eq:MagneticApriori} and $\nabla\cdot A=0$
in $B(L_0+r)$.
Then,
\begin{align}\label{eq:6.106}
\tr\big[ \phi [T_h(A) -V] \phi \big]_{-} \geq 
\tr\big[ \phi[ T_h(A_r) -V] \phi \big]_{-} - C h^{1+\e_0} \Lambda_{L_0}.
\end{align}
\end{lemma}
Since \eqref{eq:1} was already proven in
Corollary~\ref{cor:BoundMagn},
Lemma~\ref{lm:smoothA} implies \eqref{AtoAr} and therefore finishes
 the proof of Theorem~\ref{thm:replace}.
The rest of this section will be devoted to the proof of Lemma~\ref{lm:smoothA}.

\begin{proof}
We will use localizations as in Section~\ref{sec:red}
 with $L = L_1$. Notice that the apriori assumption \eqref{Bapr} 
will now be satisfied on all the small boxes $B_u(2L_1)$ with $u\in B(\frac{5}{4}L_0)$,
 since
\begin{align}
h^{-2} \int_{B_u(2L_1)} |\nabla \otimes A|^2 \leq h^{-2} \int_{B(2L_0)} 
|\nabla \otimes A|^2
\leq Ch^{\al} \Big(\frac{L_0}{L_1}\Big)^3 \Lambda_{L_1} = Ch^{1 - 9 \e_0} \Lambda_{L_1},
\end{align}
where we used \eqref{eq:MagneticApriori} and the fact that $V_u\ge c_0/2$.
For the Pauli case, $T_h(A)= [\bsigma\cdot (D+A)]^2$, 
consider 
\begin{align}
T_h(A) &= T_h(A_r) + \sigma(A-A_r)\sigma(D)+\sigma(D)\sigma(A-A_r)- (A-A_{r})\cdot (A+A_{r})
\nonumber \\
&=
T_h(A_r) + 2  D (A-A_{r}) - (A-A_{r})\cdot (A+A_{r}) + h\sigma(B-B_r) . \label{thdec}
\end{align}
Here we used that $A$ and $A_r$ are divergence free which property holds
 on the support of $\phi$. Note that
\eqref{thdec} will be used only on $\mbox{supp}\, \phi$.

We denote 
\begin{align}\label{def:gamma0}
\gamma_0 := {\bf 1}_{(-\infty,0)}\big(\phi[T_h(A)-V]\phi\big)
\end{align}
and we collect a few apriori estimates on $\gamma_0$.

\begin{lemma}\label{lm:gamma0}
With the definition \eqref{def:gamma0} and assuming that $A$ satisfies
\eqref{eq:MagneticApriori} and $V$ is bounded, we get
\be
   \tr \gamma_0 \le Ch^{-9\e_0/4}\Lambda_{L_0}
\label{trgamma0}
\ee
\be
   \tr D^2\phi\gamma_0\phi\le Ch^{-9\e_0/4}\Lambda_{L_0}.
\label{trDgamma0}
\ee
\end{lemma}

\begin{proof}
By the variational principle and $T_h(A)\ge (D+A)^2- |B|$,
\begin{align}
  \tr \gamma_0 &\le 2\tr_{L^2(\bR^3)} {\bf 1}_{(-\infty,0)}\big( \phi\big(\; (D+A)^2
  - |B| - V\big)\phi\big) \nonumber \\
&\le 2 \tr_{L^2(\bR^3)} {\bf 1}_{(-\infty,0)}\big( (D+A)^2
  - (|B| + V) {\bf 1}_{\supp \phi}\big). 
\end{align}
Notice that the last inequality uses the fact that
 we consider the strictly negative eigenvalues.

By the CLR estimate we get, since $V$ is bounded,
\begin{align*}
  \tr \gamma_0 &\leq Ch^{-3}\int_{\supp \phi} ( |B|+ |V|)^{3/2}
  \le C\Lambda_{L_0}+ Ch^{-3} \int_{B(L_0)}|B|^{3/2}\\
&  \le  C\Lambda_{L_0} +C h^{-3}L_0^{3/4} \Big( \int_{B(L_0)} |\nabla \otimes A|^2\Big)^{3/4}\\
&  \le  C\Lambda_{L_0} +C h^{-3}L_0^{3/4} \Big( h^{2+\al} \Lambda_{L_0}\Big)^{3/4}
  =  Ch^{-9\e_0/4}\Lambda_{L_0}
\end{align*}
using \eqref{eq:MagneticApriori} and $\al= 1-3\e_0$, which proves \eqref{trgamma0}.

For the rest of this proof, set
\be
  A' = A\phi', \qquad V' = V\phi',
\label{A3prime}
\ee
where $\phi'$ is a cutoff function such that $\supp \phi' \subset B(3L_0/2)$,
$\phi'\equiv 1$ on $\supp \phi = B(L_0)$ and $|\nabla\phi'|\le C/L_0$.
By the magnetic Lieb-Thirring inequality,  Theorem~\ref{thm:lls}, we also have,
\begin{align}
   0\ge &  \Tr \phi (T_h(A')-V)\phi\gamma_0 \ge - Ch^{-3}\int [V']^{5/2} - 
Ch^{-3} \Big(\int |\nabla\otimes A'|^2 \Big)^{3/4} \Big( \int [V']^4\Big)^{1/4} \non\\
  \ge &  - C\Lambda_{L_0} - C h^{-3}L_0^{3/4} \Big( h^{2+\al} \Lambda_{L_0}\Big)^{3/4}
  = - Ch^{-9\e_0/4}\Lambda_{L_0},
\label{mlt10}
\end{align}
using \eqref{eq:MagneticApriori} and $\al= 1-3\e_0$. Therefore, by Schwarz and
Lieb-Thirring inequalities, from \eqref{mlt10} we get
\begin{align}
   \Tr D^2 \phi\gamma_0\phi \le & 2 \Tr \big[ T_h(A') + [A']^2]  \phi\gamma_0\phi
  \le  2 \Tr \big[ V' + [A']^2]  \phi\gamma_0\phi + Ch^{-9\e_0/4}\Lambda_{L_0}
 \non \\
  \le &  2 \Big(\int \big[ V' + [A']^2]^{5/2}\Big)^{2/5}\Big( \int 
 \big(\phi\varrho_0\phi\big)^{5/3}\Big)^{3/5}
  + Ch^{-9\e_0/4}\Lambda_{L_0}
 \non \\
\le &  C\Big( L_0^3 + \| A'\|_2^{1/2}\|A'\|_6^{9/2}\Big)^{2/5}
 \Big(h^{-2}   \Tr D^2 \phi\gamma_0\phi
  \Big)^{3/5}
  + Ch^{-9\e_0/4}\Lambda_{L_0},
\label{trd2}
\end{align}
where $\varrho_0$ is the density of $\gamma_0$.
Since $\int_{B(2L_0)}A=0$ (see  \eqref{eq:gaugezero}), we have
\be
   \| A'\|_2^2\le \int_{B(2L_0)} A^2 \le  CL_0^2 \int_{B(2L_0)} |\nabla\otimes A|^2
\label{AAA}
\ee
and
\be
  \| A'\|_6^2 \le \int |\nabla\otimes A'|^2 \le 
C\int_{B(2L_0)}\big[ |\nabla\otimes A|^2 + L_0^{-2}|A|^2\big]
  \le  C\int_{B(2L_0)} |\nabla\otimes A|^2 
\label{AAAA}
\ee
by Poincar\'e and Sobolev inequalities, thus
$$
   \| A'\|_2^{1/5}\|A'\|_5^{9/5} \le CL_0^{1/5}\int_{B(2L_0)} |\nabla \otimes A|^2
  \le CL_0^{1/5}h^{2+\al}\Lambda_{L_0} = CL_0^{3+1/5} h^{-2\e_0}.
$$
Clearly $L_0^{3+1/5} h^{-2\e_0} \ll L_0^{6/5}$, so we can continue \eqref{trd2} as
$$
  \Tr D^2 \phi\gamma_0\phi \le C\Lambda_{L_0}^{2/5} \Big(\Tr D^2 \phi\gamma_0\phi
  \Big)^{3/5} + Ch^{-9\e_0/4}\Lambda_{L_0}
$$
from which \eqref{trDgamma0} follows since $h^{-9\e_0/4}\Lambda_{L_0}\gg 1$ if
$\e_0$ is sufficiently small.
\end{proof}

Returning to the decomposition \eqref{thdec}, we
first show that the  effect of the quadratic term (in $A$) is negligible:
\begin{lemma}\label{lem:5.19}
\be
\Big|\tr (A-A_{r})(A+A_{r}) \phi\gamma_0 \phi\Big|
\leq C(r/L_0)^{1/10} (h^{-1} L_0^2) h^{-2-2\e_0} 
 \int_{B(2L_0)}  |\nabla \otimes A|^2.
\label{AAr}
\ee
In particular, assuming 
\eqref{eq:MagneticApriori} and that $\rho\ge 100\e_0$, we have
\be
  \Big|\tr (A-A_{r})(A+A_{r})  \phi \gamma_0\phi\Big|
 \leq Ch^{\frac{1}{10}(\rho-\e_0)-7\e_0+1}  \Lambda_{L_0}
  \le Ch^{1+\e_0}  \Lambda_{L_0}.
\label{AAr1}
\ee
\end{lemma}

\begin{proof} Analogously to the proof of Lemma~\ref{lm:gamma0}, 
we use H\"{o}lder and  Lieb-Thirring inequalities. The prime on the $A$
and $A_r$   denote localizations to $B(3L_0/2)$ as in \eqref{A3prime}. 
Then,  using $\| A'\|_p\le \|A\|_p$
and \eqref{AAA}--\eqref{AAAA}, we have the following estimate:
\begin{align}
\label{eq:LTforAminusAl}
\Big| \int (A-A_{r})(A&+A_{r})  \phi \rho_0\phi \,\Big|
 \leq \Big\{\int \big[(A'-A'_{r})(A'+A'_{r}) \big]^{5/2}\Big\}^{2/5} 
\Big\{\int (\phi \rho_0\phi)^{5/3}\Big\}^{3/5}\nonumber \\
&\leq \| A' - A'_{r}\|_2^{1/10} \| A' - A'_{r}\|_6^{9/10}  \| A' + A'_{r}\|_2^{1/10} \| A' + A'_{r}\|_6^{9/10} 
\Big\{ h^{-2} \tr D^2 \phi \gamma_0 \phi \Big\}^{3/5}\nonumber \\
&\leq C\Big\{ rL_0 \int_{B(2L_0)} |\nabla \otimes A|^2 \Big\}^{1/10} 
\Big\{ \int_{B(2L_0)}  |\nabla \otimes A|^2 \Big\}^{9/10} h^{-6/5-27\e_0/20}
 \Lambda_{L_0}^{3/5}\nonumber \\
&\leq C(r/L_0)^{1/10} (h^{-1} L_0^2) h^{-2-2\e_0} \int_{B(2L_0)}  |\nabla \otimes A|^2,
\end{align}
which gives \eqref{AAr}. Here 
 we also used Lemma~\ref{lem:AminusAl2} with $L=L_0$
 and  Lemma~\ref{lm:gamma0}.
\end{proof}

\bigskip

The main problem is to estimate the current term in   \eqref{thdec}, i.e. the term
\begin{align}\label{2curr}
\tr \Big[  \sigma(A-A_r)\sigma(D)+\sigma(D)\sigma(A-A_r)\Big] \phi \gamma_0 \phi
 =\tr  \Big[2  D (A-A_{r}) +h\sigma(B-B_r)\Big]\phi \gamma_0 \phi,
\end{align}
where we used that $A$ and $A_r$ are divergence-free on $\mbox{supp}\,\phi$.
We will apply localizations as in Section~\ref{sec:red}
 with $L = L_1$ on this term. Since this is a first order 
operator, we can do so without localization error, using
$$
   \int \chi_u D\chi_u \frac{\rd u}{L_1^3} = D \Big[ \frac{1}{2}\int
 \chi_u^2 \frac{\rd u}{L_1^3}
  \Big]  =0
$$
for the partition of unity defined in  \eqref{chiint}. So we have
\be
\tr D (A-A_{r}) \phi \gamma_0 \phi =
\int \tr\Big[ D (A-A_{r})  \eta_u  \gamma_0 \eta_u\Big]
 \frac{\rd u}{L_1^3},
\label{currloc}
\ee
where we set $\eta_u:= \chi_u\phi$.
 Similar expression holds for
the terms in the first line of \eqref{2curr}.

Thus, we have to estimate  (see \eqref{A3prime} for notation),
\begin{align}\label{alt1}
\int \tr \Big[ \big( 2D (A'-A_{r}')   +h\sigma(B'-B_r')\big)
\eta_u \gamma_0 \eta_u\Big]\frac{\rd u}{L_1^3}
\end{align}
or, equivalently,
\begin{align}\label{alt2}
\int \tr \Big[ \big(\sigma(A'-A_r')\sigma(D)+\sigma(D)\sigma(A'-A_r') \big)
\eta_u \gamma_0 \eta_u\Big]\frac{\rd u}{L_1^3}.
\end{align}
Here we have set
 $B':=\nabla \times A'$, $B_r':= \nabla\otimes A_r'$. We recall
that, similarly to \eqref{nablaA1}, we have
$$ 
    \int [B']^2 + \int [B'_r]^2 \le C\int_{B(2L_0)} |\nabla\otimes A|^2.
$$

We define, with $V_u$ from \eqref{eq:10},
\begin{align*}
  H^u = D^2 - V_u,\qquad \text{ and } \qquad P={\bf 1}_{(-\infty,0]}(H^u).
\end{align*}
Furthermore, we let $F_i^u, F_{>}^u, \wt F_{i}^u, \wt F_{>}^u$ be
defined as in  Section \ref{sec:decomp} with $W=V_u$.
The parameter $w$ entering the definition
of the $F$'s  has
 been fixed in \eqref{deffw} above as $w = h^{\al/2}=h^{\frac{1}{2}-\frac{3}{2}\e_0}$.

\begin{lemma}\label{lem:5.21}
Suppose $\rho \geq 100 \e_0$. Then, for any fixed $u\in B(\frac{5}{4}L_0)$
and with $\eta_u:=\chi_u\phi$ we have
  \be
    \label{eq:16}
    \tr \Big[ \big(\sigma(A-A_r)\sigma(D)+\sigma(D)\sigma(A-A_r) \big) 
\eta_u  \gamma_0 \eta_u\Big]
\geq - C h^{1+\e_0} \Lambda_{L_1} - C h^{4\e_0}{\mathcal T}_u ,
  \ee
where 
${\mathcal T}_u$ was defined in  \eqref{eq:12}.
\end{lemma}

We postpone the proof of Lemma~\ref{lem:5.21} and first finish
 the proof of Lemma~\ref{lm:smoothA}. Using 
\eqref{thdec}, \eqref{def:gamma0}, the estimate \eqref{AAr1} from
 Lemma~\ref{lem:5.19} and  \eqref{currloc}, we have
\begin{align}
\tr\big[ \phi [T_h(A) -V] \phi \big]_{-}  &= \tr\big[ \phi [T_h(A) -V]
 \phi \gamma_0\big] \non\\
&\geq
\tr\big[ \phi [T_h(A_r) -V]
 \phi \gamma_0\big] -C h^{1+\e_0} \Lambda_{L_0} \non\\
&\quad +
\int \tr \Big[ \big(\sigma(A-A_r)\sigma(D)+\sigma(D)\sigma(A -A_r) \big)
\eta_u \gamma_0 \eta_u\Big]\frac{\rd u}{L_1^3}.\non
\end{align}

We now apply Lemma~\ref{lem:5.21} and Theorem~\ref{thm:5.14} to get 
\begin{align}\label{eq:6.129}
\tr\big[ \phi [T_h(A) -V] \phi \big]_{-}  &\geq \tr\big[ \phi [T_h(A_r) -V] 
\phi \gamma_0\big] - C h^{1+\e_0} \Lambda_{L_0}
- C h^{4\e_0}  \int {\mathcal T}_u \frac{\rd u}{L_1^3} \non\\
  &\geq  \tr\big[ \phi [T_h(A_r) -V] \phi \big]_{-} - C h^{1+\e_0} \Lambda_{L_0}.
\end{align}
This finishes the proof of Lemma~\ref{lm:smoothA}.
\end{proof}

\begin{proof}[Proof of Lemma~\ref{lem:5.21}]
{F}rom now on we fix $u\in B(L_0+L_1)$ and drop the $u$ indices and superscripts for
simplicity, i.e, we set $\gamma =\gamma_0$ and $\eta=\eta_u$.
We rewrite the current, for each fixed $u$,
 by using \eqref{reality}, 
as
\begin{align}\label{eq:5122}
\tr \Big[ \big(\sigma(A-A_r)&\sigma(D)+\sigma(D)\sigma(A-A_r) \big) 
\eta \gamma \eta\Big] \non\\
& = \tr \Big[ \big(\sigma(A'-A_r')\sigma(D)+\sigma(D)\sigma(A'-A_r') \big)
\eta (\gamma -P)\eta\Big] \non\\
&= 2\tr[ F_0^2 D (A'-A_{r}')  \eta (\gamma-P) \eta] \nonumber\\
&\quad+
2\sum_{i<0} \tr[ F_i^2 D  (A'-A_{r}')  \eta (\gamma-1) \eta] - 
2\sum_{i<0} \tr[ F_i^2 D  (A'-A_{r}')   \eta (P-1) \eta] \nonumber \\
&\quad+
2\sum_{0<i\le i_0} \tr[ F_i^2 D  (A'-A_{r}')  \eta (\gamma-P) \eta]  \non\\
&\quad+ h\sum_{i\leq i_0}  \tr[ F_i^2 \sigma(B'-B_r')\eta (\gamma-P) \eta]\non\\
&\quad+  \tr \Big[ F_>^2\big(\sigma(A'-A_r')\sigma(D)+\sigma(D)\sigma(A'-A_r') \big)
\eta (\gamma-P) \eta\Big].
\end{align}
Note that we used  formula \eqref{alt1} for all terms with $i\le i_0$, while
we used \eqref{alt2} for $i>i_0$.
Notice also that the sums over negative indices can be restricted to
 $-i_0 \leq i <0$, since the $F_i$'s vanish for larger values of $|i|$.
Another important observation is that the left hand side of \eqref{eq:5122}
 is real, so it suffices to estimate the real part of each term.
The estimates will be very similar to the estimates
of the terms \eqref{tildesplit}, \eqref{DA'} and \eqref{secc}
obtained during the apriori estimates in Section \ref{sec:remA},
but $A$ will be replaced with  $A-A_r$ and we will capitalize on the
the factor $r^2\ll L^2$ gained from the smoothing  in \eqref{eq:AminusAl}
compared with the usual Poincar\'e inequality \eqref{eq:4}.
Now we treat each term in \eqref{eq:5122} separately.

\paragraph{Step 1: The $F_0$ term.}
We write
\begin{multline}
\tr[ F_0^2 D (A'-A_{r}') \eta (\gamma-P) \eta] \\= 
\tr[ F_0^2 D (A'-A'_{r}) \wt F_0 \eta (\gamma-P) \eta \wt F_0]  
+ \tr[ F_0^2 D (A'-A'_{r}) \wt G_0 \eta (\gamma-P) \eta \wt F_0] .
\label{FG1}
\end{multline}
Recall that $\wt F_0$ is slightly larger than $F_0$,
in particular $\wt F_0 F_0=F_0$  and $\wt F_0 + \wt G_0 = 1$
(see Section \ref{sec:decomp} for the precise definitions).
The first term is estimated as in \eqref{1st}, with $\omega_0 = \wt F_0  \eta^2 \wt F_0 $ 
and its density $\varrho_0$ as
\begin{align}
\big|  \tr[ F_0^2 D (A'-A'_{r}) \wt F_0 \eta (\gamma-P) \eta \wt F_0] \big| \leq
C \Big(\int  (A'-A'_{r})^2 \rho_0 \Big)^{1/2} \Big( \tr D^2 \omega_0 \Big)^{1/2}.
\end{align}
We also used that $W=V_u$ is bounded.
By the bounds on $\| \rho_0\|_{\infty}$ and
 $\tr D^2 \omega_0\le C\tr \om_0\le Cw\Lambda_{L_1}$ 
from \eqref{F0x} and \eqref{eq:AminusAl}, we have
\begin{align}
\big|  \tr[ F_0^2 D (A'-A'_{r}) \wt F_0 \eta (\gamma-P) \eta \wt F_0] \big| &\leq
C \Big( h^{-3} w^2 r^2 \Lambda_{L_1} \int_{B(2L_0)} |\nabla \otimes A| \Big)^{1/2} \nonumber \\
&\le  C h^{1+\rho -6\e_0}\Lambda_{L_1}
\end{align}
using \eqref{eq:MagneticApriori} and the choice of parameters
 \eqref{eq:parameters}, \eqref{deffw}.

For the other term in \eqref{FG1}, similarly to \eqref{2st}
and \eqref{Ddens}, we have
\begin{multline}\label{5128}
\big|  \tr[ F_0^2 D (A'-A'_{r}) \wt G_0 \eta (\gamma-P) \eta \wt F_0] \big| \\
\leq
C \tr[ D^2 \wt F_0 \eta^2 \wt F_0]^{1/2} \tr[  F_0^2 (A'-A'_{r})
 \wt G_0^2 (A'-A'_{r}) F_0^2]^{1/2}.
\end{multline}
We apply the estimate $\tr[ D^2 \wt F_0 \eta^2 \wt F_0] \leq Cw \Lambda_{L_1}$
 as before together with the first inequality of 
\eqref{appl} from the application of Lemma~\ref{lem:comm2}
 to get
\begin{align}\label{eq:Commutator5.22}
\big|  \tr[ F_0^2 D (A'-A'_{r}) \wt G_0 \eta (\gamma-P) \eta \wt F_0] \big| 
&\leq \Big(Cw \Lambda_{L_1} r h^{-2}\| \nabla \otimes A \|_{L^2(B(2L_0))}^2\Big)^{1/2}
\nonumber \\
&\le Ch^{1+\frac{1}{4}(2\rho -21\e_0)}\Lambda_{L_1}.
\end{align}
Here we used \eqref{eq:AminusAl} and that 
\begin{align}\label{5130}
\int |\nabla \otimes (A'-A_{r}')|^2 &\leq \int_{B(2L_0)}
  |\nabla \otimes (A-A_{r})|^2 + CL_0^{-2} \int_{B(3L_1/2)} |A-A_{r}|^2\nonumber \\
&\leq (1 + C (r/L_0)^2) \int_{B(2L_0)}  |\nabla \otimes A|^2,
\end{align}
to collect the $A$-terms.
In summary, we have proved 
\be
\Big| \tr[ F_0^2 (A'-A_{r}') \cdot D \eta (\gamma-P) \eta] \Big|
\le Ch^{1+\e_0}\Lambda_{L_1} \le Ch^{1+\e_0}\Lambda_{L_0}
\label{F0est}
\ee
assuming $\rho\ge 100\e_0$.

\paragraph{Step 2: $i\ge i_0+1$.}
\begin{lemma}\label{lem:5.22}
Assume that \eqref{eq:MagneticApriori} is satisfied and that $\rho \geq 100 \e_0$.
Then we have
\begin{align}
\label{eq:130}
\Big|\tr[ F_>^2 D (A- A_{r}) \eta \gamma \eta ]\Big| &\leq 
2 h^{1/4} \tr \big[ D^2 \wt F_> \eta \gamma \eta \wt F_>\big]  + C h^{5/4} \Lambda_{L_0}.
\end{align}
\end{lemma}

\begin{proof}
Recalling  $\sum_{i\ge i_0+1} F_i^2 = F_>^2$,  $F_>\wt F_>= F_>$ and $\wt F_>+ \wt G_>
=1$, we get
\begin{multline}
\label{eq:5.28}
\tr[ F_>^2 D (A- A_{r}) \eta \gamma \eta ] \\
= \tr[ F_>^2 D (A' - A'_{r}) \wt F_> \eta \gamma \eta \wt F_>] +
 \tr[ F_>^2 D (A'-A_{r}') \wt G_> \eta \gamma \eta \wt F_>].
\end{multline}
The first term we estimate as
\begin{align}\label{5131}
\Big|\tr[ F_>^2 D (A' - A_{r}') \wt F_> \eta \gamma \eta \wt F_>] \Big| \leq
\Big( \tr[ F_>^4 D^2 \wt F_> \eta \gamma \eta \wt F_> ] 
\tr[ (A' - A_{r}')^2 \wt F_> \eta \gamma \eta \wt F_>] \Big)^{1/2}.
\end{align}
In the last factor we apply H\"{o}lder, Sobolev and
 Lieb-Thirring inequalities similarly to \eqref{eq:LTforAminusAl} to get
\begin{align}
\tr[ (A' - A_{r}')^2 \wt F_> \eta \gamma \eta \wt F_>]  &\leq
\| A' - A_{r}' \|_2^{1/5} \| A' - A_{r}' \|_6^{9/5} \Big( h^{-2}
 \tr[ D^2 \wt F_> \eta \gamma \eta \wt F_>] \Big)^{3/5}\nonumber \\
&\leq r^{1/5} h^{-6/5} \int_{B(2L_1)} |\nabla \otimes A|^2 \,
 \Big( \tr[ D^2 \wt F_> \eta \gamma \eta \wt F_>] \Big)^{3/5}.
\end{align}
Inserting this into \eqref{5131}, using $\wt F_>^2 \leq 1$,
 we find
\begin{align}
\label{eq:134}
\Big|\tr[ F_>^2 D (A' &- A_{r}') \wt F_> \eta \gamma \eta \wt F_>]\Big| \leq
r^{1/10} h^{-3/5}\Big( \int_{B(2L_1)} |\nabla \otimes A|^2\Big)^{1/2}
 \Big( \tr[ D^2 \wt F_> \eta \gamma \eta \wt F_>] \Big)^{4/5} \nonumber \\
&= \tau^{5/4} \tr[ D^2 \wt F_> \eta \gamma \eta \wt F_>] 
+ \tau^{-5} r^{1/2} h^{-3} \Big( \int_{B(2L_1)} |\nabla \otimes A|^2\Big)^{5/2} \nonumber \\
&\leq  \tau^{5/4} \tr[ D^2 \wt F_> \eta \gamma \eta \wt F_>] 
+ \tau^{-5} \big( h^{5+\rho-24\e_0}\big)^{1/2} \Lambda_{L_0}
\end{align}
for any $\tau>0$, and where we used \eqref{eq:MagneticApriori} to get the last estimate.

We now consider the second term in \eqref{eq:5.28}. By a 
Cauchy-Schwarz inequality and by applying Lemma~\ref{lem:comm2} (recall \eqref{f>})
similarly as in \eqref{5128}--\eqref{5130}, we find
\begin{align}
\label{eq:135}
\Big| \tr[ F_>^2 D (A'-&A_{r}') \wt G_> \eta \gamma \eta \wt F_>] \Big|\leq
\Big(\tr[ D^2 \wt F_> (\eta \gamma \eta)^2 \wt F_>] \tr[ F_>^2 (A'-A_{r}')
 \wt G_>^2 (A'-A_{r}') F_>^2] \Big)^{1/2}\nonumber \\
&\leq \tau \tr[ D^2 \wt F_> \eta \gamma \eta \wt F_>] + \tau^{-1} h^{-2} r \int_{B(2L_1)}
 |\nabla \otimes A|^2 \nonumber \\
&\leq \tau \tr[ D^2 \wt F_> \eta \gamma \eta \wt F_>] 
 + \tau^{-1} h^{\frac{3}{2} + \rho - 3\e_0} \Lambda_{L_0},
\end{align}
where we used \eqref{eq:MagneticApriori} to get the last estimate.

Combining \eqref{eq:134} and \eqref{eq:135}, we get \eqref{eq:130} by choosing
 $\tau=h^{1/4}$.
\end{proof}

\paragraph{Step 3: $|i|\leq i_0+1$.}

\begin{lemma}\label{lem:5.23}
If $\rho > 100 \e_0$,
\begin{align}\label{eq:5.136}
&\Re \Big\{2\sum_{i<0} \tr[ F_i^2 D(A'-A_{r}') \eta (\gamma-1) \eta] - 
2\sum_{i<0} \tr[ F_i^2 D(A'-A_{r}')\eta (P-1) \eta] \nonumber \\
&\quad+
2\sum_{0<i\le i_0} \tr[ F_i^2 D(A'-A_{r}')  \eta (\gamma-P) \eta]  \non\\
&\quad+ h\sum_{i\leq i_0}  \tr[ F_i^2 \sigma(B'-B_r')\eta (\gamma-P) \eta] \Big\}\non \\
&\geq
- C h^{4 \e_0}\sum_{i \leq i_0}  2^i w  \tr[ D^2 F_i^2
 \eta \gamma \eta] \non\\
&- C h^{4 \e_0} \sum_{0> i \geq i_0} 2^{|i|} w  
\tr[ D^2 F_i^2 \eta (1-\gamma) \eta]
- C h^{1+\e_0} \Lambda_{L_0}.
\end{align}
\end{lemma}

\begin{proof}
Consider $0<i\le i_0$.
We write 
\begin{align}
\label{eq:5137}
 \tr[ F_i^2 D (A' - A_r')  \eta \gamma \eta]=
  \tr[ F_i^2 D (A' - A_r') \wt F_i \eta \gamma \eta] +
\tr[ F_i^2 D (A' - A_r') \wt G_i \eta \gamma \eta], 
\end{align}
using $\wt F_i + \wt G_i \equiv 1$.
We can estimate the first term, using $F_i \wt F_i = F_i$, as,
\begin{align}
\Big| \tr[ F_i^2 D (A' - A_r') \wt F_i \eta \gamma \eta] \Big| &\leq
\{ \tr[ F_i^4 D^2 \wt F_i \eta \gamma \eta \wt F_i] \}^{1/2}
 \{ \tr[ (A' - A_r')^2 \wt F_i \eta \gamma \eta \wt F_i]\}^{1/2}\non\\
&\leq h^{4 \e_0} 2^i w  \tr[ D^2 F_i^2
 \eta \gamma \eta]  + h^{-4\e_0} 2^{-i} w^{-1} \tr[ (A' - A_r')^2 \wt F_i^2] \non\\
 &\leq
 h^{4 \e_0} 2^i w  \tr[ D^2 F_i^2
 \eta \gamma \eta] + C h^{-4\e_0} h^{-3} r^2 \int_{B(2L_0)}|\nabla \otimes A|^2 \non\\
 &\leq
 h^{4 \e_0} 2^i w  \tr[ D^2 F_i^2
 \eta \gamma \eta] + C h^{1+2\rho-7\e_0} \Lambda_{L_0},
\end{align}
where we used \eqref{eq:supFi} and \eqref{eq:AminusAl} in the third inequality.
 Then we used \eqref{eq:MagneticApriori} and the choice of parameters 
\eqref{eq:parameters} to finish the estimate.
This is in agreement with the desired estimate if $\rho > 100\e_0$ 
(since the number of terms $i_0$ in the sum is only logarithmic in $h$).

The remaining term in \eqref{eq:5137} is estimated using Lemma~\ref{lem:comm2}
and that $\|\nabla f_i \|_1\le CW\le C$:
\begin{align}
\Big| \tr[ F_i D F_i(A' - A_r') \wt G_i \eta \gamma \eta]\Big|
&\leq 
\Big\{ \tr[ F_i^2 D^2 \eta \gamma \eta]\Big\}^{1/2} \Big\{ \tr[ \wt G_i (A' - A_r') 
 F_i^2 (A' - A_r') \wt G_i \eta \gamma \eta]\Big\}^{1/2} \non \\
&\leq h^{4\e_0} \cdot 2^i w \tr[ F_i^2 D^2 \eta \gamma \eta]
+C h^{-4\e_0} \frac{r h^{-2}}{2^i w} \| \nabla \otimes A \|_{L^2(B(2L_0))}^2 \non\\
&\leq h^{4\e_0} \cdot 2^i w \tr[ F_i^2 D^2 \eta \gamma \eta]
+ C h^{1+\rho - 6\e_0}  \Lambda_{L_0}.
\end{align}
This is in agreement with \eqref{eq:5.136} if $\rho > 100\e_0$.

We also estimate the corresponding term with $P$.
\begin{align}
 \Big| \tr[ F_i^2 D(A'-A_{r}') \eta P \eta] \Big| &\leq \Big\{
 \tr[ F_i^2 D^2 \eta P \eta ] \tr[ P \eta (A'-A_{r}') F_i^2 (A'-A_{r}')\eta P
 \Big\}^{1/2}\non\\
 & = \| D F_i \eta P \|_{HS} \Big\{ \tr[ P \eta (A'-A_{r}') F_i^2 (A'-A_{r}')\eta P
 \Big\}^{1/2}.
\end{align}
Using \eqref{kinsmall} and that $(h/wL) = h^{\e_0/2}$, the first factor
can be made smaller than an arbitrarily large power of
$h$, while the second one is bounded by
$\Lambda_{L_1}^{1/2} \| A'-A'_r\|_2$, which is not bigger than
a fixed positive power of $h$. So this term is negligible.
The similar terms for negative indices $i$ are estimated 
in the same manner with only notational changes.

Also the $\sigma(B'-B_r')$ terms are readily controlled. We leave this part to the reader.
\end{proof}

We can now finish the proof of Lemma~\ref{lem:5.21}. We combine
\eqref{eq:5122} with \eqref{F0est} and the results of Lemma~\ref{lem:5.22} and~\ref{lem:5.23}.
This finishes the proof of Lemma~\ref{lem:5.21}.
\end{proof}

\end{document}